\RequirePackage{fix-cm}
\documentclass[smallextended]{svjour3}       
\smartqed  
\usepackage{graphicx}
 \usepackage{mathptmx}      
\usepackage{amsmath,amssymb,amsfonts}
\usepackage{caption}
\usepackage{yfonts}
\usepackage[dvipsnames]{xcolor}
\usepackage{cite}
\usepackage{enumitem}

\usepackage{epstopdf}
\usepackage{booktabs}

\def\operator#1{\mathfrak{O}[#1]}

\def\dotoverparen#1{\mathop{\vbox{\ialign{##\crcr\hfill$\cdot $\hfill\crcr
\noalign{\kern-2.3ex}
\downparenfill\crcr\noalign{\kern0.4ex\nointerlineskip}
$\hfil\displaystyle{#1}\hfil$\crcr}}}\limits}
\catcode`\@=12

\newcommand{\RR}{\mathbb R}
\newcommand{\CC}{\mathbb C}
\newcommand{\ZZ}{\mathbb Z}
\newcommand{\NN}{\mathbb N}

\newcommand{\C}{\mathcal{C}}
\newcommand{\F}{\mathcal{F}}
\newcommand{\Q}{\mathcal{Q}}

\newcommand{\cZ}{\mathcal{Z}}

\newcommand{\N}{\mathcal{N}}
\newcommand{\R}{\mathcal{R}}

\newcommand{\cS}{\mathcal{S}}
\newcommand{\cB}{\mathcal{B}}
\newcommand{\cL}{\mathcal{L}}

\def\boundr{\mathbf{r}}
\def\bounde{\mathbf{e}}
\def\boundf{\mathbf{f}}
\def\boundq{\mathbf{q}}

\def\boundx{\mathbf{x}}
\def\boundz{\mathbf{z}}
\def\boundzeta{\mbox{$\pmb\zeta$}}
\def\boundxi{\mbox{$\pmb\xi$}}

\def\im{{\rm i}}
\newcommand{\blkdiag}{{\rm blkdiag}}
\def\init{{\rm init}}

\def\ypsilon{{\mathchoice%
{{\mbox{$\scriptstyle \mathfrak{M} $}}}%
{{\mbox{$\scriptstyle \mathfrak{M} $}}}%
{{\mbox{$\scriptscriptstyle \mathfrak{M} $}}}%
{{\mbox{\tiny$\scriptscriptstyle \mathfrak{M} $}}}%
}}

\newcommand{\overbar}[1]{\mkern 1.5mu\overline{\mkern-1.5mu#1\mkern-1.5mu}\mkern 1.5mu}

\usepackage{ulem}
\DeclareMathAlphabet\EuScript{U}{eus}{m}{n}
\DeclareMathAlphabet\EuScriptB{U}{eus}{b}{n}

\DeclareMathAlphabet{\mathcursive}{U}{esstixcal}{m}{n}

\catcode`\@=11
\def\downparenfill{$\m@th\braceld\leaders\vrule\hfill\bracerd$}
\def\overparen#1{\mathop{\vbox{\ialign{##\crcr\crcr \noalign{\kern0.4ex}
\downparenfill\crcr\noalign{\kern0.4ex\nointerlineskip}
$\hfil\displaystyle{#1}\hfil$\crcr}}}\limits}
\catcode`\@=12

\catcode`\@=11
\def\h@uteurmax{\dp \strutbox}
\def\noteenmarge#1{\strut\vadjust{\kern-\h@uteurmax\textedenote{#1}}}
\def\textedenote#1{\vtop to \h@uteurmax{\baselineskip\h@uteurmax\vss\llap{#1}\null}}
\catcode`\@=12

\definecolor{modifcolor}{rgb}{1,0,0}

 \journalname{Mathematics of Control, Signals, and Systems}
\begin{document}

\title{Nonlinear Robust Periodic Output Regulation of
  Minimum Phase Systems
}


\author{Daniele Astolfi        \and
        Laurent Praly \and Lorenzo Marconi 
}


\institute{D. Astolfi  \at
              Univ Lyon, Universit\'e Claude Bernard Lyon 1, CNRS, LAGEPP UMR 5007, 43 boulevard du 11 Novembre 1918, F-69100, Villeurbanne, France \\
              \email{daniele.astolfi@univ-lyon1.fr}           
           \and
        L. Praly \at
      MINES ParisTech, PSL Research University, CAS - Centre automatique et syst\'emes, Paris 75006, France
         \and
         L. Marconi \at 
       CASY-DEI,  University of Bologna, Italy 
}

\date{Received: \today / Accepted: date}

\maketitle

\begin{abstract}
In linear systems theory it's a well known fact that a regulator  given by the cascade of an oscillatory dynamics, driven by some regulated variables, and of a stabiliser stabilising the cascade of the plant and of the oscillators has the ability of  blocking on the steady state of the regulated variables any harmonics matched with the ones of the oscillators.  This is the well-celebrated internal model principle. In this paper we are interested to follow the same design route for a controlled plant that is a nonlinear and periodic system  with period $T$: we add a bunch of linear oscillators, embedding  $n_o$ harmonics that are multiple of $2 \pi/T$, driven by a ``regulated variable'' of the nonlinear system, we  look for a stabiliser for the nonlinear cascade of the plant and the oscillators, and we study the asymptotic properties of the resulting closed-loop regulated variable.  In this framework the contributions of the paper are multiple: for  specific class of minimum-phase systems we present a systematic way of  designing a stabiliser, which is uniform with respect to $n_o$,  by using a mix of high-gain and forwarding techniques; we prove that the resulting closed-loop system has a periodic steady state with period $T$ with a domain of attraction not shrinking with $n_o$; similarly to the linear case, we also show that the spectrum of the steady state closed-loop regulated variable does not contain the $n$ harmonics embedded in the bunch of oscillators and that the $L_2$ norm of the regulated variable is a monotonically decreasing function of $n_o$. The results are robust, namely the asymptotic properties on the regulated variable hold also in presence of any uncertainties in the controlled plant not destroying closed-loop stability.    

 \keywords{Nonlinear output regulation \and Repetitive Control \and Minimum Phase Systems \and Harmonic rejection}
\end{abstract}

\section{Introduction}
\label{intro}

The problem of  rejecting or tracking asymptotically periodic or quasi-periodic 
 signals
is of primary importance in many applications \cite{kurniawan2014survey,wang2009survey,longman2000iterative}.
Among them,  
robotics \cite{kasac2008passive,omata1987nonlinear}, 
power electronics \cite{mattavelli2004repetitive} and
bio-medics engineering \cite{gentili2008disturbance}, 
 just to cite a few. Such a  problem is commonly known in 
control system theory as 
\textit{robust output regulation}, see
\cite{davison1976robust,francis1976internal,byrnes1997output}, where the 
adjective robust refers to the fact that 
the  asymptotic properties are desired to hold not only for the nominal
model of the system but also for small perturbations of it. 
The solution to the robust output regulation problem for 
finite-dimensional linear time-invariant systems  is accredited 
to Francis, Wonham and Davison who at the same time,
but independently, published their main works
 during the 70's, see, e.g., \cite{davison1976robust,francis1976internal}.
The proposed solution relies on the so-called
\textit{internal model principle} coined by Francis and Wonham
in their celebrated work \cite{francis1976internal}, 
stating that output regulation
property is insensitive to plant parameter variations
 ``only if the controller utilises feedback
of the regulated variable, and incorporates in the feedback path a suitably reduplicated
model of the dynamic structure of the exogenous signals which the regulator is
required to process''.
In turns, if some overall stability properties are guaranteed, 
the presence of a copy of the exogenous dynamics 
(also denoted as exosystem) 
in the regulator provides a ``blocking-zero'' effect on 
the desired regulated output at the dynamics excited by 
such exogenous signals.
In other words, the regulated output cannot 
contain any mode of the exosystem if the overall trajectories 
are bounded.
In practice, an integral action in the controller
allows one to achieve zero-DC value of the regulated output, 
while a given oscillator at a certain frequency, 
ensures to have zero spectral component
at it \cite{astolfi2015approximate,ghosh2000nonlinear}.

In our preliminary contribution 
\cite{astolfi2015approximate}, focused on the problem of 
nonlinear robust output regulation in presence of 
periodic exosignals, we have proposed a solution 
for input-affine nonlinear systems 
that mimics the aforementioned linear paradigm \cite{davison1976robust,francis1976internal}.
It involves the following two main components.
\begin{enumerate}
\item
An internal-model unit processing the regulated
output which is  composed
by a bunch of oscillators at a given fundamental frequency  and
a certain number $n_o$, possibly infinite, of its multiples.
Such an internal model unit
guarantees a blocking-zero property on the regulated error 
in terms of spectral components, i.e. 
the regulated output cannot have harmonics at the frequencies 
embedded in the internal model unit, see
\cite[Proposition 3]{astolfi2015approximate} or 
\cite[Proposition 1]{ghosh2000nonlinear}.
Such a property is insensitive to model perturbations
as long as system trajectories are bounded. In doing so,
(contrary to large 
part of nonlinear output regulation theory, see, e.g.,
\cite{byrnes2003limit,byrnes2004nonlinear,marconi2007output})
a  precise description of the generator of the exosignals is not 
needed as long as the fundamental period
characterising all the spectrum is known. This motivates us
in describing the plant dynamics as a $T$-periodic time-varying 
nonlinear system, without need of a precise description of the 
exogenous signals dynamics.

\item
A state feedback stabiliser composed by two parts:
a preliminary state feedback in charge of stabilising
the equilibrium assumed to exist when the exosignals 
are zero, augmented by a
forwarding feedback which serves at stabilising
the overall extended dynamics composed by plant
and the internal-model unit
(see, e.g., \cite{astolfi2015approximate,astolfi2017integral,
astolfi2019francis}).
\end{enumerate}

The motivation for such a solution comes from the well-known fact that 
an input-to-state stable system driven by a periodic input
admits (at least locally) periodic solutions of the same period
\cite[\S 12]{sontag2008input}.
In our preliminary work \cite{astolfi2015approximate}, however, 
some important issues were still open. In particular, 
it was not clear whether the domain of attraction 
of the steady-state periodic solution shrinks to zero by increasing
the dimension of the internal-model unit (i.e. the number of oscillators), 
and whether asymptotic regulation can be achieved
by means of a (countable) infinite-dimensional internal-model 
(i.e. by using an infinite number of oscillators).
We gave a partial answer to this latter question 
in our second 
preliminary contribution \cite{astolfi2019francis}.

The objective of this work is therefore to  
give an exhaustive answer to both open questions
by providing a unifying result and by showing the 
practical interest of the proposed approach in  periodic output regulation frameworks.
For this,
 as in  \cite{astolfi2019francis},
 we restrict our attention to the particular class of
minimum-phase nonlinear systems, 
that is, systems possessing
a well-defined relative degree with constant 
high-frequency gain, which are described in 
 normal form (possibly after a change of coordinates), 
see \cite[Chapter 4]{isidori1995nonlinear}, and with locally exponentially 
stable zero-dynamics. This allows us to choose
an elementary 
\textit{high-gain} feedback as
a  preliminary stabilising feedback
(see, e.g., 
\cite[Lemma 2.2]{teel1995tools}) and a linear
\textit{forwarding} feedback. In this simplified context,
it is shown that 
the behaviour of the proposed regulator is
\textit{robust} with respect to model uncertainties and 
 \textit{uniform} 
in the dimension of the internal model unit (i.e., the number 
$n_o$ of oscillators), in the following sense:

\begin{itemize}
\item[$\bullet$] The high-gain feedback doesn't need to be re-parametrized
if the number of oscillators vary. It is chosen 
before hand, based on the Lipschitz system properties
(i.e., the precise knowledge of the plant's dynamics is not needed).

\item[$\bullet$] The proposed regulator ensures boundedness of the overall closed-loop 
system trajectories and the existence of  a (locally) 
exponentially stable $T$-periodic solution, with a domain of 
attraction which is uniform in the number of oscillators $n_o$, and independent
of their frequencies.
In other words, a precise knowledge of the period $T$
characterising the plant's dynamics 
(that is, the period of the exosignals)
 is not needed to ensure such a  boundedness property.
 
 \item[$\bullet$] If the fundamental frequency of the internal-model unit 
 oscillators
 is selected exactly as the one characterising 
 the plant's dynamics  $\tfrac{2\pi}T$, then the regulated output cannot
 contain harmonics at the frequencies of such oscillators, and 
 approximate regulation is achieved\footnote{Note that the notion of approximate
  considered here is different from the 
  notion of    $k$-th order approximate regulation 
   defined in  \cite{huang1994robust}, 
\cite[Chapter 2.5]{byrnes1997output}.}.
 Furthermore, the $L_2$ norm (over a period) of the steady-state output 
 is inversely proportional to the number of oscillators: 
 the quality of approximate regulation can be therefore improved
 by augmenting the number of oscillators (and not by increasing 
 the parameters of the high-gain feedback regulator as in 
 high-gain feedback control).
 
 \item[$\bullet$] If the number of oscillators is infinite, then
  asymptotic\footnote{In this case, 
exponential stability cannot be anymore guaranteed 
in view of the presence of an infinite number
of poles on the imaginary axis and the use of a bounded
(in the sense of \cite[Page 24]{tucsnak2009observation}) control operator. See also 
\cite{paunonen2010internal,paunonen2017robust,astolfi2021repetitive} for
other examples of such a phenomenon.}  output regulation is achieved. This 
infinite dimensional regulator
 preserves the same bounds in terms of high-gain feedback and 
 domain of attraction.
\end{itemize}
In conclusion, we show that the problem of periodic robust output 
regulation can be generically solved  by an
infinite-dimensional internal-model based regulator 
containing oscillators at all multiples of the basic periodic.
Moreover, a good  finite-dimensional approximation 
of such a regulator can be obtained by 
using a finite number of oscillators.
In doing so, one obtains 
a desired accuracy (in terms of $L^2$ norm) with 
bounds of the high-gain term and the domain of attraction 
which are uniform (i.e., independent) in the desired 
approximation (i.e., the number of oscillators).

 From a technical point of view, the main difficulty to deal with 
 is the fact that each time we modify the dimension of the internal 
 model unit (i.e. we vary the number of oscillators), we have 
 to study a system with different state-space dimension. 
While most of mathematical tools are well-suited 
 to study the effect of parameter's variations in system dynamics
 (see, e.g., \cite{khalil2002nonlinear} or 
  \cite[Appendix]{astolfi2017integral}), 
 it is very hard to compare objects with different dimensions
 and we are unaware of generic tools developed 
 for this specific purpose. 
 For this reasons, in this work we analyse all the aforementioned
 properties 
(all these contributions are new with respect to 
our preliminary works \cite{astolfi2015approximate,astolfi2019francis}) 
 by carefully re-doing all the proofs concerning
 existence of periodic solutions (which mainly relies on 
 fixed-point theorems) and their stability properties
 (which mainly relies on Lyapunov analysis), by  showing 
 that all these features are uniform in the internal-model system
 dimension. 
In such a perspective, the fact of focusing on 
minimum-phase systems with constant 
high-frequency gain and unitary relative degree
allows us to conceptually simplify
most of the (already complex) proofs.  The case of 
higher relative degree can be easily dealt with 
by means of partial change of coordinates
as show in \cite{serrani2001semi}. Details are 
given in Section~\ref{sec:higher_degree}.
The case of square multi-input multi-output 
(i.e. same number of inputs and outputs) 
systems with constant (and invertible) high-frequency
matrix gain is also straightforward and not considered in this work.

This  rest of the article is organised as follow.
In Section ~\ref{sec:prob} we state the problem formulation and, in Section ~\ref{sec:main}, we provide the main results of this work.With these precise elements at hand, we are in a better position to compare our results with what is available in the literature. This is done in Section ~\ref{sec_LR}. A numerical example is proposed in 
Section~\ref{sec:example}.
Conclusions are drawn in Section~\ref{sec:conclusions}.
All proofs are postponed in the Appendix.

\paragraph*{Notation.}
\label{sec:notation}
$\;$ $\RR$ is the set of real numbers and $\RR_{\geq0}:=[0,+\infty)$;
$\ZZ$ is the set of integers; 
$\NN$ is the set of non-negative integers, 
and $\NN_{>0}$ is the set of positive integers,
 $\CC$ is the set of complex
number and $\im= \sqrt{-1}$. Given $x\in\CC$, we denote 
with $\bar x$ its conjugate.
We denote by $\C^k(X;Y)$ the set of $C^k$ functions from $X$ to 
$Y$, and with 
$\C^k_T([0,T];X)$ the set of $C^k$ $T$-periodic functions from $[0,T]$ to $X$. For compactness, in the following, $\C^k_T([0,T];X)$ 
will be simply denoted as $\C_T^k(X)$.

\section{Problem Statement}
\label{sec:prob}

The objective of this work is to state and prove very precise results concerning the output regulation problem for the particular systems  that can be rewritten, under suitable change
of coordinates, in the form
\begin{equation*}\label{eq:sys0}
\begin{array}{rcl}
\dot x & =& f(w,x,e)
\\
\dot e & =& q(w,x,e) + u
\end{array}
\end{equation*}
where $(x,e)\in\RR^n\times \RR$ is the state,
with $x$ unmeasured
 $u\in\RR$ is the control input, 
 $e\in\RR$ is the
measured output
 to be regulated 
to zero
and $w\in\RR^{n_x}$ are exogenous signals
representing references to be tracked
or disturbances to be rejected. The number of 
control problems that can be recast in such a form 
is very large and examples may be found, for instance, 
in \cite{mahmoud1996asymptotic,khalil2000design,byrnes2003limit,byrnes2004nonlinear,
marconi2007output,astolfi2021repetitive}
and references therein.
As discussed in \cite[Remark 1]{byrnes2004nonlinear},
the considered class of systems may look very particular,
  as it has relative degree
1 between control input u and regulated output e.
However, the 
 design methodology described in 
what follows lends itself to a 
straightforward extension to systems with higher 
relative degree \cite[eq. (33)]{byrnes2003limit}. 
More details are postponed to Section~\ref{sec:higher_degree}.
In this work, 
we suppose that $w$ is $T$-periodic,
in other words $w$ is a sufficiently smooth function fulfilling $w(t+T)= w(t)$. 
To simplify our notations,
throughout the rest of this paper
we will replace $w$ by $t$ and we
will assume that the functions $f,q$ satisfies
$$
f(t+T,x,e) = f(t,x,e), 
\quad
q(t+T,x,e) = q(t,x,e)
$$
for any $t$. 
We are interested, moreover, in systems which are strongly minimum-phase.
In particular, we suppose that when $e=0$, the system 
$$
\dot x = f(t,x,0)
$$
admits a unique periodic solution $x_0(\cdot)$ which is exponentially 
stable with 
some
 domain of attraction.
Since, in what follows, the knowledge of $x_0$ is not required and $x$ is not accessible,
there is no loss of generality in assuming that $x_0(t)$ is the origin of the coordinates for $x$ at time $t$.
This allows us to
formulate our output 
regulation problem for the following class of systems
\begin{subequations}\label{eq:sys}
\begin{align}
\label{eq:sys_zero_dyn}
\dot x & = f(t,x,e)
\\
\dot e & = q(t,x,e) + u \label{eq:sys_e}
\end{align}
\end{subequations}
where $f:\RR\times \RR^n\times \RR\to \RR^n$
and $q:\RR\times \RR^n\times \RR\to \RR$  are $C^2$ and $T$-periodic in their first argument,  $f$ is such that $f(t,0,0)= 0$ for all $t\geq0$ and 
the function $q$  satisfies $\sup_{t\in[0,T]}|q(t,0,0)|> 0$ (otherwise 
the problem would trivially boil down to a stabilisation 
context that would be solved, 
in a semi-global context,
 by a simple high-gain 
feedback $u=-\sigma e$, with $\sigma>0$ sufficiently large, 
see, e.g., \cite{teel1995tools}).

Our approximate output regulation objective is
\begin{equation}
\label{eq:objective}
\limsup_{t\to \infty } |e(t)|\; \leq \; \bounde_p
\end{equation}
with $\bounde_p$ arbitrarily chosen.

Evidently, a simple way to achieve such a practical regulation property
could be that of implementing a high-gain 
 controller of the form $u=-\sigma e$, 
with  $\sigma>0$ large enough, see, e.g., \cite[Example 2.1]{teel1995tools}.
The drawback of this controller is that $\sigma $ is the 
only tunable parameter and it has the
undesirable property of amplification
of possible (high-frequency) measurement noise, 
thus being  unsuited in practical applications.
The 
regulator we propose incorporates also a tunable internal model allowing
to satisfy
\eqref{eq:objective} but also stronger properties
that a standard high-gain controller
could not achieve. 
Also, as explained in the introduction, we want a result as less dependent on $f$ and $q$ as possible.
The aim of this weak dependency is to make the property 
\eqref{eq:objective} robust to uncertainties in $f$ and $q$.  This is achieved by asking not the exact 
knowledge of the pair $(f,q)$ but only that it belongs to a family. Precisely,
say that we have a model pair $(f_m,q_m)$. We define from it
the set of bounding functions 
$$
 \begin{array}{rclrcl}
\displaystyle 
\sup_{(t,x,e)\in \cS_T(\boundx,\bounde)}
\left|
f_m(t,x,e)
\right|
 &\leq &\boundf (\boundx,\bounde),
 &
 \displaystyle 
\sup_{(t,x,e)\in \cS_T(\boundx,\bounde)}
\left|
\frac{\partial f_m}{\partial e} (t,x,e)
\right|
&\leq &
\boundf_{e} (\boundx,\bounde)
,
\\
\displaystyle
\sup_{(t,x,e)\in \cS_T(\boundx,\bounde)}
\left|
\frac{\partial ^2f_m}{\partial x\partial x} (t,x,e)
\right|
&\leq & \boundf _{xx} (\boundx,\bounde),
&
 \displaystyle 
\sup_{(t,x,e)\in \cS_T(\boundx,\bounde)}
\left|
\frac{\partial ^2f_m}{\partial e\partial x} (t,x,e)
\right|
& \leq &
\boundf _{ex} (\boundx,\bounde),
\\
\displaystyle 
\sup_{(t,x,e)\in \cS_T(\boundx,\bounde)}
\left|
q_m(t,x,e)
\right|
 &\leq &\boundq (\boundx,\bounde),
 &
 \displaystyle 
\sup_{(t,x,e)\in \cS_T(\boundx,\bounde)}
\left|
\frac{\partial q_m}{\partial e} (t,x,e)
\right|
&\leq &
\boundq_{e} (\boundx,\bounde)
,
\\
\displaystyle
\sup_{(t,x,e)\in \cS_T(\boundx,\bounde)}
\left|
\frac{\partial q_m}{\partial x} (t,x,e)
\right|
&\leq & \boundq _{x} (\boundx,\bounde),
&
 \displaystyle 
\sup_{(t,x,e)\in \cS_T(\boundx,\bounde)}
\left|
\frac{\partial q_m}{\partial t} (t,x,e)
\right|
& \leq &
\boundq _{t} (\boundx,\bounde).
 \end{array}
 $$
where we have introduced the sets
$$
\begin{array}{rclrcl}
\cS_x(\boundx)  &:=& \{ x  : 
|x|\leq \boundx\},
& 
\cS_e(\bounde) &:= &\{ x  : 
|e|\leq \bounde\},
\\
\cS(\boundx,\bounde) & :=& \cS_x(\boundx)\times\cS_e(\bounde),
\qquad
& 
\cS_T(\boundx,\bounde)&:=&[0,T]\times \cS_x(\boundx)\times\cS_e(\bounde),
\end{array}
$$
Then we forget $(f_m,q_m)$ and instead
consider the families of functions $\F$
and $\Q$ defined as follows. 

\begin{definition}[Family $\F$]\label{def:familyF}\itshape
 Given a triplet
of positive numbers 
 $\overline P_x,\underline P_x,\alpha>0$
and a pair of positive numbers $\boundx,\bounde>0$, 
we say that the function  $f:\RR_{\geq0}\times \RR^n\times\RR\mapsto \RR^n$ belongs to the 
family 
$\F(\overline P_x,\underline P_x,\alpha,\boundx,\bounde)$
if the following statements holds.
 \begin{itemize}
 \item[$\bullet$] 
The function $f$  is $C^2$, 
$T$-periodic in the
first
  argument and satisfies the following 
 set of inequalities
 $$
\hskip -\leftmargin
 \begin{array}{rclrcl}
\displaystyle 
\sup_{(t,x,e)\in \cS_T(\boundx,\bounde)}
\left|
f(t,x,e)
\right|
 &\leq &\boundf (\boundx,\bounde),
 &
 \displaystyle 
\sup_{(t,x,e)\in \cS_T(\boundx,\bounde)}
\left|
\frac{\partial f}{\partial e} (t,x,e)
\right|
&\leq &
\boundf_{e} (\boundx,\bounde)
,
\\
\displaystyle
\sup_{(t,x,e)\in \cS_T(\boundx,\bounde)}
\left|
\frac{\partial ^2f}{\partial x\partial x} (t,x,e)
\right|
&\leq & \boundf _{xx} (\boundx,\bounde),
&
 \displaystyle 
\sup_{(t,x,e)\in \cS_T(\boundx,\bounde)}
\left|
\frac{\partial ^2f}{\partial e\partial x} (t,x,e)
\right|
& \leq &
\boundf _{ex} (\boundx,\bounde).
 \end{array}
$$
\vskip -2em\refstepcounter{equation}\label{LP1}\null \hfill$(\theequation)$\vskip -2em
$$
\hskip -\leftmargin
\begin{array}{rclrcl}
\displaystyle 
\sup_{(t,x,e)\in \cS_T(\boundx,\bounde)}
\left|
f(t,x,e)
\right|
 &\leq &\boundf (\boundx,\bounde),
 &
 \displaystyle 
\sup_{(t,x,e)\in \cS_T(\boundx,\bounde)}
\left|
\frac{\partial f}{\partial e} (t,x,e)
\right|
&\leq &
\boundf_{e} (\boundx,\bounde)
,
\\
\displaystyle
\sup_{(t,x,e)\in \cS_T(\boundx,\bounde)}
\left|
\frac{\partial ^2f}{\partial x\partial x} (t,x,e)
\right|
&\leq & \boundf _{xx} (\boundx,\bounde),
&
 \displaystyle 
\sup_{(t,x,e)\in \cS_T(\boundx,\bounde)}
\left|
\frac{\partial ^2f}{\partial e\partial x} (t,x,e)
\right|
& \leq &
\boundf _{ex} (\boundx,\bounde).
 \end{array}\quad \null 
$$
 \item[$\bullet$]
  There exists a $T$-periodic $C^1$ positive definite matrix
   $P_x:\RR_{\geq0}\to\RR^{n\times n}$,  satisfying
  \begin{eqnarray}
&\displaystyle 0  \; < \; \underline P_x I \; \leq \;  P_x(t) \leq 
\;  \overline P_x I,
\label{eq:Pass1}
\\[.5em]
\label{eq:Pass2}
&\displaystyle 
\dot P_x(t) + P_x(t) \frac{\partial f}{\partial x}(t,0,0)+ 
 \frac{\partial f^\top}{\partial x}(t,0,0) P_x(t) \; \leq  \;  - 2\alpha P_x.
  \end{eqnarray}
 \end{itemize}
\end{definition}

\begin{definition}[Family $\Q$]
\label{def:familyQ}\itshape
Given a pair of positive numbers $\boundx,\bounde>0$, 
we say that the function  $q:\RR_{\geq0}\times \RR^n\times\RR\mapsto \RR^n$ belongs to the 
family 
$\Q(\boundx,\bounde)$
if it
 is  $C^2$, $T$-periodic in its first argument and satisfies the following 
 set of inequalities
\begin{equation}
\label{LP2}
 \begin{array}{rclrcl}
\displaystyle 
\sup_{(t,x,e)\in \cS_T(\boundx,\bounde)}
\left|
q(t,x,e)
\right|
 &\leq &\boundq (\boundx,\bounde),
 &
 \displaystyle 
\sup_{(t,x,e)\in \cS_T(\boundx,\bounde)}
\left|
\frac{\partial q}{\partial e} (t,x,e)
\right|
&\leq &
\boundq_{e} (\boundx,\bounde)
,
\\
\displaystyle
\sup_{(t,x,e)\in \cS_T(\boundx,\bounde)}
\left|
\frac{\partial q}{\partial x} (t,x,e)
\right|
&\leq & \boundq _{x} (\boundx,\bounde),
&
 \displaystyle 
\sup_{(t,x,e)\in \cS_T(\boundx,\bounde)}
\left|
\frac{\partial q}{\partial t} (t,x,e)
\right|
& \leq &
\boundq _{t} (\boundx,\bounde).
 \end{array}
\end{equation}
\end{definition}

Note that systems of the form 
\eqref{eq:sys} satisfying $f,q\in \F,\Q$
are typically obtained 
when deriving a normal form
in presence of smooth periodic reference to be tracked 
or perturbation to be rejected.
See, for instance, 
\cite{mahmoud1996asymptotic,khalil2000design,astolfi2021repetitive}
and references therein.

\noindent
\textit{Remark}\;
Note that in light of \eqref{eq:Pass1}, \eqref{eq:Pass2}, the set
$\F(\overline P_x,\underline P_x,\alpha, \boundx,\bounde)$ characterises
functions for which the zero dynamics 
of \eqref{eq:sys}, namely 
 $\dot x = f(t,x,0)$, is locally exponentially
stable, with a given 
decreasing rate $-\alpha$. Indeed, 
the function $V(t,x) = x^\top P_x(t) x$ 
can be used as Lyapunov function to establish such 
 stability properties.
Since 
in this work we are not interested in establishing 
semi-global results, the properties 
\eqref{eq:Pass1} and \eqref{eq:Pass2} will be the only assumptions
made on the zero-dynamics \eqref{eq:sys_zero_dyn}.
Such assumptions are indeed milder than those commonly stated  in semi-global 
output regulation results, 
where typically the zero-dynamics
\eqref{eq:sys_zero_dyn} is 
asked to be input-to-state stable (in short, ISS)
or integral ISS (in short, iISS)
with respect to $e$, 
see, e.g., 
\cite{serrani2001semi,khalil2000design,byrnes2004nonlinear,xu2013global}.

\section{Main Results}
\label{sec:main}

\subsection{Internal-Model Based Regulator Design}\label{sec:controller}
The feedback law we propose is made of two sets of
tunable
parameters:
\begin{itemize}
\item[$\bullet$] 
an integer $n_o\in \NN$ and two  positive real numbers 
$\sigma,\mu>0$;
\item[$\bullet$] 
two sequences of positive real numbers $n_{z\ell }$ and $\omega _\ell $ satisfying
\begin{subequations}\label{eq:conditions_series}
\begin{align} \label{eq:sequence_nzk}
\sum_{\ell =0}^\infty n_{z\ell } 
\;=\;
\overline{N}_z^2  & \; <  \; +\infty \ ,
\\\label{27}
n_{z(\ell+1)} & \;< \; n_{z\ell} & \forall \; 0 \leq \ell  \ ,
\\ \label{26}
\ell\,  n_{z\ell} &\;\leq \; m\,  n_{zm}  &\forall (\ell,m):\,  0 < m \leq  \ell \ ,
\\ \label{23}
\ell^2\,  n_{z\ell} &\; \leq\;  m^2\,  n_{zm} & \forall (\ell,m):\,  0 \leq  \ell \leq  m \ ,
\\
\label{eq:basic_omega}
\omega_\ell & \; = \; \ell \widehat \omega & \forall\ell >0 \ ,
\end{align}
\end{subequations}
for some $\widehat \omega>0$.
Actually,
as an illustration or for more specificity, we consider often the particular case
\begin{equation}
\label{eq:nzk}
\begin{array}{rcl}
n_{z0}&=&2, 
\\
n_{z\ell } &= &\dfrac{1}{\ell ^{1+\varepsilon}}, 
\quad \forall \ \ell \in \NN_{>0}
\quad
\varepsilon\in (0,1].
\end{array}
\end{equation}
\end{itemize}
The proposed dynamic controller takes
the form
\begin{subequations}
\label{eq:controller}
\begin{align}
\label{eq:imu_feedback}
\dot z &=  \Phi z +\Gamma  e
\\
u&= -\sigma e + \mu M^\top N_z (z-Me)  \ ,
\label{eq:hg_feedback}
\end{align}
\end{subequations}
where $z = (z_0, \ldots,z_{n_o})\in\RR^{2n_o+1}$,  is the state
of the controller  and 
the matrices $\Phi, N_z\in \RR^{(2n_o+1) \times (2n_o+1)}$
and $\Gamma,M\in \RR^{2n_o+1}$ are defined as 
\begin{equation}
\label{eq:defPhiM}
\renewcommand{\arraystretch}{1.7}
\begin{array}{rclrcl}
\Phi& := &\blkdiag  \big(0, \Phi_1, \ldots, \Phi_{n_o} \big), 
\quad &
\Phi_\ell  &=& 
\renewcommand*{\arraystretch}{.8}
\begin{pmatrix}
0 & \omega_\ell 
\\[.5em]
-\omega_\ell  & 0
\end{pmatrix},
\\[.5em]
N_z & = &
 \blkdiag \big(n_{z0}, N_{z1}, \ldots, N_{zn_o}\big),
\quad
&
N_{z\ell }&=&
 n_{z\ell}
\renewcommand*{\arraystretch}{.8}\;
I_2
\quad \forall \ \ell  = 1, \ldots, n_o, 
\\[.5em]
M & = & (1, M_1^\top, \ldots, M_{n_o}^\top)^\top ,
\quad
&M_\ell  &=&  (1, 0)^\top  \quad \forall \ \ell  = 1, \ldots, n_o,
\\[.5em]
\Gamma & = &-(\Phi + \sigma I) M\ .
\end{array}
\end{equation}
It can be readily seen from the definition 
of \eqref{eq:defPhiM}, that the regulator 
\eqref{eq:controller} is composed of two parts:
an internal-model unit \eqref{eq:imu_feedback}, 
that is, the $z$-dynamics, characterised by an integrator
and a bunch of linear oscillators at frequencies 
$\omega_\ell$, and a linear stabilizing term
\eqref{eq:hg_feedback} having embedded the  high-gain feedback law $-\sigma e$ needed for 
stability purposes 
(see, e.g., 
\cite{teel1995tools,marconi2007output}). 
As shown in the rest of the paper,  the
feedback law  
\eqref{eq:controller}
guarantees that \eqref{eq:objective} holds with 
any (arbitrarily small) 
$\bounde_p$ 
chosen a priori and independent of 
$n_o$, and moreover
\begin{equation}
\label{eq:objective_inf}
\lim_{n_o\to + \infty }\limsup_{t\to +\infty } |e(t)|\;=\; 0,
\end{equation}
when the basic frequency characterizing the oscillators 
$\Phi_\ell$ of the internal-model unit
 is selected as $\widehat\omega = \tfrac{2\pi}T$, that is
$\omega_\ell = \ell\tfrac{2\pi}T$ for any $\ell= 1, \ldots, n_o$.

\subsection{Approximate Regulation}\label{sec:finite}

First, we study the interconnection of the finite-dimensional
controller \eqref{eq:controller}
 in closed-loop with the system \eqref{eq:sys}
 and we establish a certain number of properties concerning
 the existence and stability of a steady-state trajectory
 and the norm of the corresponding regulated output.
The proof of the following theorem is postponed to 
Section~\ref{sec:proof1}. 
 
\begin{theorem}\label{thm1}
Given a triplet 
$(\overline P_x,\underline P_x,\alpha)$
and any fixed $\mu\geq1$, there exist 
positive real numbers
 $\boundx, \bounde,\boundz,\sigma ^\star$,
independent of $n_o$, such that, 
 the closed-loop system 
\eqref{eq:sys}, \eqref{eq:controller}, 
with
any
$f\in \F(2\boundx, 2\bounde,\overline P_x,\underline P_x,\alpha)$
 and
any
$q\in\Q(2\boundx, 2\bounde)$,
satisfies the following statements:
\begin{enumerate}[label=\arabic*)]
\item for any $\sigma>\sigma^\star$ and any $n_o>0$,
the closed-loop system 
\eqref{eq:sys}, \eqref{eq:controller}
admits a $T$-periodic solution
$ (x_p, e_p, z_p)\in \C^2_T(\RR^n\times \RR \times \RR^{2n_o+1})
 $
 satisfying
\begin{equation}\label{eq:theorem_periodic_bound}
\sup_{t\in[0,T]}|x_p(t)|
\leq \boundx \, , 
\qquad
\sup_{t\in[0,T]}|e_p(t)| \leq \bounde \, , 
\qquad
\sup_{t\in[0,T]}\sqrt{z_p(t)^\top N_z z_p(t)}
\leq \boundz \, .
\end{equation} 
Moreover $(x_p, e_p, z_p)$ is locally 
exponentially stable with a 
domain of attraction that  includes the set  
\begin{equation}
\N(\boundx, \bounde,\boundz)=
\Big\{(x,e,z): \, |x| \leq  2\boundx,
\, |e| \leq  2\bounde,\,
\sqrt{z^\top N_z z} \leq 2\boundz
\Big\}
\label{eq:theorem_periodic_domain}
\end{equation}
which is independent of $n_o$.

\item There exists  $\psi_1, \psi_x>0$,
independent of $n_o,\sigma, \mu$, such that, 
for any $\sigma>\sigma^\star$ and any $n_o>0$,
the corresponding $T$-periodic solution  $(x_p, e_p, z_p)$
of the closed-loop system 
\eqref{eq:sys}, \eqref{eq:controller}
established in item 1)
satisfies
\begin{equation}
\label{eq:theorem_2_bound_e1}
\sup_{t\in[0,T]}|e_p(t)| 
\leq \dfrac{\psi_1}{\sigma}, 
\qquad \sup_{t\in[0,T]}|x_p(t)| 
\leq \dfrac{\psi_x}{\sigma}.
\end{equation}

\item Let $\sigma>\sigma^\star$
and $n_o$ be fixed.
If, for some $\ell$ in $\{1, \ldots, n_o\}$, there exists an integer $\mathfrak{K}_\ell > 0$
such that $\omega _\ell$
in  \eqref{eq:controller},
satisfies 
 $$
\omega_\ell = \mathfrak{K}_\ell \tfrac{2\pi}{T},
$$
then
the corresponding $T$-periodic solution  $ (x_p, e_p, z_p)$
of the closed-loop system 
\eqref{eq:sys}, \eqref{eq:controller}
established in item 1) satisfies
\begin{equation}
\label{eq:FourierCoeff}
\int_{0}^T \sin(\mathfrak{K}_\ell\tfrac{2\pi}T t)e_p(t) = 
\int_{0}^T \cos(\mathfrak{K}_\ell\tfrac{2\pi}T t)e_p(t) = 0.
\end{equation}

\item
Suppose that we have $\widehat \omega = \tfrac{2\pi}T$ in \eqref{eq:basic_omega}.
Then, there exists  $\psi_2>0$,
independent of $n_o,\sigma,\mu$, such that, 
for any $\sigma>\sigma^\star$ and $n_o>0$,
the corresponding $T$-periodic solution  $(x_p, e_p, z_p)$
of the closed-loop system 
\eqref{eq:sys}, \eqref{eq:controller}
established in item 1)
satisfies
\begin{equation}
\label{eq:theorem_2_bound_e2}
\int_0^T|e_p(t)|^2 dt 
\leq \dfrac{\psi_2}{(n_o+1)^2},
\end{equation}
$$
\int_0^T \left(\begin{array}{@{\,}c@{\,}}
\cos(\omega_\ell t) \\ \sin(\omega_\ell t)
\end{array}\right)e_p(t) dt\;=\; 0
\qquad \forall \ell \in \{1, \ldots, n_o\}.
$$
\end{enumerate}
\end{theorem}

Theorem~\ref{thm1} establishes several properties
about the solutions of the closed-loop system 
\eqref{eq:sys}, \eqref{eq:controller}.
Item 1) states that for $\sigma$ large enough, 
the closed-loop system admits an exponentially stable
$T$-periodic steady-state trajectory.
Existence, stability and domain of 
attraction of such a steady-state is 
robust to model uncertainties and
independent of the  parameters
of the internal-model unit in   \eqref{eq:controller}, 
that it is independent of 
$n_o$, and the sequences $n_{z\ell}, \omega_\ell$
characterising the frequencies of the oscillators
of the  $z$-dynamics, provided the 
conditions  \eqref{eq:conditions_series} hold.
Furthermore, it is shown that  the initial 
condition $z(0)=0$ for the internal-model unit
\eqref{eq:controller} 
 is always a ``good initial solution''
as it is always contained in the domain of attraction
\eqref{eq:theorem_periodic_domain}. 

Item 2 establishes that the controller \eqref{eq:controller}
preserves the high-gain property of a feedback without 
internal-model: 
the infinity norm of the steady-state trajectory 
of the regulated output $e$ can be arbitrarily made small
by augmenting the parameter $\sigma$, see 
\eqref{eq:theorem_2_bound_e1}. In other words, 
we don't loose the properties of a simple high-gain feedback 
$u=-\sigma e$. Moreover, such a high-gain property is robust to 
model uncertainties and independent of the  parameters 
$n_{z\ell}$, $\omega_\ell$ in \eqref{eq:conditions_series}.

Item 3 characterises the behaviour of the steady-state 
of the regulated output $e$ when the frequency 
of one oscillator is a multiple of the basic frequency 
$\tfrac{2\pi}{T}$ characterising the periodicity 
of the frequencies $f,q$. In this case, the
Fourier coefficient of $e$ corresponding to that frequency
 is zero. Such a property evidently suggests a 
 strategy to select the parameters $\omega_\ell$ in 
 \eqref{eq:controller} when the  periodicity 
 $T$ of the functions $f,q$ is known. 
 This is  well established in item 4. 
 
 In particular, the inequality 
 \eqref{eq:theorem_2_bound_e2} shows that the 
 $L_2$ norm of the steady-state regulated output $e$
 can be made arbitrarily small by augmenting the 
 number of oscillators, if those are chosen so that
 their frequency is multiple of the basic frequency 
$\tfrac{2\pi}{T}$. This is a consequence of the fact that
each corresponding Fourier coefficient is zero, 
as established by \eqref{eq:FourierCoeff}.
Note that although  a similar result 
were already proved in 
\cite{astolfi2015approximate}, 
\cite{ghosh2000nonlinear}, 
the novelty of item 4) is that here we are able
to show that the inequality  \eqref{eq:theorem_2_bound_e2}
is uniform in the parameters of the controller \eqref{eq:controller}. 
This implies that, from a practical point of view, 
one can first fix the parameters $\mu,\sigma$, 
and then arbitrarily increase $n_o$ so that 
to reduce the $L_2$ norm of the regulated output.
Furthermore, in doing so,  the 
 domain of attraction
is  guaranteed to always exist and contain a 
prescribed set of initial conditions independent of 
the parameters of the internal-model.
 In other words, when the period $T$ is known, 
 the approximate  regulation objective
 \eqref{eq:objective} can be satisfied 
 by augmenting the number of oscillators and not 
 the high-gain parameter.

In the next section, we will show that  the bound 
\eqref{eq:objective_inf} hold when considering
the  non-implementable infinite-dimensional case, 
corresponding to the limit case
of the regulator \eqref{eq:controller}
 in which $n_o=+\infty$.

\subsection{Exact Regulation}\label{sec:infinite}

In this section we want to study the limit case
(non-implementable)
 in which the period $T$ is perfectly known and the
  number of oscillators in the regulator 
 \eqref{eq:controller} is chosen as infinite 
$n_o=\infty$,  with
$\widehat\omega= \tfrac{2\pi}T$, that is
 $\omega_\ell =\tfrac{2\pi}T$ for all 
 $\ell \in\{0,1, \ldots, \infty\}$. 
  In particular we aim at establishing that, 
in such a case, 
exact regulation is achievable and 
\eqref{eq:objective_inf} is satisfied. 
 To this end, 
let us define the linear operators $\Phi, N_z, M, \Gamma$ 
as 
\begin{equation}
\label{eq:defPhiM-inf}
\begin{array}{rclrclrcl}
\Phi& := &\big(\Phi_\ell  \big)_{\ell \in \NN_0}, 
\quad &
\Phi_0 & := &0, \quad &
\Phi_\ell  & := & \ell \dfrac{2\pi}T
\begin{pmatrix}
0 & 1 \\ -1& 0
\end{pmatrix},
\\[1em]
N_z & := &  \big(N_{z\ell }\big)_{\ell \in \NN_0},
\quad
& N_{z0} & := &n_{z0},  &
\quad 
N_{z\ell } & := &  n_{z\ell }
\begin{pmatrix}
1 & 0 \\ 0 & 1
\end{pmatrix},
\\[1em]
M & := & \big(M_{\ell }\big)_{\ell \in \NN_0},
\quad
& M_0 & := & 1, 
& M_\ell  &:=& (1, 0)^\top ,
\\[1em]
\Gamma &: = &-(\Phi + \sigma I) M
\end{array}
\end{equation} 
 with $(n_{z\ell })_{\ell \in \NN_{\geq0}}$ being the sequence
 defined in \eqref{eq:sequence_nzk}.
We denote with 
 $\cZ$ the space of sequences
\begin{equation}
\label{eq:def_space_Z}
\cZ := \{ z=(z_\ell )_{\ell \in \NN_0}, \; z_0\in \RR, \;
\; z_\ell  \in \RR^2, \; \ell \in \NN
\}
\end{equation}
and we define the space 
 $\cL^2_{N_z}$ as
\begin{equation}\label{eq:def_space_L2Nz}
\begin{array}{rcl}
\cL^2_{N_z} & :=&\displaystyle
\left\{ z\in \cZ: \|z\|_{N_z}^2 :=
 {z^\top N_z z} =  \sum_{\ell =0}^\infty  n_{z\ell }|z_\ell |^2 < \infty
\right\}
\end{array}
\end{equation}
with $N_z$ being defined in \eqref{eq:defPhiM-inf}. 
This space, being linearly isometric with the standard $\cL^2$ space, is complete.
 In this section, we
 address 
 the specific case in which the regulator is selected as 
 \begin{equation}
\label{eq:controller-inf}
\begin{array}{rcl}
\dot z &=&  \Phi z +\Gamma  e
\\[.5em]
u&=& -\sigma e + \mu M^\top N_z (z-Me)  \ ,
\end{array}
\end{equation}
where $z \in \cZ$,  is the state 
with initial condition $z(0) \in \cL^2_{N_z}$, and 
the linear operators $\Phi, M, N_z$ are now defined as in
\eqref{eq:defPhiM-inf}.
As now the state $z$ is a vector of infinite,
but countable
dimension, 
we will consider only solutions 
in the space $\cL^2_{N_z}$.
We have the following result showing that exact regulation can be achieved.
The proof of the following theorem is postponed to 
Section~\ref{sec:proof1}.

\begin{theorem}\label{thm2}
Let the triplet 
$(\overline P_x,\underline P_x,\alpha)$ be given 
and  fix any $\mu\geq1$.
Consider the real numbers $\boundx, \bounde, \boundz$ and $\sigma^\star$
given by Theorem~\ref{thm1}.
There exists $\sigma^\star_\infty < \sigma^\star$
such that, 
for any $\sigma  > \sigma^\star_\infty$,
 any trajectory
$(x(t),e(t),z(t))$
of the closed-loop system 
\eqref{eq:sys}, \eqref{eq:controller}, 
with $f\in \F(2\boundx, 2\bounde,\overline P_x,\underline P_x,\alpha)$ 
and $q\in\Q(2\boundx, 2\bounde)$,
starting in the set
$$
\{(x,e,z)\in \RR^n\times \RR\times \cL_{N_z}^2: |x|\leq 2\boundx, \; |e|\leq 2\bounde, \;
\|z\|_{N_z} \leq 2\boundz\},
$$
is defined and complete forward in time, 
bounded  in $\RR^n\times\RR\times \cL^2_{N_z}$,  
and satisfies
$\lim_{t\to\infty}e(t) = 0$ and $\lim_{t\to\infty}x(t) = 0$.
\end{theorem}

Theorem~\ref{thm2} establishes that with an
infinite-dimensional regulator, exact
regulation can be achieved if the  period 
$T$ characterising the functions $f,q$ 
is perfectly known, and if the regulator embeds
an infinite-number of oscillators at $\tfrac{2\pi}T$
and all its multiples. In practice, 
such a result confirms  
the property  \eqref{eq:objective_inf} of the 
finite dimensional  regulator 
\eqref{eq:controller}.
It is noticed that in statement we showed that
the domain of attraction 
in terms of  $(x,e)$-coordinates
is not reduced with respect to  those given 
by Theorem~\ref{thm1}. Furthermore, 
for the same values of $\boundx,\bounde$, the resulting
high-gain parameter $\sigma$
 can be chosen smaller with an exact 
 infinite-dimensional regulator. 
 The main motivation to this fact
is that, 
in the approximate case, the steady-state solution $x_p$
does not coincide with the origin,
thus reducing the
stability margins ensured by \eqref{eq:Pass2}. 
In turn, $\sigma$ has to be chosen larger 
to compensate
such a loss.

\subsection{Literature Review}
\label{sec_LR}
The problem of periodic output regulation 
has been studied in the past decades by many authors with 
many different tools and ideas. Although various approaches as ours allow to cope with more general dynamics, for the sake of precision, we restrict our discussion to systems which admit a normal form like the one at the beginning of 
Section \ref{sec:prob} or more generally like (\ref{eq:sys_rel_deg_r}). We revisit the 
following main approaches based on the use of ``smooth 
regulators''.

\begin{itemize}
\item \textit{Nonlinear output regulation.}
Starting from
the notable results on the so-called 
non-linear regulator equations 
\cite{byrnes2003limit} in a non-equilibrium context,
the development of output regulation theory has been 
mainly pursued in the context of minimum-phase systems 
of the form \eqref{eq:sys}. 
The design of internal-models has been focused mainly 
in the sense of input-cancellation/observation, that is,
with the purposes of reproducing the asymptotic behaviour of the 
zeroing steady-state input $-q(t,0,0)$
for system \eqref{eq:sys_e},
see, e.g., 
\cite{byrnes2004nonlinear,marconi2007output}.
Although these approaches can ensure 
asymptotic regulation with finite-dimensional 
regulators, it is not clear whether they 
can be extended in a non-matching case, i.e. without the use of 
a normal form \eqref{eq:sys}. Furthermore, 
as discussed in 
see \cite{bin2018robustness,bin2022robustness}, 
asymptotic regulation is lost as 
soon as unstructured 
model uncertainties are considered.
Approximate asymptotic solutions have been  also 
proposed in \cite{marconi2008uniform,bin2019adaptive,
gao2017learning}, but again, the extension to 
more general classes of systems (as in \cite{astolfi2015approximate})
is not clear.

\item \textit{Repetitive control.} 
Based on the fact that a delay
can be used  as a universal generator of periodic signal, 
the repetitive control approach was first proposed at the end of the 80's, 
\cite{hara1988repetitive} for linear systems in order to solve
periodic output regulation problems, with remarkable 
results in the context of 
discrete-time systems
\cite{longman2000iterative} and  practical applications 
\cite{kurniawan2014survey}.
Nonlinear extensions were proposed 
for instance in \cite{ghosh2000nonlinear,
mattavelli2004repetitive,astolfi2021repetitive}.
Similarly to our approach, finite-dimensional approximations
based on Fourier approximation were used 
in \cite{ghosh2000nonlinear}, 
based on the fact that a delay can be equivalently represented 
by an infinite-number of oscillators\footnote{This can be shown, for instance, by using Riesz bases, 
 see Example 2.6.12 in \cite{tucsnak2009observation}.}
 In this spirit, our 
work contributes also to the repetitive control theory 
clarifying the uniformity aspects in terms of
basin of attraction of periodic solutions with respect 
to the approximation of the internal-model unit.

\item \textit{Adaptive learning control.}
An approach similar to repetitive control is also the one 
denoted as adaptive learning control, 
see, e.g., 
\cite{del2003adaptive,marino2009iterative,verrelli2016larger}, 
developed in the context of minimum-phase systems \eqref{eq:sys}, 
with the objective
of estimating the Fourier coefficients of the 
zeroing steady-state input $-q(t,0,0)$
\cite{del2003adaptive}, or canceling it 
by means of delays \cite{marino2009iterative,verrelli2016larger}.
Extensions to systems not possessing a  normal forms 
and practical implementations issues related to 
the delay (and therefore
the asymptotic properties of 
a discretised regulator) 
have not been discussed.

\item \textit{Input disturbance observers}. Finally, for systems in normal form, 
input disturbance observers can be used in output regulation 
of minimum-phase systems \eqref{eq:sys}, see, e.g., 
\cite{han2009pid}. Again, the extension of such an approach to 
more general classes of systems (as in \cite{astolfi2015approximate}) is not clear.
\end{itemize}

\subsection{Higher-Relative Degree Case via Partial-State Feedback}
\label{sec:higher_degree}

Consider now  a system of relative degree higher than one
and described by 
\begin{equation}
\label{eq:sys_rel_deg_r}
\begin{array}{rcl}
\dot \chi & =& f_0(t,\chi,\xi_1)
\\
\dot \xi_i & = & \xi_{i+1} \qquad i = 1,\ldots, r-1,
\\
\dot \xi_{r} & =& q_0(t,\chi,\xi) + u ,
\end{array}
\end{equation}
with $\chi\in \RR^n$, $\xi = (\xi_1, \ldots, \xi_r)^\top\in \RR^n$, 
and suppose that our output regulation objective is now given by 
\begin{equation}
\label{eq:objective_1}
\limsup_{t\to \infty } |\xi_1(t)|\; \leq \; \boundxi_{1p}.
\end{equation}
Following for instance \cite{serrani2001semi}, 
we consider the change of coordinates 
$$
\xi_r \mapsto e := \xi_r + \sum_{i=1}^{r-1}a_i \xi_i
$$
where $a_i$ are chosen so that 
$\lambda^{r-1} + a_1 \lambda^{r-2} + \ldots + a_{r_2}\lambda + a_{r-1}$
is a Hurwitz polynomial.
In the new coordinates, system
\eqref{eq:sys_rel_deg_r} reads as
\begin{equation}
\label{eq:sys_rel_deg_r2}
\begin{array}{rcl}
\dot x & = & f(t,x,e)
\\
\dot e & =& q(t,x,e) + u 
\end{array}
\end{equation}
with $x  := (\chi^\top,y^\top)^\top$, 
$y  :=  (\xi_1, \ldots, \xi_{r-1})^\top$, 
$$
\begin{array}{rcl}
f(t,x,e) & := & 
\begin{pmatrix}
f_0(t,\chi,Cy)
\\[.5em]
 A y + B e
\end{pmatrix}
\\[2em]
A & := & \begin{pmatrix}
0_{r-2,1} & I_{r-2}
\\
-a_1 & -a_2 \cdots - a_{r-1}
\end{pmatrix}, 
\quad
B := \begin{pmatrix}
0_{r-2, 1} \\ 1
\end{pmatrix}, 
\quad
C := \begin{pmatrix}
1 & 0_{1, r-1}
\end{pmatrix},
\\[1em]
q(t,x,e)& := & \displaystyle
q_0(t,\chi,\xi_1, \ldots, \xi_{r-1}, 
e-  \sum_{i=1}^{r-1}a_i \xi_i)
+
\sum_{i=1}^{r-2}a_i \xi_{i+1}
+a_{r-1}(e - \sum_{i=1}^{r-1}a_i \xi_i).
\end{array}
$$
By construction, 
if $f_0, q_0$ are $C^2$ and $T$-periodic, 
so are  
$q$ and $f$. 
 As a consequence, we can apply the control design proposed in 
 Sections~\ref{sec:finite}, \ref{sec:infinite}
 to system \eqref{eq:sys_rel_deg_r2}. 
 In particular, concerning the result  of Theorem~\ref{thm1}, 
 it is straightforward to see that the control law
 \eqref{eq:controller} applies to system 
 \eqref{eq:sys_rel_deg_r2}
insuring the desired properties
for the new regulated output $e$ of system \eqref{eq:sys_rel_deg_r2}.
Then, by linearity of the change of coordinates, the properties
on $y,e$ can be used to analyse the behaviour of $\xi_1$.
Indeed, the $r$-th derivative of 
$\xi_1$, i.e., $\xi_1^{(r)}$, is given by
$$
\xi _1^{(r)}\;=\; -\sum_{k=1}^{r-1} a_k \xi _1^{(r-1-k)} + e\,.
$$
With  $\lambda^{r-1} + a_1 \lambda^{r-2} + \ldots + a_{r_2}\lambda + a_{r-1}$
being a Hurwitz polynomial, we have
$$\limsup_{t\to \infty }|\xi_1 (t)|\leq 
\mathcursive{k}_\infty \bounde_p
\quad ,\qquad 
\int_0^T \xi _1(t)^2 dt \leq
\mathcursive{k}_2 \int_0^T e(t)^2 dt$$
where the real numbers $\mathcursive{k}_\infty $ and 
$\mathcursive{k}_2 $ depends only on the $a_k$.
The regulation objective \eqref{eq:objective_1}
is then satisfied 
by selecting $\bounde_p< \dfrac{\boundxi_{1p}}{\mathcursive{k}_\infty}$.

A formal theorem concerning the design of a regulator 
for the output regulation problem for system 
\eqref{eq:sys_rel_deg_r} is therefore not given as it can 
be directly inherited by combining the computations
 of this section
and the results of  Theorems~\ref{thm1} and \ref{thm2}.

\section{Example}
\label{sec:example}

\subsection{Linear Bode Analysis}

In this section we provide a numerical example
to corroborate the theoretical results of 
Section~\ref{sec:main}.
First, in order to have a deeper insight 
of the proposed algorithm, we consider 
a simple linear case 
with no zero-dynamics
given by 
\begin{equation}
\label{eq:example-linear}
\dot e = u  + q(t)
\end{equation}
and we analyse the difference between the transfer functions
of a high-gain feedback  regulator
\begin{equation}\label{eq:example-hgf}
u = -\sigma e
\end{equation}
and the internal model based regulator 
\eqref{eq:controller}. 
Figure~\ref{fig:bode} 
shows the transfer function between the input 
$q$ and the output $e$
for the closed-loop systems 
\eqref{eq:example-linear}, \eqref{eq:example-hgf}, 
respectively
 \eqref{eq:example-linear}, \eqref{eq:controller},
when $\sigma=2$ and the other parameters are selected as 
$\mu=1$, $n_{z\ell}$ selected as in \eqref{eq:nzk} with
$\varepsilon=0.5$, $\widehat{\omega}=2\pi$, and $n_o=10$.
As expected, it is readily seen that the effect of the internal model 
is to add blocking zeros at the desired frequencies 
$\ell\widehat\omega$, $\ell=0, 1, \ldots, 10$, while preserving
the same transfer function of the high-gain feedback 
\eqref{eq:example-hgf} at higher frequencies. 
In the next section, we will show in a numerical simulation
that this blocking effect is preserved in the nonlinear 
context, confirming the results of 
Theorems~\ref{thm1} and \ref{thm2}.

\subsection{Numerical Example}

Consider here as a simple example 
a system with unitary relative degree of the form 
\eqref{eq:sys}. For the simulations, the nominal functions $f,q$ are selected as 
\begin{equation}\label{eq:example}
\begin{array}{rcl}
f(t,x,e) &= &
\begin{pmatrix}
-\tfrac{1}{5} x_1+\sqrt3 x_2 + \tfrac1{10}\sin(x_2) 
\\
-\sqrt3 x_1 -x_2  +  \tfrac1{10} x_2^2
\end{pmatrix}
+
\begin{pmatrix}
0 \\ x_2 e
\end{pmatrix} 
+
\begin{pmatrix}
\sin(2\pi t)
\\
\cos(4\pi t)(1+\sin(2\pi t))
\end{pmatrix},
\\
q(t,x,e) & =& 1+x_1+ \arctan(e x_2)+ 
\sum_{k=1}^4\cos^k(2\pi t) 
\end{array}
\end{equation}
The functions $f,q$ have been randomly chosen
``ugly''  
 so that  other frameworks for
asymptotic regulation 
such as \cite{byrnes2004nonlinear,marconi2007output}
cannot be explicitly applied.
It is readily seen that, around the origin, the $x$-dynamics is locally exponentially stable, while for large values of $x$, it is unstable
due to the term $x_2^2$.
Furthermore, $f,q$ are $T$-periodic smooth functions with period 
$T = 1$.
Hence, the  $f,q$ are included in the Families
 $\F$ and $\Q$ of Definitions~\ref{def:familyF} and
\ref{def:familyQ}.

We then evaluated the steady-state performances 
of a simple high-gain feedback \eqref{eq:example-hgf},
and the internal model based  regulator \eqref{eq:controller}, 
in two scenarios: without any measurement noise, 
and in presence of high-frequency measurement noise $v$.
In the noisy scenario, 
the  high-gain feedback \eqref{eq:example-hgf}
becomes
 $$
 u = -\sigma (e + v)
 $$
while
 the internal model based  regulator \eqref{eq:controller} 
 reads as 
 $$
 \begin{array}{rcl}
\dot z & = & \Phi z+ \Gamma (e+ v)
\\ 
 u & = & - \sigma e + \mu M^\top N_z [z - M(e+v)].
 \end{array}
 $$
We then performed
different simulations.
Since we are not interested in transient response, 
the initial conditions are selected always as 
$x(0)=(1,-2)$, $e(0)=4$ and $z(0)=0$.
For the high-gain feedback, we selected different values of 
$\sigma=\{2,10,20,40\}$.
For the controller \eqref{eq:controller}
 we selected $\sigma=2$, $\mu=1$
and $n_{z\ell}$ selected as in \eqref{eq:nzk} with
$\varepsilon=0.5$ The number of oscillators is varied
in the range $\{0,\ldots, 4\}$,
with $0$ corresponding to a pure integral action, 
while the  basic frequency \eqref{eq:basic_omega} is selected 
as $\widehat\omega=\{\tfrac{2\pi}T, 0.99\tfrac{2\pi}T, 
0.95\tfrac{2\pi}T, \tfrac{2\pi}T \varphi_g\}$
with  $\varphi_g =\tfrac{1+\sqrt5}2$ being the golden number.
Simulations have been performed with 
Matlab-Simulink, with a fixed-step time algorithm, 
sampling time $10^{-5}$, over a time length of 150 seconds.

As established in Theorem~\ref{thm1}, 
convergence to a steady-state trajectory $e_p$ is guaranteed
no matter the dimension $n_o$ and the frequency $\hat \omega$.
Table~\ref{table:high-gain} lists the values of the
$L_\infty$ norm and $L_2$ norm of the steady-state error $e_p$
with the high-gain feedback \eqref{eq:example-hgf} for
the different values of $\sigma$. Applying a 
FFT (fast Fourier transform) on the last 
20 seconds of simulations (so that solutions reached
their steady-state)
 we identified the 
corresponding main frequencies 
Figure~\ref{fig:Fourier_hgf} of the steady-state error
$e_p$ for $\sigma=2$, showing that the main Fourier 
coefficients are at frequencies $2k\pi$, with $k=0,1,\ldots, 4$, 
as expected since $T=1$.
Then, Table~\ref{table:example-norms}
lists the values of the
$L_\infty$ norm and $L_2$ norm of the steady-state error $e_p$
for the controller \eqref{eq:controller} in the different 
scenarios.
The corresponding Fourier coefficients
for $\widehat{\omega}=\{1,0.99,0.95\}\tfrac{2\pi}T$
 are shown 
in Figures~\ref{fig:Fourier_nominal}, \ref{fig:Fourier_1percent}
and \ref{fig:Fourier_5percent}. We can observe
that when the basic frequency $\widehat\omega$ is 
less accurate, the errors on the higher 3-4  frequencies 
 becomes more important and therefore the blocking effect of the zero
becomes less useful. As a consequence, 
when the frequency is not perfectly known, 
it may be not so interesting to put only multiple 
frequencies of $\widehat{\omega}$.

 From Tables~\ref{table:high-gain}
 and \ref{table:example-norms}, 
 it is immediately seen the remarkable 
 increase of performances (in terms of steady-state $L_\infty$
 and $L_2$ norms)
 with the proposed controller, even with a very bad knowledge
 of $T$, with respect to a pure high-gain feedback:
  for the same value of $\sigma$, 
we obtain a sensible reduction of both norms when 
$\widehat{\omega}=\tfrac{2\pi}{T}\varphi_g$, 
and to have the similar performances with a high-gain feedback
we should take  a much higher value of $\sigma$.  
However, note that when $\hat\omega$ is taken very distant from the nominal value
$\tfrac{2\pi}T$, adding extra oscillators is not very useful:
indeed the sup and $L_2$ norms are not reduced anymore.
On the contrary, when $T$ is perfectly known, with only 4 oscillators
we are already able to achieve almost perfect tracking 
as the remaining error is nearly negligible.

 Finally, the same scenarios have been 
 performed in presence 
 of high-frequency measurement noise $v$ 
generated by colouring some random white noise 
with a high-pass filter. In simulations we used 
a Simulink ``Band-Limited
White Noise'' block
with noise power $10^{-3}$ and  sampling 
time  $10^{-5}$ and the following transfer function 
$$
 H(s) = \dfrac{s^2}{s^2+3s+2}
$$ 
with $s$ being the Laplace operator.
Simulations show that  $\sup_{t\in[0,\infty)}|v(t)| \leq  40$.
The values of the $L_\infty$ and $L_2$ norms 
(computed over one random period $T$ among the last 30 seconds
of simulations)  
 of the asymptotic value of the error 
are listed in Tables~\ref{table:high-gain}
and \ref{table:example-norms}.
Again, simulations confirm the advantages 
(in term of $L_\infty$ and $L_2$ norms) of the 
proposed internal model based regulator
\eqref{eq:controller}
over a high-gain feedback controller 
\eqref{eq:example-hgf}.

\section{Conclusions}
\label{sec:conclusions}

In this work we addressed the problem of 
exact and approximate periodic  robust output regulation 
for minimum-phase systems. We investigated 
the use of an internal-model based regulator
which is a straightforward  
extension of  the linear case established by 
Francis, Wonham and Davison, and conceptually 
similar to what is used in practical repetitive 
control scheme approaches. In practice, the internal model
unit is composed by a bunch of linear oscillators processing
the regulated output.  

The main contribution of this work is to establish
that the domain of attraction of an exponentially stable
periodic solution is uniform in the  parameters characterising 
the proposed regulator (i.e. the number of oscillators 
and their frequencies). 
The result is also robust to model uncertainties
as only the Lipschitz properties of the system 
dynamics (and not the exact expressions) 
are used in the computations.
Is is shown that the quality 
(in terms of $L_2$ norm)
of the  periodic steady-state of the regulated
error improves by augmenting the number of oscillators 
and that exact regulation can be achieved when the number of oscillator
is infinite and the period is perfectly known.
Simulations confirm the theoretical findings
and shows the improvements of the proposed regulator 
with respect 
to a pure high-gain feedback controller or a simple PI controller.

The main results of this work shows that the best performances
with the proposed regulator
 can be obtained when the period characterising
the periodicity of disturbances/references is perfectly known.
This consideration leads to many open questions 
concerning the knowledge of the frequencies
of perturbations/references in practical applications, 
optimal choices of the parameters of the regulator
in terms of number of oscillators and their frequencies,
and  possible strategies for offline/online identification
of such frequencies.  
For all these problems, adaptive, identification or learning
 techniques
may be a key  tool, see 
\cite{nikiforov1998adaptive,kocan2020supervised,afshar2019adaptive,
bin2019adaptive,marino2016hybrid,serrani2001semi,
marino2015online,xu2013global,gao2017learning} as few examples.

Finally, from the theoretical point of view, 
an exhaustive analysis and  extension 
of the results presented in this article
to systems not in normal form
and possibly non-minimum phase, as in 
\cite{astolfi2015approximate,astolfi2017integral},
remains
 an open problem.

     \begin{figure}\sidecaption
\resizebox{1\hsize}{!}{\includegraphics*{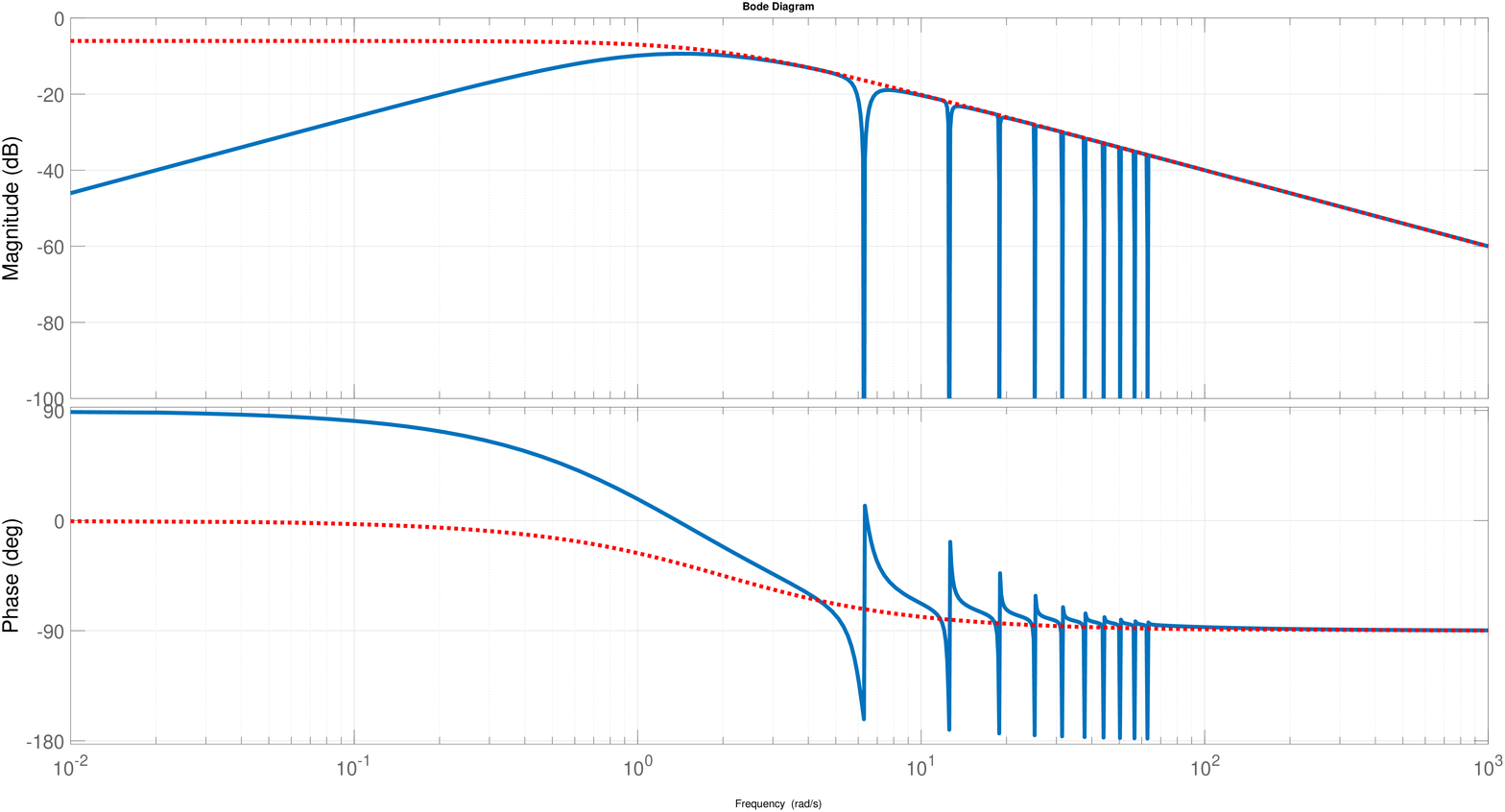}}
\caption{Bode diagram of  the transfer between $e$
and $q$ for the closed-loop systems 
\eqref{eq:example-linear}, \eqref{eq:example-hgf}, 
in blue, 
respectively
 \eqref{eq:example-linear}, \eqref{eq:controller},	in red. }\label{fig:bode}
\end{figure}
   
\begin{center}
\begin{table}\renewcommand{\arraystretch}{1.5}
\begin{tabular}{p{.5cm}|p{2cm}p{2cm}|p{2cm}p{2cm}}
\multicolumn{1}{c}{ } & \multicolumn{2}{c}{no measurement noise} & 
\multicolumn{2}{c}{with measurement noise}
\\
\toprule
$\sigma$  & $\sup_{t\in[0,T]}|e_p(t)|$ & $\tfrac1T\int_0^T|e_p(t)|^2dt$
& $\sup_{t\in[0,T]}|e_p(t)|$ & $\tfrac1T\int_0^T|e_p(t)|^2dt$
\\
\toprule
2 & 1.2555 &   0.9657  & 1.3023 & 0.9743  \\ 
5 &    0.6577 &     0.4083   & 0.7995 & 0.4244 \\ 
10   & 0.400 &   0.2166 &  0.6619 &     0.2489 \\
20   & 0.2248 &    0.1126 & 0.6201 &       0.1765\\
40  & 0.1181 &    0.0572 &  0.6736 &       0.1735\\ 
\bottomrule
\end{tabular}
\caption{Simulation of the example \eqref{eq:example}
with high-gain feedback control \eqref{eq:example-hgf}, for different 
selections of $\sigma$, with and without measurement noise. 
The noise is generated as random
white noise coloured with a high-pass filter.
 }\label{table:high-gain}
 \end{table}
\end{center}

\begin{center}\renewcommand{\arraystretch}{1.5}
\begin{table}
\begin{tabular}{p{.1cm}p{.1cm}p{1cm}p{.1cm}|p{1.7cm}p{1.7cm}|p{1.7cm}p{1.7cm}}
 \multicolumn{4}{c}{} & 
 \multicolumn{2}{c}{no measurement noise} &
  \multicolumn{2}{c}{with measurement noise}\\
\toprule
$\sigma$ & $\mu$ &$\widehat\omega$   & $n_o$  & $\sup_{t\in[0,T]}|e_p(t)|$ & $\tfrac1T\int_0^T|e_p(t)|^2dt$
 & $\sup_{t\in[0,T]}|e_p(t)|$ & $\tfrac1T\int_0^T|e_p(t)|^2dt$
\\
\toprule
2 & 1 & - & 0 & 0.3074 &  0.1777  & 0.3915  & 0.2024 \\ 
\toprule
2 & 1& $2\pi$ & 1 & 0.0917  &   0.0549 & 0.2453 &   0.0872 \\
2 & 1&$2\pi$ & 2 & 0.0178  & 0.0099 & 0.1822 & 0.0545  \\
2 & 1&$2\pi$  & 3 &    0.0049 &    0.0035  &  0.1679 &  0.0514\\
2 &1& $2\pi$  & 4 &    3.04  $\cdot 10^{-8}$   &    
 1.66 $\cdot 10^{-8}$  &    0.1730 &     0.0523
\\
\midrule
2 &1& $0.99\cdot2\pi$ & 1 & 0.1145  &    0.0587 & 0.2716 & 0.0985\\
2 &1& $0.99\cdot2\pi$    & 2 &       0.0835  &     0.0371 
&  0.2391 &   0.0814  \\
2 & 1&$0.99\cdot2\pi$    & 3 &    0.0837
 &       0.0369 & 0.2292 &  0.0790
\\ 
\midrule
2 &1& $0.95\cdot2\pi$   & 1 & 0.2045  &  0.0996  &   0.3693 & 
0.1478\\
2 &1&  $0.95\cdot2\pi$   & 2 &     0.2038  &     0.0980 
& 0.3818 &    0.1498\\
2 & 1& $0.95\cdot2\pi$   & 3 &   0.2041  &      0.0982
&  0.3876 &    0.1499
\\ 
\midrule
2 & 1&$2\pi \varphi_g$ & 1 & 0.2915   &   0.1788 &   0.3764 &    0.2020\\
2 &1& $2\pi\varphi_g$  & 2 &  0.2928 & 0.1790 &  0.3761 & 0.2017\\
2 & 1&$2\pi\varphi_g$   & 3 & 0.2929  & 0.1790 &   0.3792 &     0.2015\\
\bottomrule
\end{tabular}
\caption{Simulation of the example \eqref{eq:example}
with the regulator \eqref{eq:controller}, for different 
selections of $n_o$ and $\hat\omega$, with 
$n_{z\ell}$ selected as in \eqref{eq:nzk} with
$\varepsilon=0.5$. Note that $\varphi_g=\tfrac{1+\sqrt{5}}2$ corresponds to the golden number.
 }\label{table:example-norms}
 \end{table}
\end{center}

     \begin{figure}\sidecaption
\resizebox{1\hsize}{!}{\includegraphics*{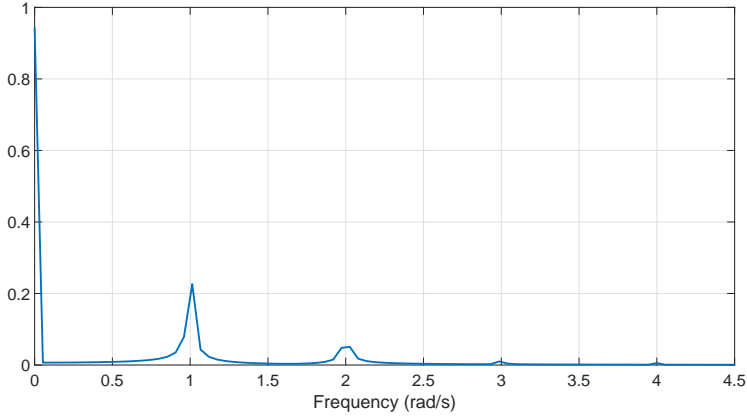}}
\caption{FFT (fast Fourier transform) of the steady-state regulated error
 $e_p(t)$ for the example~\eqref{eq:example} with high-gain 
 feedback \eqref{eq:example-hgf}, with $\sigma=2$. }\label{fig:Fourier_hgf}
\end{figure}

     \begin{figure}\sidecaption
\resizebox{1\hsize}{!}{\includegraphics*{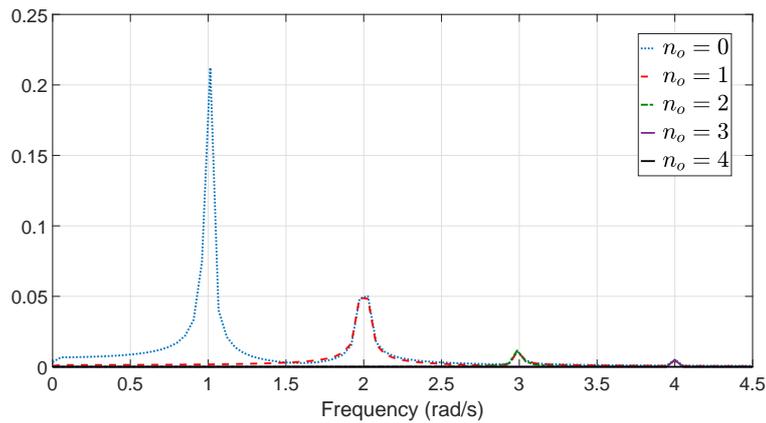}}
\caption{FFT (fast Fourier transform) of the steady-state regulated error
 $e_p(t)$ for the example~\eqref{eq:example} with controller 
 \eqref{eq:controller},  and parameters selected as $\sigma=2$, $\mu=1$, $\varepsilon=0.5$, 
 $n_o=\{0,\ldots, 4\}$ and $\widehat{\omega}=\tfrac{2\pi}{T}$.	 }\label{fig:Fourier_nominal}
\end{figure}

         \begin{figure}\sidecaption
\resizebox{1\hsize}{!}{\includegraphics*{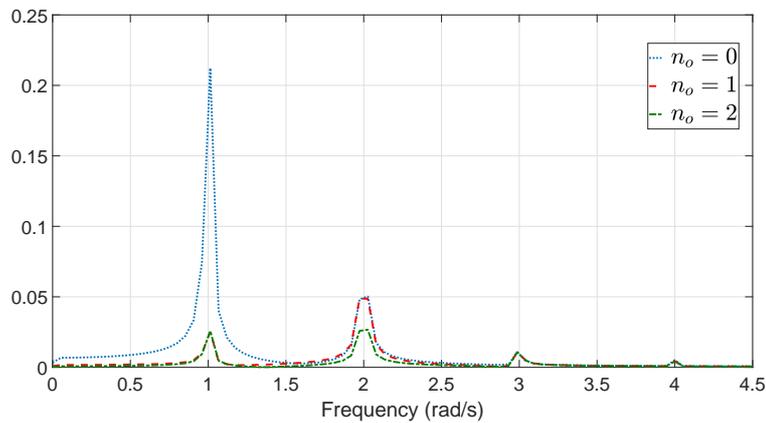}}
\caption{FFT (fast Fourier transform) of the steady-state regulated error
 $e_p(t)$ for the example~\eqref{eq:example} with controller 
 \eqref{eq:controller}, and parameters selected as $\sigma=2$, $\mu=1$, $\varepsilon=0.5$, 
 $n_o=\{0,1, 2\}$ and $\widehat{\omega}=0.99\tfrac{2\pi}{T}$.
 For $n_o=3$, the line is nearly overlapped with 	the one 
 of $n_o=2$. }\label{fig:Fourier_1percent}
\end{figure}

             \begin{figure}\sidecaption
\resizebox{1\hsize}{!}{\includegraphics*{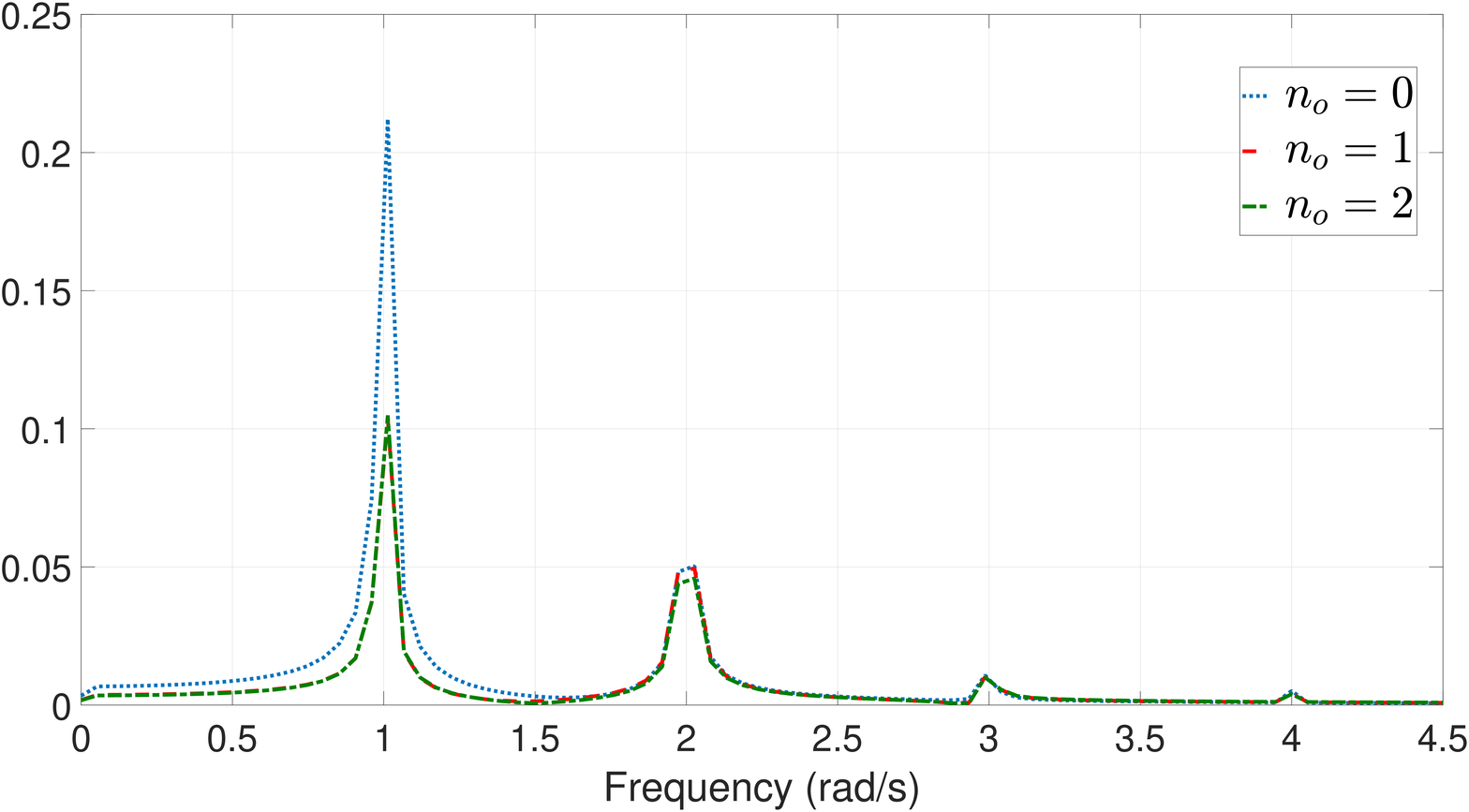}}
\caption{FFT (fast Fourier transform) of the steady-state regulated error
 $e_p(t)$ for the example~\eqref{eq:example} with controller 
 \eqref{eq:controller}, and parameters selected as $\sigma=2$, $\mu=1$, $\varepsilon=0.5$, 
 $n_o=\{0,1, 2\}$ and $\widehat{\omega}=0.95\tfrac{2\pi}{T}$.
 From $n_o>1$ the curves are nearly overlapped.	 }\label{fig:Fourier_5percent}
\end{figure}

   \begin{appendix}
\section{Proof of Theorem~\ref{thm1}}
\label{sec:proof1}

A simple way to establish Theorem~\ref{thm1} could be by
showing that the origin of  \eqref{eq:closed-loop}
for $q=0$ is exponentially stable, and then perturb 
such solution with  a small $q$.
Fixed point theorems and exponential stability arguments
would prove the desired result, 
see, e.g.,  Theorem 3.1, Chapter 8.3, in \cite{miller1982ordinary}.
However, in doing so, all the results would be 
$n_o$-dependent.
Since the objective of this proof is 
to show that this is not the case, namely
the existence of a stable periodic solution
is verified for any choice of $n_o$, 
with bounds that do not depend on $n_o$, 
we are forced to redo  the proof, following the classical route but
re-entering into the details and being careful and precise with the bounds.

\subsection{Preliminaries}
First, let us make the following change of coordinates
\begin{equation}
\label{eq:change_of_coordinates_z_zeta}
z\mapsto \zeta := z- Me
\end{equation}
which transform the closed-loop system 
\eqref{eq:sys}, \eqref{eq:controller}
into
\begin{subequations}\label{eq:closed-loop}
\begin{align}
\label{eq:closed-loop_x}
\dot x & = f(t,x,e)
\\
\dot e & =  q(t,x,e) - \sigma e + \mu M^\top N_z \zeta \label{eq:closed-loop_e}
\\
\dot \zeta & = (\Phi - \mu M M^\top N_z)\zeta  - Mq(t,x,e).
\label{eq:closed-loop_zeta}
\end{align}
\end{subequations}
Our approach to study this system is to decompose it as~:
\begin{equation}
\label{2}
\renewcommand{\arraystretch}{1.5}
\begin{array}{rcl}
\dot x ^+& = &\frac{\partial f}{\partial x}(t,0,0) x^+
\;+\; 
\left[f(t,x^-,e^-)-\frac{\partial f}{\partial x}(t,0,0) x^-\right]
\\
\dot e^+ & = &- \sigma e^+ \;+\;  \left[q(t,x^-,e^-) + \eta  \right]
\end{array}
\end{equation}
where $x^-$ and $e^-$ are inputs and $\eta $ is the output of the system
\begin{equation}
\label{1}
\dot \zeta  = (\Phi - \mu M M^\top N_z)\zeta  - Mq(t,x^-,e^-)
\quad ,\qquad  \eta \;=\; \mu M^\top N_z \zeta  
\end{equation}
Indeed we recover the system \eqref{eq:closed-loop} when inputs $(x^-,e^-)$ equal outputs $(x^+,e^+)$ and
we can benefit from
the following properties~:
\begin{itemize}
\item[$\bullet$] The origin of the subsystem $\dot x= f(t,x,0)$
and therefore of $\dot x ^+=\frac{\partial f}{\partial x}(t,0,0) x^+$
is locally 
exponentially stable in light of 
\eqref{eq:Pass1}, \eqref{eq:Pass2}.

\item[$\bullet$] The origin of subsystem $\dot e =  - \sigma e$
is  exponentially stable.

\item[$\bullet$] The (linear) subsystem \eqref{1} with input $q$ and output $\eta $ is linear and we shall show in
Lemma~\ref{lemma:hurwitzPhiM} below it is stable.
\end{itemize}
We give the $\zeta$ subsystem  (\ref{1}) a special treatment because of its strong dependence on $n_o$. 
Annoying features are, for example,
\begin{itemize}
\item[$\bullet$]
the 
dimension of $\zeta $ is $2n_o+1$;
\item[$\bullet$]
\null  $\displaystyle 
\textsf{trace}(\Phi - \mu M M^\top N_z)\;=\; \mu \,  \left[1 + \sum_{k=1}^{n_o} n_{zk}\right]
$\hfill \null \\[0.5em]
which, with (\ref{eq:sequence_nzk}) and Lemma~\ref{lemma:hurwitzPhiM}, implies the real part of the eigen
values of $(\Phi - \mu M M^\top N_z)$ 
tends to $0$ as $n_o$ tends to infinity.
\end{itemize}

In the following, we start by studying the $\zeta$-subsystem. Then we show that, with a
suitable choice of $\sigma, \mu$
and bounds $\boundx, \bounde$, arguments of our bounding functions
in \eqref{LP1} and \eqref{LP2}
for $f$, $\frac{\partial f}{\partial x}$ 
and $q$, there is a periodic solution
$(x_p,e_p)$ satisfying~:
$$
(x_p,e_p)\;=\; (x^-,e^-)\;=\; (x^+,e^+)
$$
Finally we prove it is exponentially stable and study its domain of attraction and its properties.

\subsection{Study of the system (\ref{1})}

\begin{lemma}\label{lemma:hurwitzPhiM}
Let $\Phi$, $M, N_z$ be 
defined as in \eqref{eq:defPhiM}.
Then, for any $n_o\in \NN$ and any
$\mu>0$, the pair $(\Phi,M^\top N_z)$ is observable and the matrix 
$(\Phi- \mu MM^\top N_z)$ is Hurwitz.
In particular, there exist a symmetric positive definite matrix $P_\zeta $,
depending on $n_o$,
and, for any strictly positive real number $\mu $, there exists a strictly positive real number $\kappa $,
depending on $n_o$,
such that we 
have
\begin{multline}
\label{app:ineq_PhiM}
 (N_z  + \kappa P_\zeta)(\Phi- \mu MM^\top N_z) 
+ 
(\Phi- \mu MM^\top N_z)^\top (N_z  + \kappa P_\zeta)
\;\leq \;- \mu N_z M M^\top N_z-\kappa N_z 
\  .
\end{multline}
\end{lemma}

\begin{proof}
The pair $(\Phi_\ell,M_\ell)$
in \eqref{eq:defPhiM}
is observable.
Then, observability of $(\Phi,M^\top N_z)$
is a direct consequence of the block-diagonal structure
of the matrix $\Phi$ and the fact that $N_z$
is diagonal. 
Since the pair $(\Phi,M^\top N_z)$ is observable
 and $N_z$ is positive definite,
there exist a matrix $K$ and a positive definite matrix $P_\zeta $ 
satisfying
$$
P_\zeta  (\Phi -K M^\top N_z) + (\Phi -K M^\top N_z) ^\top P_\zeta  \;=\; -2N_z 
\  .
$$
On another hand,
since \eqref{eq:defPhiM}
implies 
$N_z \Phi + \Phi^\top N_z = 0$, 
we obtain, 
adding and subtracting the term $\mu N_zMM^\top N_z$,
$$
N_z(\Phi- \mu MM^\top N_z) + (\Phi- \mu MM^\top N_z) ^\top N_z
+ 2 \mu N_z M M^\top N_z\;=\; 0
\  .
$$
By combining these two equations, we get, with $\kappa $ any strictly positive real number,
\begin{multline}\label{eq:ineq_temp}
 (N_z  + \kappa P_\zeta )(\Phi- \mu MM^\top N_z)
+
 (\Phi- \mu MM^\top N_z)^\top(N_z  + \kappa P_\zeta ) 
\\ =  - 2\mu N_z M M^\top N_z  -2\kappa N_z 
 - \kappa \mu ( P_\zeta MM^\top N_z+
N_z  MM^\top  P_\zeta ) 
\\
+ \kappa ( P_\zeta KM^\top N_z + N_z  M  K^\top P_\zeta).
\end{multline}
But, for any matrices $A$ and $B$ and  real numbers $\kappa $ and $\mu $,
we have the  following identity
$$
\mu AA^\top + \kappa \left(BA^\top + A B^\top\right)
=  \left(\sqrt{\mu } A + \frac{\kappa }{\sqrt{\mu }} B\right)
\left(\sqrt{\mu } A + 
\frac{\kappa }{\sqrt{\mu }} B\right)^\top 
- \frac{\kappa ^2}{\mu } BB^\top .
$$
By using previous identity 
in which $A = N_zM$ and $B=P_\zeta (\mu M - K)$, we obtain
\begin{multline*}
- \mu N_z M M^\top N_z - \kappa \mu \left(P_\zeta MM^\top N_z+ N_z MM^\top P_\zeta \right)
+ \kappa \left[P_\zeta KM^\top N_z+ N_z MK^\top P_\zeta \right]
\\
= -\left(\sqrt{\mu }N_zM +\frac{\kappa}{\sqrt{\mu }} P_\zeta (\mu M-K)\right)
\left(\sqrt{\mu }N_zM +\frac{\kappa}{\sqrt{\mu }} P_\zeta 
(\mu M-K)\right)^\top
\\
+ 
\frac{\kappa^2}{\mu } P_\zeta (\mu M-K)(\mu M-K)^\top P_\zeta
\  .
\end{multline*}
Hence, by combining such identity with 
\eqref{eq:ineq_temp},
we obtain
\begin{multline*}
 (N_z  + \kappa P_\zeta )(\Phi- \mu MM^\top N_z)
 +
 (\Phi- \mu MM^\top N_z) ^\top (N_z  + \kappa P_\zeta ) 
 \\
 =
-\left(\sqrt{\mu }N_zM +\frac{\kappa}{\sqrt{\mu }} P_\zeta (\mu M-K)\right)
\left(\sqrt{\mu }N_z M + \frac{\kappa}{\sqrt{\mu }} P_\zeta (\mu M-K)\right)^\top
\\
- \kappa \left(N_z-\frac{\kappa}{\mu } P_\zeta (\mu M-K)(\mu M-K)^\top P_\zeta\right)
- \mu N_z M M^\top N_z-\kappa N_z.
\end{multline*}
Finally, 
by selecting 
$\kappa $ small enough so that
$$
N_z \geq \frac{\kappa}{\mu } P_\zeta (\mu M-K)(\mu M-K)^\top P_\zeta,
$$
we obtain
\eqref{app:ineq_PhiM}.\qed
\end{proof}

\begin{lemma}\label{lemma:transfer_zeta}
 Consider system 
\begin{equation}
\label{pf:subsys_zeta}
\dot \zeta = (\Phi- \mu MM^\top N_z)
\zeta 
- Mv
\end{equation}
with $\Phi$, $M, N_z$ be 
defined as in \eqref{eq:defPhiM} and 
$n_{zk}$
satisfying \eqref{eq:sequence_nzk}, e.g. as in \eqref{eq:nzk}. 
The transfer function between $v$ and  
$\N_z^\frac{1}{2} \zeta$ satisfies
\begin{equation}
\label{eq:transfer_function}
\bar{\zeta }(\im\omega) ^\top N_z \zeta (\im\omega) 
=
\dfrac{\displaystyle\sum_{\ell=0}^{n_o} \frac{\omega^2+\omega_\ell^2}
{(\omega_\ell^2-\omega^2)^2 } n_{z\ell}}
{1 + \displaystyle\mu^2 \left(\sum_{\ell=0}^{n_o} \frac{\omega}
{\omega_\ell^2-\omega^2 } n_{z\ell}\right)^2} \; |v(\im\omega)|^2.
\end{equation}
Furthermore, there exists $\kappa_0,\kappa_1>0$ independent of $n_o$
such that
\begin{equation}
\label{eq:transfer_function_bound}
\bar{\zeta }(\im\omega) ^\top N_z \zeta (\im\omega) \; \leq \;
( \kappa_0 + \kappa_1\omega^2)|v(\im\omega)|^2 \qquad \forall \, \omega\in \RR.
\end{equation}
\end{lemma}

\begin{proof}
Consider  the 
change of coordinates $\eta\mapsto y := N_z^{\frac12}\zeta$
giving
$$
\dot y = (\Phi -\mu N_z^{\frac12} M M^\top 
N_z^{\frac12})y - N_z^{\frac12}M v.
$$
By defining with $G(\im\omega)$
the transfer function between $v$ and $y$,
it is readily seen that it can be computed as
$$
G(\im\omega) = \left[\im\omega I - \Phi +\mu N_z^{\frac12} M M^\top N_z^{\frac12}\right]^{-1}N_z^{\frac12}M
$$
and its transpose conjugate given by 
$$
G^*(\im\omega) = M^\top N_z^{\frac12}\left[-\im\omega I + \Phi +\mu N_z^{\frac12} M M^\top N_z^{\frac12}\right]^{-1}
$$
where we used the fact that $\Phi^\top = - \Phi$.
By temporarily using the compact notation
$\Lambda := \im\omega I -\Phi$, and  
$\Upsilon := \mu^{\frac12} N_z^{\frac12}M$,
and by 
Woodbury matrix identity\footnote{Also known as 
matrix inversion lemma, the Woodbury matrix identity states that 
for $A$ invertible and $u,v$ column vectors, the following holds:
$(A+uv^\top)^{-1}= A^{-1} - \tfrac{1}{1+v^\top A^{-1}u}A^{-1}u v^\top A^{-1}$.} (recall that $\Upsilon $ is a vector), 
we obtain
$$
\begin{array}{rcl}
G^*(\im\omega)G(\im\omega)
& = & \dfrac{1}{\mu}
\Upsilon^\top(-\Lambda + \Upsilon\Upsilon^\top)^{-1}(\Lambda+\Upsilon\Upsilon^\top)^{-1}\Upsilon
\\
& = &\dfrac{1}{\mu}
 \Upsilon^\top\left[
 (-\Lambda)^{-1} - \dfrac{\Lambda^{-1}\Upsilon\Upsilon^\top \Lambda^{-1}}{1+\Upsilon^\top(-\Lambda^{-1})\Upsilon}
 \right]
 \left[
 \Lambda^{-1} - \dfrac{\Lambda^{-1}\Upsilon\Upsilon^\top \Lambda^{-1}}{1+\Upsilon^\top \Lambda^{-1}\Upsilon}
 \right]\Upsilon
 \\
 & = &\dfrac{1}{\mu}
  \dfrac{- \Upsilon^\top \Lambda}{1-\Upsilon^\top \Lambda^{-1}\Upsilon}\; \dfrac{\Lambda \Upsilon}{1+\Upsilon^\top \Lambda^{-1}\Upsilon}
\; = \; -\dfrac{1}{\mu}
 \dfrac{\Upsilon^\top \Lambda^{-2}\Upsilon}{1-(\Upsilon^\top \Lambda^{-1}\Upsilon)^2}
  \\
 & = & 
\dfrac{-M^\top N_z^{\frac12}(\im\omega I-\Phi)^{-2}N_z^{\frac12} M}
{1-\mu^2(M^\top N_z^{\frac12}(\im\omega I-\Phi)^{-1}N_z^{\frac12}M)^2 }.
\end{array}
$$
Now, note that we have
algebraically
$$
\begin{array}{l}
(\im\omega I-\Phi_\ell)^{-1} =
\begin{pmatrix}
\im\omega & -\omega_\ell
\\
\omega_\ell & \im\omega
\end{pmatrix}^{-1}
= 
\dfrac{1}{\omega_\ell^2-\omega^2}
\begin{pmatrix}
\im\omega & \omega_\ell
\\
-\omega_\ell & \im\omega
\end{pmatrix}
\\
\\
M_\ell ^\top N_{z\ell}^{\frac12} (\im\omega I-\Phi_\ell)^{-1}
 N_{z\ell}^{\frac12} M_\ell
=
\dfrac{\im\omega}{\omega_\ell^2-\omega^2} |M_\ell|^2 n_{z\ell}
\  .
\end{array}
$$
Moreover, 
$$
\begin{array}{l}
(\im\omega I-\Phi_\ell)^{-2} =
\begin{pmatrix}
\im\omega & -\omega_\ell
\\
\omega_\ell & \im\omega
\end{pmatrix}^{-2}
= 
\dfrac{-1}{(\omega_\ell^2-\omega^2)^2}
\begin{pmatrix}
\omega^2+\omega_\ell^2 & -2\im\omega\omega_\ell
\\
2\im\omega\omega_\ell & \omega^2+\omega_\ell^2
\end{pmatrix}
\\
\\
M_\ell ^\top N_{z\ell}^{\frac12} (\im\omega I-\Phi_\ell)^{-2}
 N_{z\ell}^{\frac12} M_\ell
=
-\dfrac{\omega^2+\omega_\ell^2}{(\omega_\ell^2-\omega^2)^2} |M_\ell|^2 n_{z\ell}
\  .
\end{array}
$$
This yields
$$
M_\ell ^\top N_{z\ell}^{\frac12} (\im\omega I-\Phi_\ell)^{-1}
 N_{z\ell}^{\frac12} M_\ell
\;=\; \im\sum_{\ell=0}^{n_o} \frac{\omega}
{\omega_\ell^2-\omega^2 }|M_\ell |^2 n_{z\ell}
$$
$$
M_\ell ^\top N_{z\ell}^{\frac12} (\im\omega I-\Phi_\ell)^{-2}
 N_{z\ell}^{\frac12} M_\ell
\;=\;- \sum_{\ell=0}^{n_o} \frac{\omega^2+\omega_\ell^2}
{(\omega_\ell^2-\omega^2)^2 }|M_\ell |^2 n_{z\ell}
$$
which finally gives 
the expression \eqref{eq:transfer_function}
in which we used also the definition $M_\ell=(1, 0)$
and the block diagonal form of the matrix $\Phi$.

%
%

We are left with proving inequality \eqref{eq:transfer_function_bound}.
By letting
$x= \frac{\omega }{\widehat\omega}$, and recalling 
\eqref{eq:basic_omega}, 
the fraction in \eqref{eq:transfer_function} reads
\begin{equation}
\label{eq:zeta_pk_simplified}
\mathcursive{T}(x)\;=\; \dfrac{\displaystyle\sum_{\ell=0}^{n_o} \frac{\omega^2+\omega_\ell^2}
{(\omega_\ell^2-\omega^2)^2 } n_{z\ell}}
{1 + \displaystyle\mu^2 \left(\sum_{\ell=0}^{n_o} \frac{\omega}
{\omega_\ell^2-\omega^2 } n_{z\ell}\right)^2} 
\;=\; 
\dfrac{\displaystyle
\sum_{\ell=0}^{n_o} 
\frac{\ell^2+x^2}{(\ell^2 - x^2)^2}n_{z\ell}
}
{\displaystyle
1 + \mu^2   x^2
\left(
\sum_{\ell=0}^{n_o}
 \frac{1}{\ell^2-x^2} n_{z\ell}
\right)^2
} \ .
\end{equation}
The rest of the proof follows by direct application of 
Lemma~\ref{lemma:transfer_function} given in the 
Supplementary Material at the end of this article.
\qed
\end{proof}

\begin{remark}
The expression  
\eqref{eq:transfer_function} can be also rewritten as 
\begin{equation}
\label{eq:transfer_function2}
 \dfrac{\bar{\zeta }(\im \omega) ^\top N_z \zeta (\im \omega)}{|v(\im \omega)|^2}  
=
\frac{\displaystyle 
\sum_{\ell=0}^{n_o}
\left(\omega _\ell ^2 + \omega^2\right)\,  n_{z\ell}
\left[\prod_{m=0,\neq \ell}^{n_o}\left(
\omega_m^2-\omega^2\right)\right]^2
}{\displaystyle 
\left[\prod_{m=0}^{n_o}\left(
\omega_m^2-\omega^2\right)\right]^2
+
\mu ^2 \omega^2 \left[\sum_{\ell=0}^{n_o}
n_{z\ell}
\prod_{m=0,\neq \ell}^{n_o}\left(
\omega_m^2-\omega^2\right)\right]^2
}  
\end{equation}
showing that there is actually no singularity at $\omega = \omega_\ell$, and we have in particular
$$
 \dfrac{\bar{\zeta }(\im \omega _\ell) ^\top N_z \zeta (\im \omega_\ell)}{|v(\im \omega_\ell)|^2}  
=
\frac{\displaystyle 2
}{\displaystyle 
\mu ^2 
n_{z\ell}
}   \ .
$$
\end{remark}

\begin{lemma}
\label{lemma:periodic_zeta}
Consider
again
system
\begin{equation}
\label{pf:subsys_zeta2}
\dot \zeta = (\Phi- \mu MM^\top N_z)\zeta - Mv(t)
\end{equation}
with $v\in \C^1_T(\RR)$ and $\Phi$, $M, N_z$ be 
defined as in \eqref{eq:defPhiM}  and $n_{zk}$ 
satisfying \eqref{eq:sequence_nzk}, e.g. as in \eqref{eq:nzk}.
For any $n_o\in \NN$, it has a unique periodic solution $\zeta_p$ satisfying
\begin{equation}
\label{pf:L2L2gain_zeta}
\mu  ^2\,  \int_0^T |M^\top N_z\zeta _p (t)| ^2 dt \leq   \int_0^T|v(t)|^2 dt,
\end{equation}
\begin{equation}
\label{pf:L2_zeta}
\sup_{t\in[0,T]}
\zeta_p(t)^\top N_z \zeta_p(t)
\; \leq  \;
 \left(\dfrac{1}\mu +  \dfrac{\kappa_0}{T} \right) \int_0^T|v(t)|^2 dt 
 + \dfrac{\kappa_1}{T}\int_0^T|\dot v(t)|^2 dt 
\end{equation}
with $\kappa_0,\kappa_1$ given by Lemma~\ref{lemma:transfer_zeta}.
\end{lemma}

\begin{proof}
According to Lemma \ref{lemma:hurwitzPhiM}, 
the matrix $\Phi- \mu MM^\top N_z$ is Hurwitz. 
Hence (see Lemma~\ref{lemma:periodic-solution}) 
system \eqref{pf:subsys_zeta2} admits a unique periodic solution 
given by 
$$
\zeta_p(t)= \Psi(\Phi- \mu MM^\top N_z, -Mv(t))
$$
with $\Psi$ defined below in \eqref{3}.
Then, compute
\begin{equation}\label{pf_lemma_eq0}
\begin{array}{rcl}
\frac{1}{2}\dot{\overparen{\zeta _p (t)^\top N_z \zeta _p (t)}}
 &=&\zeta _p (t)^\top N_z [\Phi- \mu MM^\top N_z] \zeta _p (t) - \zeta _p (t)^\top N_z M v(t)
\\[.5em]
& =& -\mu |M^\top N_z\zeta _p (t)| ^2 - \zeta _p (t)^\top N_z M v(t)
\\[.5em]
& \leq  &
-\frac{\mu}{2}|M^\top N_z\zeta _p (t)| ^2+
 \frac{1}{2\mu } |v(t)|^2.
\end{array}
\end{equation}
By integrating and using periodicity, this yields
to \eqref{pf:L2L2gain_zeta}. 

In order to show the second
inequality
of the 
statement of the lemma, the function $v$ being in $\C^1_T(\RR)$ can be expressed as the sum
of its Fourier series as
$$
v(t) =  \sum_{k\in \ZZ}v_k \exp(\im k\tfrac{2\pi}{T} t).
$$
By using the Bessel-Parseval identity
and continuity of the derivative $\dot v$, we have
\begin{eqnarray}\label{pf_lemma_eq1}
\left[{\textstyle\frac{2\pi}{T}}\right]^2 \sum_{k=0}^\infty  k^2 |v_k|^2 
\leq \dfrac{1}{T} \int_0^T |\dot v(t)|^2 dt.
\end{eqnarray}
We have similarly
\begin{equation}\label{pf_lemma_eq2}
\dfrac{1}{T}\int_{0}^T \zeta_p(t)^\top N_z \zeta_p(t) dt =
\sum_{k=0}^\infty
\bar{\zeta }_{pk} ^\top N_z \zeta _{pk},
\quad
\zeta_{pk} = \dfrac{1}{T}\int_{0}^{T}
\zeta_p(t)\exp(\im k\tfrac{2\pi}{T} t)dt .
\end{equation}
Now, by applying the expression of the transfer function
 \eqref{eq:transfer_function} to each element 
 of the Fourier series of the product
 $\zeta_p(t)^\top N_z \zeta_p(t)$, we  obtain 
 $$
\bar{\zeta }_{pk} ^\top N_z \zeta _{pk}  = 
\dfrac{\displaystyle\sum_{\ell=0}^{n_o} 
\frac{\omega_\ell^2+\left[{\textstyle\frac{2\pi}{T}}\right]^2 k^2}
{(\omega_\ell^2-\left[{\textstyle\frac{2\pi}{T}}\right]^2 k^2)^2 } n_{z\ell}}
{1 + \displaystyle\mu^2 \left(\sum_{\ell=0}^{n_o} \frac{\left[{\textstyle\frac{2\pi}{T}}\right]^2 k^2}
{\omega_\ell^2-\left[{\textstyle\frac{2\pi}{T}}\right]^2 k^2 } n_{z\ell}\right)^2} \; |v_k|^2.
 $$
As a consequence, by using inequality
\eqref{eq:transfer_function_bound} on previous
identity, 
and by using again equations 
\eqref{pf_lemma_eq1} and \eqref{pf_lemma_eq2},
we further obtain
\begin{align}
\sum_{k=0}^\infty \notag
\bar{\zeta }_{pk} ^\top N_z \zeta _{pk} 
 &\leq  \sum_{k=0}^\infty
\left(\kappa_0 + \kappa_1 \left[{\textstyle\frac{2\pi}{T}}\right]^2 k^2\right) |v_{k}|^2
\\ 
& \leq  \dfrac{\kappa_0}{T} \int_0^T |v(s)|^2ds
+ \dfrac{\kappa_1}{T}\int_0^T |\dot v(s)|^2ds \ .
\label{pf_lemma_eq3}
\end{align}
We note also that, as a consequence of \eqref{pf_lemma_eq0},
we have
$$
\displaystyle \sup_{t\in [0,T]}
\zeta _p(t)^\top N_z \zeta _p(t) 
\;-\; 
\displaystyle \inf_{t\in [0,T]}
\zeta _p(t)^\top N_z \zeta _p(t) 
\; \leq \; \frac{1}{\mu }\,  \int_0^T |v(s)|^2 ds
\  .
$$
On another hand, we have
$$
\displaystyle \inf_{t\in [0,T]}\zeta _p(t)^\top N_z \zeta _p(t)
\; \leq \;  \frac{1}{T}\,  \int_{0}^{T} \zeta _p(t)^\top N_z \zeta _p(t) ds
\  .
$$
This yields
\begin{equation}
\label{pf_lemma_eq4}
\displaystyle \sup_{t\in [0,T]}
\zeta _p(t)^\top N_z \zeta _p(t) 
\; \leq \; 
\frac{1}{T}\,  \int_{0}^{T} \zeta _p(t)^\top N_z \zeta _p(t) ds
\;+\; 
\frac{1}{\mu }\,  \int_0^T |v(s)|^2 ds
\  ,
\end{equation}
Finally, 
by combining \eqref{pf_lemma_eq4}
with  \eqref{pf_lemma_eq2} and \eqref{pf_lemma_eq3},
we obtain
\eqref{pf:L2_zeta}. \qed
\end{proof}

\subsection{Existence of a Periodic Solution}

\begin{proposition}\label{prop:periodic_solution}
For any triplet $(\overline P_x,\underline P_x,\alpha)$,
there exist
strictly positive
real numbers  $\boundx_p, \bounde_p,\boundzeta_p,\boundr_x, \boundr_e,\sigma ^\star_p>0$
(independent of $n_o$)
 such that, for any $n_o>0$, any $\sigma>\sigma ^\star_p$,
and $\mu\geq1$, system 
\eqref{eq:closed-loop},
with $f\in \F(\boundx_p, \bounde_p,\overline P_x,\underline P_x,\alpha)$ and $q\in\Q(\boundx_p, \bounde_p)$,
admits 
a $T$-periodic solution 
$(x_p, e_p, \zeta_p) \in \C^2_T([0,T];\RR^n\times \RR \times \RR^{2n_o+1})$
 satisfying, for all $t\in[0,T]$, 
 the following inequalities
\begin{subequations}\label{eq:prop_periodic_bounds_xezeta}
\begin{align}\label{eq:prop_periodic_bound_x}
|x_p(t)|\leq & \min\left\{ \boundx_p,  \frac{\boundr_x}{\sigma}
\right\} \ ,
\\ \label{eq:prop_periodic_bound_e}
|e_p(t)|\leq & \min\left\{ \bounde_p,  \frac{\boundr_e}{\sigma} \right\} \ ,
\\ \label{eq:prop_periodic_bound_zeta}
\sqrt{\zeta_p(t)^\top N_z \zeta_p(t)} \leq  & \boundzeta_p \ .
\end{align}
\end{subequations} 
 \end{proposition}
 
\begin{proof}
Here we exploit the ability of expressing the system
 \eqref{eq:closed-loop} as in \eqref{2} and \eqref{1}.
With the notations
\begin{equation}
\label{eq:f_decomposition}
\delta f(t,x,e):= f(t,x,e) - F(t)x, 
\qquad
F(t) :=  \frac{\partial f}{\partial x}(t,0,0),
\end{equation}
system 
\eqref{eq:closed-loop} reads as follows when $(x^+,e^+)=(x^-,e^-)$
\begin{eqnarray}
\label{4}
\dot x ^+& = &F(t)x^+ + \delta f(t,x^-,e^-)
\\\label{5}
\dot e^+ & = &- \sigma e^+ \;+\;  \left[q(t,x^-,e^-) + \eta  \right]
,
\end{eqnarray}
where
\begin{equation}
\label{6}
\dot \zeta  \;=\;  (\Phi - \mu M M^\top N_z)\zeta  - Mq(t,x^-,e^-)
\quad ,\qquad  \eta \;=\; \mu M^\top N_z \zeta 
\end{equation}
Written this way, we see that we have a mapping from the input functions $(x^-,e^-)$ to the functions
$(x^+,e^+)$, solution of (\ref{4}), (\ref{5}), with the intermediate function $\zeta $ solution of (\ref{6}).

Each subsystem above can be compactly written as~:
\begin{equation}
\label{7}
\dot \chi = F(t)\chi + g(t)
\end{equation}
where $F$ is $T$-periodic. We have the following very standard result.
See, e.g., \cite[Lemma 5.1]{Hale92} or \cite[Chapter 8.2]{miller1982ordinary}.

\begin{lemma}\label{lemma:periodic-solution}
Consider system  \eqref{7} with $g$ is in $\C^0_T(\RR^n)$.
If the matrix $[I - \phi_F(t,t-T)]$ is invertible with $\phi_F(t,s)$ denoting the state 
transition matrix  of $F$, system \eqref{7} admits a unique periodic solution $\chi _p$
which can be expressed as
\begin{equation}
\label{3}
\chi_p(t) = \Psi(F(t),g(t)) : =   [I - \phi_F(t,t-T)]^{-1}
\int_{t-T}^{t}\phi_F(t,s)g(s)ds.
\end{equation}
Furthermore, if there exists a $T$-periodic
function $P_\chi$
satisfying
  \begin{eqnarray}
0  \; < \; \underline P_\chi  I \; \leq \;  P_\chi (t) \leq 
\;  \overline P_\chi  I \,,
\label{app:Pass1}
\\[.5em]
\label{app:Pass2}
\dot P_\chi (t) + P_\chi (t) F(t)+ 
 F^\top(t) P_\chi (t) \; \leq  \;  - 2\alpha P_\chi \,,
  \end{eqnarray}
then we have
\begin{align}\label{8}
\sup_{t\in[0,T]}|\Psi(F(t),g(t))|
& \leq  \dfrac{K_\alpha}\alpha \sup_{t\in[0,T]} |g(t)|
\\ \label{11}
\sup_{t\in[0,T]}|\Psi(F(t),g(t))|
& \leq \dfrac{K_{\alpha,2}}{\sqrt{\alpha}}
\sup_{t\in[0,T]} \sqrt{\int_0^T g(t)^2dt}
\end{align}
where 
\begin{eqnarray}\label{app:Kalpha1}
K_\alpha
&=& \sqrt{\frac{\overline{P}_\chi  }{\underline{P}_\chi }}
 \left(\exp(\alpha T)-1+\sqrt{\frac{\overline{P}_\chi  }{\underline{P}_\chi }} \right)\exp(-\alpha T),
\\
 K_{\alpha,2}
&=& K_\alpha 
 \sqrt{\dfrac{T(1+\exp(-\alpha T))}{2(1-\exp(-\alpha T))}}.
 \label{app:Kalpha2}
\end{eqnarray}
\end{lemma}

Being interested in periodic solutions for (\ref{4}), (\ref{5}) and (\ref{6}), this leads to the consideration of the following operators
\begin{eqnarray}\label{eq:operator_x}
\displaystyle 
\operator{ x^-,e^-}_x(t)&=& 
\Psi\left(
\vrule height 0.5em depth 0.5em width 0pt
F(t)\,  ,\,  \delta{f}(t, x^-(t),e^-(t))\right),
\\ \label{9}
\displaystyle
\operator{x^-,e^-}_e(t)&=& \Psi\left(
\vrule height 0.5em depth 0.5em width 0pt
-\sigma \,  ,\,  q(t, x^-(t), e^-(t))+\eta [x^-,e^-](t)\right) ,
\end{eqnarray}
where~:
\begin{equation}
\label{10}
\eta [x^-,e^-](t)\;=\; \mu M^\top N_z \Psi\left(\Phi-\mu MM^\top N_z, -Mq(t,x^-(t),e^-(t)\right).
\end{equation}
When $(x^-,e^-)$ is $T$-periodic, $(\operator{ x^-,e^-}_x,\operator{x^-,e^-}_e)$ is the unique $T$-periodic 
solution of (\ref{4}), (\ref{5}). So to establish our result it is sufficient to show that there exists a 
$T$-periodic function $(x^-,e^-)$ 
satisfying
$$
(\operator{ x^-,e^-}_x,\operator{x^-,e^-}_e)\;=\; (x^-,e^-).
$$
Our next step is,
%
omitting the superscript $^-$ to lighten the notations, to show that the operator
  $(x,e) \mapsto (\operator{ x,e}_x,\operator{ x,e}_e)$
  is a contraction on the set of $T$-periodic functions satisfying
\begin{equation}
\label{eq:temp_bounds_xe}
\sup_{t\in[0,T]}|x(t)|\leq \boundx_p, 
\quad
\sup_{t\in[0,T]}|e(t)|\leq \bounde_p,
\end{equation}
when the bounds
$\boundx_p, \bounde_p$
are chosen small enough, and $\sigma$ is chosen large enough, this independently of $n_o$.

To this end, by recalling the definitions given in 
Section~\ref{sec:notation}, we list 
inequalities for the functions $q$ and $f$, obtained as consequences of the following fact\\[0.5em]
\textit{For any $C^1$ function $\varphi:\RR\times \RR^n \times \RR\to \RR^m$
we have
\begin{align} \notag
\left|\varphi(t,x_a,e)-\varphi(t,x_b,e)\right|
&= \displaystyle
 \left|\left(\int_0^1 \frac{\partial \varphi}{\partial x}(t,x_b+s(x_a-x_b),e)) 
ds\right) (x_a-x_b) \right|
\\
&\leq  \displaystyle
\sup_{(t,x,e)\in\mathcal{S}_T(\boundx,\bounde)} 
\left|\dfrac{\partial \varphi}{\partial x}(t,x,e) \right| 
 |x_a - x_b|
\label{app:ineq_varphi}
\end{align}
for all $x_a,x_b \in \mathcal{S}_x(\boundx)$, 
$e\in \mathcal{S}_e(\bounde)$,
and all $t\in [0,T]$.}

First, we have
\begin{eqnarray}
\label{app:bound_q}
|q(t, x,e)| & \leq & \boundq (\boundx,\bounde)
\\ 
\label{app:bound_q_ab}
| q(t, x_a,e_a)
-
 q(t, x_b,e_b)| 
 &\leq&
 \boundq_e(\boundx,\bounde)\,  |e_a-e_b|
+ 
\boundq_x(\boundx,\bounde)\,  | x_a-  x_b|
\\
\label{app:bound_f_ab}
\left|f(t,x,e_a)-f(t,x,e_b)\right|
&\leq &
\boundf_e(\boundx,\bounde)\,  |e_a-e_b|
\end{eqnarray}
for all $x\in\mathcal{S}_x(\boundx)$, $e\in\mathcal{S}_e(\bounde)$,
$(x_a,e_a)$ and $(x_b,e_b)$ in  $\mathcal{S}(\boundx,\bounde)$, and all $t\in [0,T]$.
Furthermore, by using the  definitions of 
$F$ and $\delta f$ given in \eqref{eq:f_decomposition},  
we obtain
$$
\begin{array}{rcll}
\left|
f(t,x,0)-F(t)  x
\right|
& =&
\displaystyle 
\left|\int_0^1
\left[\left(\frac{\partial f}{\partial x}(t,s x,0)-
\frac{\partial f}{\partial x}(t,0,0))\right)ds\right]  x
\right|
\\
& =&
\displaystyle 
\left|\int_0^1
\left[
\int_0^s
\left[\frac{\partial ^2 f}{\partial x\partial x}(t,r x,0)
dr\right]  x ds\right] x
\right| &
\; \leq \;
\displaystyle 
\frac{1}{2}\boundf_{xx}(\boundx,0)\,  | x|^2.
\end{array}
$$
By combining the previous bound with 
\eqref{app:bound_f_ab} in which $e_a= e$, $e_b=0$, we get
for $\delta f$ defined in \eqref{eq:f_decomposition},
\begin{equation}\label{app:bound_deltaf}
|\delta f(t, x,e)|\; \leq \; \boundf_e(\boundx,\bounde)\,  |e|
\;+\; 
\frac{1}{2}\boundf_{xx}(\boundx,0)\,  | x|^2
\end{equation}
for all 
$(x,e)\in  \mathcal{S}(\boundx,\bounde)$
and all $t\in [0,T]$.
With similar computations, we also obtain
\begin{multline} 
 |\delta f(t, x_a,e_a)-\delta f(t, x_b,e_b)| \leq
 \boundf_e(\boundx,\bounde)\,  |e_a-e_b|
+
\frac{1}{2}\boundf_{xx}(\boundx,0) | x_a- x_b|^2
\\
+ \left[ \boundf_{ex}(\boundx,\bounde) \bounde
+
\boundf_{xx}(\boundx,0) \boundx\right]
 | x_a- x_b|
 \label{app:bound_delta_f_ab}
\end{multline}
for all $(x_a,e_a), (x_b,e_b)\in  \mathcal{S}(\boundx,\bounde)$
and all $t\in [0,T]$.

Now, by using the definition of 
$\operator{x,e}_x$ in \eqref{eq:operator_x}, and the assumptions
\eqref{eq:Pass1} and \eqref{eq:Pass2}, inequality \eqref{8} gives
\begin{align} \notag
\sup_{t\in [0,T]}|\operator{x,e}_x(t)| &\leq \dfrac{K_\alpha}{\alpha}
\sup_{t\in [0,T]}|\delta{f}(t, x(t),e(t))|
\\
&\leq 
\frac{K_\alpha }{\alpha }\!\left[ \boundf_e(\boundx_p,\bounde_p) \bounde_p
+
\frac{1}{2}\boundf_{xx}(\boundx_p,0)\boundx_p^2\right],
\label{eq:bound_op_x}
\end{align}
in which we used \eqref{eq:temp_bounds_xe}
and \eqref{app:bound_deltaf} to derive the second inequality.
Then, note that by definition of \eqref{eq:operator_x}, 
and by exploiting the linearity of the function 
$\Psi$ defined in
\eqref{3},
we have
$$
\begin{array}{rcl}
\operator{ x_a,e_a}_x(t)-\operator{ x_b,e_b}_x(t)
& = &
\Psi\big(F,\delta{f}(t, x_a(t),e_a(t))\big)
- 
\Psi\big(F,\delta{f}(t, x_b(t),e_b(t))\big)
\\
& =&
\Psi\big(F,\delta{f}(t, x_a(t),e_a(t))- \delta{f}(t, x_b(t),e_b(t))\big),
\end{array}
$$
and therefore, by using  \eqref{8}, \eqref{app:bound_delta_f_ab}, 
and 
$\sup_{s\in [0,T]} | x_a(s)- x_b(s)| \leq  2 \boundx_p$, we obtain
\begin{multline}
\label{eq:bound_op_x_ab}
|\operator{ x_a,e_a}_x(t)-\operator{ x_b,e_b}_x(t)|  \leq
\\   
\displaystyle 
\dfrac{K_\alpha}{\alpha }\left[ \boundf_{ex}(\boundx_p,\bounde_p) \bounde_p
+2
\boundf_{xx}(\boundx_p,0) \boundx_p\right]\!\!
 \sup_{s\in [0,T]} | x_a(s)- x_b(s)|
  + 
\dfrac{K_\phi }{\alpha } \boundf_e(\boundx_p, \bounde_p) 
\sup_{s\in[0,T]}
|e_a(s) - e_b(s)|.
\end{multline}
Similarly for 
$\operator{ x,e}_e$ in \eqref{9}, with
$$
\operator{ x,e}_e(t)  \;=\; \Psi \big(-\sigma, q(t,x,e)\big)
+  \Psi\big(-\sigma, \eta  [x,e](t)\big)
$$
and \eqref{11} where
$$
\underline P_\chi = \overline P_\chi  = 1
 ,\qquad 
\alpha =\sigma 
 ,\qquad 
K_\alpha  =1 ,\qquad 
K_{\alpha ,2}=
 \sqrt{\dfrac{T(1+\exp(-\sigma  T))}{2(1-\exp(-\sigma  T))}},
$$
we obtain
$$
\sup_{t\in [0,T]}|\operator{x,e}_e(t)| 
 \leq 
 \dfrac{K_{\alpha,2}}{\sqrt{\sigma}} 
\left[
\sqrt{\int_0^T|q(t,x(t),e(t))|^2 dt}+
\sqrt{\int_0^T|\eta [x,e](t)|^2 dt}
\right].
$$
Then 
\eqref{pf:L2L2gain_zeta}
gives 
\begin{equation*}
\sup_{t\in [0,T]}|\operator{x,e}_e(t)| 
 \leq 
 \dfrac{2K_{\alpha,2}}{\sqrt{\sigma}} 
\sqrt{\int_0^T|q(t,x(t),e(t))|^2 dt}.
\end{equation*}
By using
$\int_0^T|q(t)|^2 dt\leq T \sup_{t\in[0,T]}|q(t)|^2$
and bound
  \eqref{app:bound_q}, 
we finally obtain 
\begin{equation} 
\sup_{t\in [0,T]}|\operator{x,e}_e(t)| 
\leq  \beta(\sigma)
\boundq(\boundx_p,\bounde_p).
\label{eq:bound_op_e}
\end{equation}
in which  $\beta$ is defined as
\begin{equation}
 \beta(\sigma) :=   
\sqrt{\frac{2T(1+\exp(-\sigma T))}{\sigma(1-\exp(-\sigma T))}}.
\label{eq:beta_definition}
\end{equation}
Note that $\beta$ is a continuous strictly decreasing positive function 
for $\sigma>0$. In particular, $\lim_{\sigma\to\infty}\beta(\sigma)= 0$.

 Finally, as done to obtain inequality \eqref{eq:bound_op_x_ab}, 
 we can use linearity of the operator $\Psi$ and 
  inequalities \eqref{eq:bound_op_e} and
    \eqref{app:bound_q_ab}
 to derive
  \begin{multline}
\sup_{t\in [0,T]}|\operator{x_a,e_a}_e(t)-
\operator{x_b,e_b}_e(t)| 
 \leq  \\[-.5em] \beta(\sigma)
\left( \boundq_e(\boundx_p,\bounde_p)\,  \sup_{s\in [0,T]} | e_a(s)- e_b(s)|
+ 
\boundq_x(\boundx_p,\bounde_p)\,   \sup_{s\in [0,T]} | x_a(s)- x_b(s)|\right).
\label{eq:bound_op_e_ab}
\end{multline}
In conclusion, with
\eqref{eq:bound_op_x} and \eqref{eq:bound_op_e}, we have established
that, if $\boundx_p$, $\bounde_p$ 
and $\sigma$ satisfy
\begin{subequations}\label{eq:condition_prop_1_bound}
\begin{eqnarray}
\label{eq:condition_prop_1_bound_xpep}
\dfrac{K_\alpha }{\alpha }\left[ \boundf_e(\boundx_p,\bounde_p)\,  \bounde_p
\;+\; 
\frac{1}{2}\boundf_{xx}(\boundx_p,0)\,\boundx_p^2\right] & \leq & \boundx_p
\  ,
\\ 
\label{eq:condition_prop_1_bound_sigma}
 \beta(\sigma)
\boundq (\boundx_p,\bounde_p)
& \leq & \bounde_p \ ,
\end{eqnarray}
\end{subequations}
with $\beta$ defined in \eqref{eq:beta_definition},
and 
if $(x,e)$ are $T$-periodic continuous functions satisfying \eqref{eq:temp_bounds_xe},
then $(\operator{ x,e}_x,\operator{ x,e}_e)$ are $T$-periodic continuous functions satisfying
\begin{equation}
\label{eq:long_bound1}
\sup_{t\in [0,T]}|\operator{ x,e}_x(t)|\; \leq \; \boundx_p,
\quad
\sup_{t\in [0,T]}|\operator{ x,e}_e(t)|\; \leq \; \bounde_p.
\end{equation}
Similarly, with (\ref{eq:bound_op_x_ab}) and (\ref{eq:bound_op_e_ab}), we have established
that, if $(x_a,e_a)$ and $(x_b,e_b)$ are $T$-periodic continuous functions satisfying (\ref{eq:temp_bounds_xe}),
$$
\left(\begin{array}{c}
|\operator{ x_a,e_a}_x(t)-\operator{ x_b,e_b}_x(t)|
\\
|\operator{ x_a,e_a}_e(t)-\operator{ x_b,e_b}_e(t)|
\end{array}\right)
\; \leq \; \mathfrak{M} (\boundx_p,\bounde_p, \sigma )
\left(\begin{array}{c}
\sup_{s\in[0,T]}
|x_a(s) - x_b(s)|
\\
\sup_{s\in[0,T]}
|e_a(s) - e_b(s)|
\end{array}\right)
$$
with the notation
\begin{equation}
\mathfrak{M} (\boundx_p,\bounde_p, \sigma )\; := \;
\begin{pmatrix}
\ypsilon_{xx} (\boundx_p,\bounde_p, \sigma )& \ypsilon_{xe}(\boundx_p,\bounde_p, \sigma )
\\
\ypsilon_{ex} (\boundx_p,\bounde_p, \sigma )& \ypsilon_{ee}(\boundx_p,\bounde_p, \sigma )
\end{pmatrix}
\label{eq:def_ypsilon}
\end{equation}
where
$\mathfrak{M} $ is the following matrix with strictly positive entries
$$
\renewcommand{\arraystretch}{1.5}
\begin{array}{@{}r@{\,  }c@{\,  }l@{\:  ,\ }r@{\,  }c@{\,  }l@{}}
\ypsilon_{xx} (\boundx_p,\bounde_p, \sigma )& =& \tfrac{K_\phi }{\alpha }\left[ \boundf_{ex}(\boundx_p,\bounde_p)\,  \bounde_p 
+
2\boundf_{xx}(\boundx_p,0)\,\boundx_p\right]
& \qquad
\ypsilon_{xe} (\boundx_p,\bounde_p, \sigma )& = & \tfrac{K_\phi }{\alpha }\boundf_e(\boundx_p,\bounde_p) 
\\
\ypsilon_{ex} (\boundx_p,\bounde_p, \sigma )& = & \beta(\sigma)\boundq_x(\boundx_p,\bounde_p)
& \qquad
\ypsilon_{ee} (\boundx_p,\bounde_p, \sigma )&= & \beta(\sigma)\boundq_e(\boundx_p,\bounde_p)
\end{array}
$$
It follows from Perron-Frobenius theorem that there exist
strictly positive real numbers $p_x(\boundx_p,\bounde_p, \sigma )$, $p_e(\boundx_p,\bounde_p, \sigma )$ and $\gamma(\boundx_p,\bounde_p, \sigma )$, depending on $(\boundx_p,\bounde_p, \sigma )$ and satisfying
\begin{equation}
\label{eq:def_Upsilon_prop_1}
\mathfrak{M}^\top
\begin{pmatrix}
p_x
\\
p_e
\end{pmatrix} = \gamma\,   \begin{pmatrix}
p_x
\\
p_e
\end{pmatrix} .
\end{equation}
Here, $\gamma $ is a simple eigenvalue of $\mathfrak{M}$ and the spectral radius of this matrix. It
 is strictly smaller than $1$ if and only if we have 
\begin{eqnarray*}
1 - [\ypsilon _{xx} + \ypsilon _{ee}] + [\ypsilon _{xx}\ypsilon _{ee}-\ypsilon _{xe}\ypsilon _{ex}]&>&0,
\\
1-[\ypsilon _{xx}\ypsilon _{ee}-\ypsilon _{xe}\ypsilon _{ex}]&>&0
\  ,
\end{eqnarray*}
i.e. 
\begin{subequations}\label{eq:contraction_condition}
\\[1em]
$\displaystyle 
1\; >\; 
\tfrac{K_\alpha }{\alpha }\left[ \boundf_{ex}(\boundx_p,\bounde_p)\,  \bounde_p 
+
2\boundf_{xx}(\boundx_p,0)\,\boundx_p\right]
+
\beta(\sigma)\boundq_e(\boundx_p,\bounde_p)
$\refstepcounter{equation}\label{eq:ypsilon_ex}\hfill$(\theequation)$
\\\null\hfill$\displaystyle
+\tfrac{K_\alpha\beta(\sigma) }{\alpha }
\left(\boundf_e(\boundx_p,\bounde_p)
\boundq_x(\boundx_p,\bounde_p)
-
\left[ \boundf_{ex}(\boundx_p,\bounde_p)\,  \bounde_p 
+
2\boundf_{xx}(\boundx_p,0)\,\boundx_p\right]
\boundq_e(\boundx_p,\bounde_p)
\right)
\  ,
$\\[1em]$
1\; >\; 
\tfrac{K_\alpha\beta(\sigma) }{\alpha }
\left(\boundf_e(\boundx_p,\bounde_p)
\boundq_x(\boundx_p,\bounde_p)
-
\left[ \boundf_{ex}(\boundx_p,\bounde_p)\,  \bounde_p 
+
2\boundf_{xx}(\boundx_p,0)\,\boundx_p\right]
\boundq_e(\boundx_p,\bounde_p)
\right)
$\refstepcounter{equation}\label{eq:ypsilon_xx}\hfill$(\theequation)$
\\[1em]
\end{subequations}
With this at hand, \eqref{eq:bound_op_x_ab},
 \eqref{eq:bound_op_e_ab} and \eqref{eq:def_Upsilon_prop_1} give
\begin{multline}
p_x(\boundx_p,\bounde_p,\sigma )\,  
\sup_{t\in [0,T]}|\operator{ x_a,e_a}_x(t) - \operator{ x_b,e_b}_x(t)|
 + 
p_e(\boundx_p,\bounde_p,\sigma )\,  
\sup_{t\in [0,T]}|\operator{ x_a,e_a}_e(t)-\operator{ x_b,e_b}_e(t)| 
 \\ \displaystyle 
 \leq
   \gamma (\boundx_p,\bounde_p,\sigma )
\left[p_x(\boundx_p,\bounde_p,\sigma )
\sup_{t\in [0,T]} | x_a(t)- x_b(t)|
+
p_e(\boundx_p,\bounde_p,\sigma )
\sup_{t\in [0,T]}|e_a(t) - e_b(t)|
\right].\label{eq:contraction}
\end{multline}

Now, let $\cB_{\boundx_p,\bounde_p}(\RR^{n+1})$ denote the closed subset of
$\C_T^0(\RR^{n+1})$
 defined as
$$
\cB_{\boundx_p,\bounde_p}(\RR^{n+1})
:=\Big\{(x,e)\in
\C^0_T(\RR^{n+1})
  : 
\sup_{t\in [0,T]} |x(t)| \leq 
\boundx_p ,  
\sup_{t\in [0,T]} |e(t)| \leq 
\bounde_p \Big\}.
$$
This set 
$\cB_{\boundx_p,\bounde_p}(\RR^{n+1})$ 
equipped with the norm
$$
\|(x,e)\|  := 
p_x(\boundx_p, \bounde_p, \sigma)
\sup_{t\in[0,T]} |x(t)|
+
p_e(\boundx_p, \bounde_p, \sigma)
\sup_{t\in[0,T]} |e(t)|
$$ 
is a complete metric space.
Assuming for the time being (see below) there exists a
triple $(\boundx_p,\bounde_p,\sigma ^\star_p)$ satisfying
\eqref{eq:condition_prop_1_bound}
and \eqref{eq:contraction_condition}, we
have established that, 
for any $\sigma>\sigma ^\star_p$, we have
the following properties.
\begin{enumerate}
\item The function $( x,e)\mapsto (\operator{ 
x,e}_x,\operator{ x,e}_e)$ 
maps a function in
$\cB_{\boundx_p,\bounde_p}(\RR^{n+1})$ into a function in 
$\cB_{\boundx_p,\bounde_p}(\RR^{n+1})$,
since  \eqref{eq:condition_prop_1_bound} implies \eqref{eq:long_bound1}.
\item The function $( x,e)\mapsto (\operator{ 
x,e}_x,\operator{ x,e}_e)$ 
is a contraction,
the gain
 $\gamma(\boundx_p,\bounde_p,\sigma )$ in \eqref{eq:contraction}
being strictly smaller than $1$ when \eqref{eq:contraction_condition} holds.
\end{enumerate}
We conclude,
from the Banach fixed point theorem, that
 there 
exists a fixed point $({x}_{p},e_p)$ in 
$\cB_{\boundx_p,\bounde_p}(\RR^{n+1})$, namely
there exist $x_p$ and $e_p$ satisfying
$\operator{{ x}_{p},e_p}_x= { x}_{p}$
and $\operator{{ x}_{p},e_p}_e= e_{p}$.
In particular, $x_p,e_p$ are 
$C^0$,
 $T$-periodic, satisfy
\begin{equation}
\label{eq:long_bound_xep}
\sup_{t\in [0,T]}|x_p(t)|\; \leq \; \boundx_p,
\quad
\sup_{t\in [0,T]}|e_p(t)|\; \leq \; \bounde_p,
\end{equation}
and are solution of 
\begin{subequations}
\begin{align}\label{eq:periodic_xp}
\dot { x}_{p}=&F(t) {x}_{p}
 +   \delta f(t,x_p,e_p)
\  ,
\\ \label{eq:periodic_ep}
\dot e_p  = &
- \sigma  e_p + \mu M^\top N_z \zeta_p(t)  + q(t,x_p,e_p)
\end{align}
where $\zeta _p$ is the unique $T$-periodic solution of
\begin{equation}\label{eq:periodic_zetap}
\dot \zeta _p=
[\Phi - \mu M M^\top N_z]\zeta _p
- M q(t,x_p(t),e_p(t)).
\end{equation}
\end{subequations}
as stated by Lemma~\ref{lemma:periodic_zeta}.
Also because of $F$, $\delta f$ and $q$ are $C^1$, the periodic solution
$(x_p,e_p)$ 
is $C^2$.

In order to determine a bound for $\boundzeta_p$, 
we   apply  Lemma~\ref{lemma:periodic_zeta}
and inequality \eqref{pf:L2_zeta}
to the solution $\zeta_p$ defined in 
\eqref{eq:periodic_zetap}. 
To this end, we need
a bound for $\dot q$.
It can be expressed as 
\begin{multline*}
\dot q(t,x_p(t),e_p(t))=
 \frac{\partial q}{\partial t}(t,x_p(t),e_p(t))
 +
 \frac{\partial q}{\partial x}(t,x_p(t),e_p(t))f(t,x_p(t),e_p(t))
 \\ 
 +\frac{\partial q}{\partial e}(t,x_p(t),e_p(t))\left[-\sigma e_p(t) + \mu  M^\top N_z \zeta _p(t)+ q(t,x_p(t),e_p(t))\right] .
\end{multline*}
As a consequence, the previous expression of 
$\dot q$  yields, using \eqref{pf:L2L2gain_zeta},
and the bounds  \eqref{eq:long_bound_xep}
for $x_p, e_p$,
$$
\begin{array}{rl} \displaystyle
\int_0^T|\dot q_p(t)|^2 dt \leq   
& 
 3T \boundq_t(\boundx_p,\bounde_p)^2 
+ 3T\boundq_x(\boundx_p,\bounde_p)^2\boundf(\boundx_p,\bounde_p)^2
\\ 
& \displaystyle +
9\boundq_e(\boundx_p,\bounde_p)^2\int_0^T \left[\sigma ^2 e_p(t)^2 + \mu^2  |M^\top N_z \zeta _p(t)|^2+ 
|q(t,x_p(t),e_p(t))|^2\right] dt
\\
\leq  & 3T \boundq_t(\boundx_p,\bounde_p)^2
+3T\boundq_x(\boundx_p,\bounde_p)^2\boundf(\boundx_p,\bounde_p)^2
\\
& \displaystyle
+9\boundq_e(\boundx_p,\bounde_p)^2\left[\sigma ^2 \int_0^T e_p(t)^2 dt+ 
2\int_0^T |q(t,x_p(t),e_p(t))|^2 dt\right] .
\end{array}
$$
On another hand, we obtain, by integration and using  \eqref{pf:L2L2gain_zeta} again, 
$$
\begin{array}{l}
\displaystyle
0=\frac{e_p(T)^2 - e_p(0)^2}{2}
\leq 
\\ \displaystyle
\qquad \leq  -\sigma \,  \int_0^T e_p(t)^2 dt
+ \mu   \int_0^T M^\top N_z \zeta _p(t) e_p(t) dt
+ \int_0^T  q(t,x_p(t),e_p(t)) e_p(t) dt
\\ \displaystyle
\qquad \leq  -\frac{\sigma }{2}  \int_0^T e_p(t)^2 dt
+ \frac{1}{\sigma }  \mu^2    \int_0^T |M^\top N_z \zeta _p(t)|^2 dt
+\frac{1}{\sigma } \int_0^T  |q(t,x_p(t),e_p(t))|^2 dt
\\ \displaystyle
\qquad \leq  -\frac{\sigma }{2} \int_0^T e_p(t)^2 dt
+\frac{2}{\sigma }  \int_0^T  |q(t,x_p(t),e_p(t))|^2 dt
\end{array}
$$
and therefore
$$
\begin{array}{rcl}
\int_0^T |\dot q_p(s)|^2 ds
&\leq& \displaystyle 3T \left[\boundq_t(\boundx_p,\bounde_p)^2+\boundq_x(\boundx_p,\bounde_p)^2\boundf(\boundx_p,\bounde_p)^2\right]+
54 \boundq_e(\boundx_p,\bounde_p)^2  \int_0^T |q(t,x_p(t),e_p(t))|^2 dt
\\
& \leq &\displaystyle  3T \left[\boundq_t(\boundx_p,\bounde_p)^2+\boundq_x(\boundx_p,\bounde_p)^2\boundf(\boundx_p,\bounde_p)^2\right]
+
54\boundq_e(\boundx_p,\bounde_p)^2 T \boundq(\boundx_p,\bounde_p)^2.
\end{array}
$$
By combining this inequality with 
 \eqref{pf:L2_zeta} and \eqref{app:bound_q}, we finally obtain
 an expression for the bound $ \boundzeta_p$ in inequality
 \eqref{eq:prop_periodic_bound_zeta}, that is,
\begin{multline*}
 \boundzeta_p^2(\mu):= 
 \left[\frac{T}{\mu} +  \kappa_0
 + 54\kappa_1 \boundq_e(\boundx_p,\bounde_p)^2 \right] 
 \boundq(\boundx_p,\bounde_p)^2
+ 3\kappa_1
  \left[\boundq_t(\boundx_p,\bounde_p)^2+\boundq_x(\boundx_p,\bounde_p)^2\boundf(\boundx_p,\bounde_p)^2\right]
\end{multline*}
with $\kappa_0,\kappa_1$ given by Lemma~\ref{lemma:transfer_zeta}.
It is independent of $\sigma$ and of $\mu$ by noting 
that it is a decreasing function of $\mu$ for $\mu\geq1$, that is, 
select $\boundzeta_p(1)$ in the statement of the Theorem.

Now, inequalities \eqref{eq:long_bound_xep}
are parts
of the inequalities
\eqref{eq:prop_periodic_bound_x}
and \eqref{eq:prop_periodic_bound_e}.
On another hand,
with  \eqref{8} and 
\eqref{app:bound_deltaf}, we
obtain
\begin{align*}
\sup_{t\in[0,T]}|x_p(t)|\leq 
  \dfrac{K_\alpha}{\alpha}\left[
\boundf_e(\boundx_p,\bounde_p)\sup_{t\in[0,T]}|e_p(t)|
+ \dfrac{1}{2}\boundf_{xx}(\boundx_p,0)\boundx_p
\sup_{t\in[0,T]}|x_p(t)|
\right]
\end{align*}
which gives
\begin{equation}
\label{eq:bound_prop1_xp}
\sup_{t\in[0,T]}|x_p(t)|\leq   \dfrac{K_\alpha\boundf_e(\boundx_p, \bounde_p)}{\alpha - 
\tfrac12K_\alpha\boundf_{xx}(\boundx_p,0)\boundx_p}\sup_{t\in[0,T]}|e_p(t)|,
\end{equation}
where the denominator of the right hand side is strictly positive according to 
\eqref{eq:condition_prop_1_bound_xpep}.
Now recall that by definition of \eqref{eq:sequence_nzk}
and of the matrices $M, N_z$, we have
$\sqrt{M^\top N_z M} \leq   \overbar N_z$.
As a consequence, inequalities \eqref{8} and \eqref{eq:prop_periodic_bound_zeta}
give
\begin{align*}
\sup_{t\in[0,T]} |e_p(t)| \leq & \dfrac1\sigma \left[ \sup_{t\in[0,T]} 
q(t,x_p(t),e_p(t)) + \sup_{t\in[0,T]} M^\top N_z \zeta_p(t) \right]
                \\
                \leq & \dfrac1\sigma \left[\boundq(\boundx_p, \bounde_p)  + \sqrt{M^\top N_z M} \sup_{t\in[0,T]}\sqrt{ \zeta_p(t)^\top N_z \zeta_p(t)} \right]
                \\ \leq &
                \dfrac{1}\sigma\big(\boundq(\boundx_p, \bounde_p)  + \overbar N_z\boundzeta_p\big)
\end{align*}
This yields the remaining parts of
of the inequalities
\eqref{eq:prop_periodic_bound_x}
 and \eqref{eq:prop_periodic_bound_e}
with
 $$
 \boundr_e :=
 \boundq(\boundx_p, \bounde_p)  + \overbar N_z\boundzeta_p\ ,
\qquad
\boundr_x := \dfrac{K_\alpha\boundf_e(\boundx_p, \bounde_p)}{\alpha - \tfrac12K_\alpha\boundf_{xx}(\boundx_p,0)\boundx_p}\boundr_e\  .
 $$

Finally, to complete the proof, we need to show that it does exist
a triple of positive real numbers
$(\boundx_p,\bounde_p,\sigma ^\star_p)$
satisfying \eqref{eq:condition_prop_1_bound}
and \eqref{eq:contraction_condition}, i.e.
\\[1em]$\displaystyle 
\boundx_p\; \geq \; 
\dfrac{K_\alpha }{\alpha }\left[ \boundf_e(\boundx_p,\bounde_p)\,  \bounde_p
\;+\; 
\frac{1}{2}\boundf_{xx}(\boundx_p,0)\,\boundx_p^2\right] 
\  ,
$\refstepcounter{equation}\label{12}\hfill$(\theequation)$
\\[0.51em]$\displaystyle
 \bounde_p\; \geq \; 
 \beta(\sigma)
\boundq (\boundx_p,\bounde_p)
 \ ,
$\refstepcounter{equation}\label{13}\hfill$(\theequation)$
\\[0.51em]$\displaystyle
\null \  \,   1\; >\; 
\dfrac{K_\alpha }{\alpha }\left[ \boundf_{ex}(\boundx_p,\bounde_p)\,  \bounde_p 
+
2\boundf_{xx}(\boundx_p,0)\,\boundx_p\right]
+
\beta(\sigma)\boundq_e(\boundx_p,\bounde_p)
$\refstepcounter{equation}\label{14}\hfill$(\theequation)$
\\[0.51em]\null\hfill$\displaystyle
+\dfrac{K_\alpha\beta(\sigma) }{\alpha }
\left(\boundf_e(\boundx_p,\bounde_p)
\boundq_x(\boundx_p,\bounde_p)
-
\left[ \boundf_{ex}(\boundx_p,\bounde_p)\,  \bounde_p 
+
2\boundf_{xx}(\boundx_p,0)\,\boundx_p\right]
\boundq_e(\boundx_p,\bounde_p)
\right)
\  ,
$\\[0.51em]$\displaystyle
\null \  \,   1\; >\; 
\dfrac{K_\alpha\beta(\sigma) }{\alpha }
\left(\boundf_e(\boundx_p,\bounde_p)
\boundq_x(\boundx_p,\bounde_p)
-
\left[ \boundf_{ex}(\boundx_p,\bounde_p)\,  \bounde_p 
+
2\boundf_{xx}(\boundx_p,0)\,\boundx_p\right]
\boundq_e(\boundx_p,\bounde_p)
\right)\!.
$\refstepcounter{equation}\label{15}\hfill$(\theequation)$
\\[1em]
Since the function $\beta $ is continuous, strictly decreasing and tends to $0$ as $\sigma $ tends to infinity, it is sufficient to 
find $\boundx_p$ and $\bounde_p$ satisfying~:
\begin{eqnarray*}
\boundx_p &> &
\dfrac{K_\alpha }{\alpha }\left[ \boundf_e(\boundx_p,\bounde_p)\,  \bounde_p
\;+\; 
\frac{1}{2}\boundf_{xx}(\boundx_p,0)\,\boundx_p^2\right] 
\  ,
\\
1&> &
\dfrac{K_\alpha }{\alpha }\left[ \boundf_{ex}(\boundx_p,\bounde_p)\,  \bounde_p 
+
2\boundf_{xx}(\boundx_p,0)\,\boundx_p\right]
\  .
\end{eqnarray*}
Since the bounding functions $\boundf_{\bullet}$ are increasing in each of their arguments,
%
%
%
%
we choose 
$$
\boundx_p<\boundx_{p*}
\quad ,\qquad 
\bounde_p\leq \bounde_{p*}(\boundx_p),
$$
where $\boundx_{p*}$ is a strictly positive real number
satisfying
$$
\boundx_{p*} \boundf_{xx}(\boundx_{p*},0)\; <\; \frac{\alpha }{2K_{\alpha}}
$$
and $\bounde_{p*}(\boundx_{p})$ is a strictly positive real number  satisfying
\begin{eqnarray*}
\bounde_{p*}(\boundx_{p})\,  \boundf_e(\boundx_{p},\bounde_{p*}(\boundx_{p}))
&\leq &
\frac{\alpha }{K_\alpha }\boundx_{p} - \frac{1}{2}\boundf_{xx}(\boundx_{p},0)\,\boundx_{p}^2,
\\
\bounde_{p*}(\boundx_{p})
 \boundf_{ex}(\boundx_{p},\bounde_{p*}(\boundx_{p})) 
&<&
\frac{\alpha }{K_{\alpha}} - 2\boundx_{p} \boundf_{xx}(\boundx_{p},0) 
\  .
\end{eqnarray*}
In this way (\ref{12}) holds and it remains to find $\sigma ^\star_p$ large enough to satisfy the following 
$3$ inequalities
\begin{equation}
\begin{array}{rcl}
\beta(\sigma ^\star_p)  &\leq & \dfrac{\bounde_p}{\boundq (\boundx_p,\bounde_p)}
\\[1em]
\beta (\sigma ^\star_p )
&< &
 \frac{1-\frac{K_\alpha }{\alpha }\left[ \boundf_{ex}(\boundx_p,\bounde_p)\,  \bounde_p 
+
2\boundf_{xx}(\boundx_p,0)\,\boundx_p\right]
}
{
\boundq_e(\boundx_p,\bounde_p)
+\frac{K_\alpha }{\alpha }
\left(\boundf_e(\boundx_p,\bounde_p)
\boundq_x(\boundx_p,\bounde_p)
-
\left[ \boundf_{ex}(\boundx_p,\bounde_p)\,  \bounde_p 
+
2\boundf_{xx}(\boundx_p,0)\,\boundx_p\right]
\boundq_e(\boundx_p,\bounde_p)
\right)
}
\\[1em]
\beta (\sigma ^\star_p )& < &
\dfrac{\alpha }{ 
K_\alpha
\left(\boundf_e(\boundx_p,\bounde_p)
\boundq_x(\boundx_p,\bounde_p)
-
\left[ \boundf_{ex}(\boundx_p,\bounde_p)\,  \bounde_p 
+
2\boundf_{xx}(\boundx_p,0)\,\boundx_p\right]
\boundq_e(\boundx_p,\bounde_p)
\right)
}.
\end{array}
\label{19}
\end{equation}
Note $n_o$ plays no role in this process implying that the triple $(\boundx_p,\bounde_p,\sigma ^\star_p)$ does 
not depend on $n_o$.\qed
 \end{proof}

\begin{remark}
\label{rem1}
We stress the fact that $\boundx_p$ can be chosen arbitrarily in $]0,\boundx_{p*}[$. Then $\boundx_p$ being 
fixed,
$\bounde_p$ can be 
chosen arbitrarily in $]0,\bounde_{p*}(\boundx_{p})]$. And finally, $\boundx_p$ and $\bounde_p$ being fixed,
$\sigma ^\star_p$ is chosen to satisfy the $3$ inequalities 
in \eqref{19}
\end{remark}

\subsection{Exponential Stability of the Periodic Solution}

\begin{proposition}\label{prop:exponential_solution}
For any triplet $(\overline P_x,\underline P_x,\alpha)$,
 there exist
strictly positive
real numbers  $\boundx_a, \bounde_a>0$
(independent of $n_o$)
and, for any $\boundzeta_a>0$
and $\mu\geq1$, 
there exists $\sigma ^\star_a>0$
(independent of $n_o$)
such that, 
if system \eqref{eq:closed-loop},
with $\sigma>\sigma ^\star_a$, 
 $f\in \F(\boundx_a, \bounde_a,\overline P_x,\underline P_x,\alpha)$ and $q\in\Q(\boundx_a, \bounde_a)$,
admits a $T$-periodic solution $(x_p, e_p, z_p)$
satisfying
\begin{equation}
\label{eq:prop_exp_bound_periodic}
\sup_{t\in [0,T]}|x_p(t)| \leq  \boundx_a, 
\
 \sup_{t\in [0,T]}|e_p(t)|\leq  \bounde_a, 
\
\quad \sup_{t\in [0,T]}\sqrt{\zeta_p(0)^\top N_z \zeta_p(0)} \leq \boundzeta_a,
\end{equation} 
then, such periodic solution  is exponentially stable
with a domain of attraction containing the set 
\begin{equation}
\N_\zeta(\boundx_a, \bounde_a,\boundzeta_a)=
\Big\{(x,e,\zeta): \; |x-x_p(0)| \leq  3\boundx_a,
\; |e-e_p(0)| \leq  3\bounde_a, \;
\sqrt{(\zeta-\zeta_p(0))^\top N_z (\zeta-\zeta_p(0))} \leq 3\boundzeta_a
\Big\}.
\label{eq:setN}
\end{equation}
\end{proposition}

\begin{proof}
The assumed existence of
a periodic solution  $(x_p,e_p,\zeta_p)$
satisfying \eqref{eq:prop_exp_bound_periodic} for system
\eqref{eq:closed-loop} allows us to consider the following change of coordinates
$$
(x,e,\zeta) \mapsto (\tilde x, \tilde e, \tilde \zeta) :=
(x-x_p(t),e-e_p(t),\zeta - \zeta_p(t))
$$
System 
\eqref{eq:closed-loop} is transformed into
\begin{equation}
\label{eq:closed-loop3}
\begin{array}{rcl}
\dot {\tilde x} & = & F_p(t) \tilde x+ \delta \tilde f(t,\tilde x,\tilde e)
\\
\dot {\tilde e} & =&  \tilde q(t,\tilde x,\tilde e) - \sigma \tilde e + \mu M^\top N_z \tilde\zeta
\\
\dot \zeta & = & (\Phi - \mu M M^\top N_z)\tilde\zeta  - M\tilde q(t,\tilde x,\tilde e),
\end{array}
\end{equation}
with the definitions
\begin{align}
\label{eq:def_Fp}
F_p(t) := &\frac{\partial f}{\partial x}(t,x_p(t),e_p(t))
\ , 
\\
\label{eq:f_decomposition2}
\delta \tilde f(t,\tilde x,\tilde e) := &
 \displaystyle\tilde f(t,\tilde x,\tilde e)-  F_p(t) \tilde x
  \ ,
\\ \label{eq:tildef_def}
 \tilde f(t,\tilde x,\tilde e):=&f(t,x,e)-f(t,x_p(t),e_p(t))
\  ,
\\
\tilde q(t,\tilde x,\tilde e):=&q(t,x,e)-q(t,x_p(t),e_p(t))
\  .\label{eq:tildeq_def}
\end{align}

Exponential stability of the origin for the system (\ref{eq:closed-loop3}) can be established from the first 
order approximation. But we want to establish not only exponential stability but also that the domain of attraction contains $
\N_\zeta(\boundx_a, \bounde_a,\boundzeta_a)$ defined in \eqref{eq:setN}. So instead we go with a Lyapunov 
analysis.

First, we define $\boundx_a,\bounde_a$ and $\sigma ^\star_a$ via a set of inequality.
To this end,
to any pair $(\boundx_a,\bounde_a)$ we associate a pair
$(\boundx_b,\bounde_b)$ as follows
\begin{equation}
\label{eq:temp_bound_xe2}
\boundx_b:= \left(
3{\sqrt2} \sqrt{\frac{\overline P_x}{\underline P_x}}+ 2 \right)\boundx_a, 
\qquad
\bounde_b := 6 \bounde_a.
\end{equation}
and the real number
 \begin{equation}
\rho_0(\boundx_a,\bounde_a) := 
\frac{\overline{P}_x}{2\underline{P}_x}
\left[\boundf _{xx} (\boundx_a,\bounde_a) \boundx_a
+\boundf _{ex} (\boundx_a,\bounde_a) \bounde_a
+\underline{P}_x
\boundf _{xx} (\boundx_b,\bounde_a) (\boundx_b+\boundx_a)\right].\label{eq:def_rho_0}
 \end{equation}
 By definition, the function $\rho_0$
 is a non-decreasing function satisfying $\rho_0(0,0)=0$.
 With the above definitions, let 
 $\boundx_a,\bounde_a>0$ be any pair of positive numbers
 satisfying the following two inequalities
 \begin{subequations}\label{eq:boundx_bounde}
\begin{align}
\label{eq:boundx_bounde_1}
9 \overline{P}_x \boundx_a^2 + 9 \bounde_a^2
 & \leq \frac12 \min\left\{\underline{P_x} (\boundx_b-\boundx_a)^2,
(\bounde_b-\bounde_a)^2\right\},
\\
\label{eq:boundx_bounde_2}
\rho_0(\boundx_a,\bounde_a) & 
 \leq \frac13 \alpha \underline{P}_x,
\end{align}
\end{subequations}
that depends only on the triplet of numbers $(\overline P_x,\underline P_x,\alpha)$.
Such $\boundx_a,\bounde_a>0$ always exists 
in view of the above definition of $\boundx_b,\bounde_b$
in \eqref{eq:temp_bound_xe2} and the properties of $\rho_0$.
Then, with $\boundx_a,\bounde_a$ 
and $\mu\geq1$ being
fixed, for any given $\boundzeta_a>0$, let $h>0$ a small 
positive number 
satisfying 
the following inequality 
 \begin{equation}\label{eq:ineq_bounds_b}
h \leq \frac13 
  \min \left\{ \frac1{(3\boundzeta_a)^2} \min\left\{\underline{P_x} (\boundx_b-\boundx_a)^2,
(\bounde_b-\bounde_a)^2\right\},
\; 
\mu
\dfrac{\alpha \underline{P}_x }{\boundq_x(\boundx_b,\bounde_b)^2}
\right\}
\end{equation}
Finally, with $\boundx_a,\bounde_a,\boundx_b,\bounde_b,\mu$ and $h$ fixed, we select
$\sigma ^\star_a$
as 
\begin{equation}\label{eq:ineq_bounds_sigma}
\sigma ^\star_a:=\frac{
\left[\overline{P}_x \boundf_e(\boundx_b, \bounde_b)
+\boundq_x(\boundx_b,\bounde_b)\right]^2
}{
2\left[\alpha \underline P_x- \rho_0(\boundx_a,\bounde_a) - 
\tfrac{h}{\mu}\boundq_x(\boundx_b, \bounde_b)^2\right]
}+\frac{\mu }{2h}+  \boundq_e(\boundx_b, \bounde_b)
+\tfrac{h}{\mu}\boundq_e(\boundx_b, \bounde_b)^2.
\end{equation}
Again, $n_o$ plays no role in the above inequalities. This implies
the triple $(\boundx_a,\bounde_a,\sigma_1^\star)$ does 
not depend on $n_o$.

With \eqref{eq:prop_exp_bound_periodic}, 
$f\in \F(\boundx_b, \bounde_b,\overline P_x,\underline P_x,\alpha)$
and $q\in\Q(\boundx_b, \bounde_b)$,
and by using  inequality \eqref{app:ineq_varphi}, we obtain
\begin{align}\label{app:bound_tildeq}
|\tilde q(t,\tilde x,\tilde e)| \leq  &
\boundq_e(\boundx_b,\bounde_b)  |\tilde e|
+
\boundq_x(\boundx_b,\bounde_b) |\tilde x|\\
|\delta \tilde f(t,\tilde x,\tilde e)| \leq   &
\boundf_e(\boundx_b,\bounde_b)\,  |\tilde e|
+
\boundf_{xx}(\boundx_b,\bounde_a)\,  |\tilde x|^2
\label{app:bound_delta_tildef}
\end{align}
for all 
$(x,e)\in  \mathcal{S}(\boundx_b,\bounde_b)$, 
$t\in [0,T]$.
Also,
when $F_p$ is sufficiently 
close to $F(t)$, we 
can still use the function $P_x$ defined in 
\eqref{eq:Pass1}, \eqref{eq:Pass2}
to build a candidate Lyapunov function 
for the $\tilde x$-subsystem. 
For this, let us define the function
\begin{equation}
\label{eq:def_deltaF}
\delta F(t) := F_p(t)-F(t)= \frac{\partial f}{\partial x}(t,x_p(t),e_p(t))-\frac{\partial f}{\partial x}(t,0,0).
\end{equation}
It is $T$-periodic and satisfies 
\begin{equation}
\label{app:bound_delta_F}
|\delta F(t)|
\leq 
|F_p(t)- F(t)|
 \leq   \boundf _{xx} (\boundx_a,\bounde_a) \boundx_a
+\boundf _{ex} (\boundx_a,\bounde_a) \bounde_a,
\end{equation}
for all  $t\in [0,T]$.
Therefore, by using \eqref{eq:Pass2}, we compute
$$
\begin{array}{rcl}
\dot{P}_x(t)
+P_x(t) F_p(t)
+F_p(t)^\top P_x(t)
 & \leq &-2\alpha P_x + \left[P_x(t) \delta F(t) + \delta F(t)^\top P_x(t)\right]
\\
& \leq & -2\left(\alpha \underline{P}_x - 
\frac{\overline{P}_x}{\underline{P}_x}\left[\boundf _{xx} (\boundx_a,\bounde_a) \boundx_a
+\boundf _{ex} (\boundx_a,\bounde_a) \bounde_a\right]\right) I
\end{array}
$$
for all 
$(x,e)\in  \mathcal{S}(\boundx_b,\bounde_b)$, $t\in [0,T]$.
Hence, by defining $V_x : = \tilde x^\top P_x \tilde x$, we 
obtain
\begin{eqnarray*}
\dot V_x & = &
2\tilde x^\top P_x(t)(F_p(t)\tilde x + \delta \tilde f(t,\tilde x,\tilde e))
+\tilde x^\top \dot P_x(t) \tilde x
\\
& \leq & -
2\Big(\!\alpha \underline{P}_x  - 
\frac{\overline{P}_x}{\underline{P}_x}\left[\boundf _{xx} (\boundx_a,\bounde_a) \boundx_a
+\boundf _{ex} (\boundx_a,\bounde_a) \bounde_a\right]\!\Big) 
|\tilde x|^2
+
2  \overline P_x  \boundf_e(\boundx_b,\bounde_b)\,  |\tilde e|\, 
 |\tilde x |
+
2 \overline P_x \boundf_{xx}(\boundx_b,\bounde_a)\,  |\tilde x|^3
\\
&\leq &
- 2[ \alpha \underline{P}_x - \rho_0(\boundx_a,\bounde_a)]
  |\tilde x|^2 + 2  \overline P_x  \boundf_e(\boundx_b,\bounde_b)\,  |\tilde e|\, 
 |\tilde x |
\end{eqnarray*} 
for all 
$(x,e)\in  \mathcal{S}(\boundx_b,\bounde_b)$, $t\in [0,T]$,
 where in the last step we used the inequality
$|\tilde x|\leq
 |x_p(t)|
+|x|\leq \boundx_a+\boundx_b$, 
and $\rho_0$ has been defined in \eqref{eq:def_rho_0}.
Next, the derivative of $V_e : = \tilde e^2$ satisfies
$$
\dot V_e  \leq 
 -2\sigma |\tilde e|^2 
+ 2\mu M^\top N_z \tilde \zeta \tilde e
+ 2 \boundq_e(\boundx_b,\bounde_b)\,  |\tilde e|^2
+ 2 \boundq_x(\boundx_b,\bounde_b)\,  |\tilde x| |\tilde e|,
$$
for all 
$(x,e)\in  \mathcal{S}(\boundx_b,\bounde_b)$, $t\in [0,T]$,
in which we used \eqref{app:bound_tildeq}.
Finally, for the $\tilde \zeta$ dynamics, 
define the function  $V_\zeta = \tilde \zeta^\top N_z\tilde \zeta$.
By using again inequality \eqref{app:bound_tildeq},
its derivative satisfies
\begin{align*}
\dot V_\zeta & = 
2 \zeta^\top N_z [(\Phi - \mu M M^\top N_z)\tilde\zeta  - M\tilde q(t,\tilde x,\tilde e)]
\\
& \leq  
- \mu |M^\top N_z \tilde \zeta |^2
+\frac2\mu \boundq_e(\boundx_b,\bounde_b)^2|\tilde e|^2
+\frac2\mu \boundq_x(\boundx_b,\bounde_b)^2|\tilde x|^2
\end{align*}
for all 
$(x,e)\in  \mathcal{S}(\boundx_b,\bounde_b)$, $t\in [0,T]$.
As a consequence, by collecting all the inequalities together,
we conclude that the time derivative of
the function
$$
U :=  V_x + V_e + h V_\zeta, 
\qquad \tilde \chi := (|\tilde x|, |\tilde e|,
|M^\top N_z\tilde \zeta|)^\top ,
$$
satisfies
$$
\dot U \leq 
-\tilde \chi^\top \R(\boundx_a, \bounde_a,\sigma,\mu,h) \tilde \chi , 
$$
with $\R$ defined as
\begin{equation}
\label{eq:R-lyap1}
\R(\boundx_a, \bounde_a,\sigma,\mu,h) := \begin{pmatrix}
2(\alpha \underline{P}_x -\rho_{11}(\boundx_a,\bounde_a,\mu,h)) &&  -\rho_{12}(\boundx_a,\bounde_a) && 0 
\\
-\rho_{12}(\boundx_a,\bounde_a) && 2(\sigma-\rho_{22}(\boundx_a,\bounde_a,\mu,h)) && -\mu
\\
 0&& -\mu && \mu h
\end{pmatrix},
\end{equation}
and where the functions $\rho_{11}$,
$\rho_{12}$ and $\rho_{22}$ are 
non-decreasing functions
defined as
$$
\label{eq:R-lyap2}
\begin{array}{rl}
\rho_{11}(\boundx_a,\bounde_a,\mu,h)  := & 
 \rho_0(\boundx_a,\bounde_a) + 
\tfrac{h}{\mu}\boundq_x(\boundx_b, \bounde_b)^2,
\\
\rho_{12}(\boundx_a,\bounde_a) : = &  \overline{P}_x \boundf_e(\boundx_b, \bounde_b)
+\boundq_x(\boundx_b,\bounde_b),
\\
\rho_{22}(\boundx_a,\bounde_a,\mu,h)  : = &
  \boundq_e(\boundx_b, \bounde_b)
+\tfrac{h}{\mu}\boundq_e(\boundx_b, \bounde_b)^2 .
\end{array}
$$
According to Sylvester's criterion, the matrix $\R$
is positive definite if and only if the leading principal minors
are all strictly positive, that is, we need to satisfy
\begin{eqnarray}
\label{16}
\alpha \underline P_x - \rho_{11}& >0,
\\
\label{17}
4(\alpha \underline P_x - \rho_{11})(\sigma-\rho_{22})-\rho_{12}^2 &>0,
\\
\label{18}
2(\alpha \underline{P}_x -\rho_{11})
\left[2(\sigma-\rho_{22}) 
 h
-
\mu
\right]
-
h
\rho_{12}^2
&>0,
\end{eqnarray}
where the arguments of the functions $\rho_{ij}$
have been omitted for compactness.
These inequalities are always satisfied
for any $\sigma>\sigma ^\star_a$ since
\begin{itemize}
\item[$\bullet$]
\eqref{16} is implied by \eqref{eq:boundx_bounde_2} 
and \eqref{eq:ineq_bounds_b}.
\item[$\bullet$]
\eqref{17} is implied by \eqref{16} and \eqref{18}.
\item[$\bullet$]
\eqref{18} is implied by \eqref{eq:ineq_bounds_sigma}.
\end{itemize}
As a consequence, for any $\sigma>\sigma ^\star_a$,
there exists a  positive real number
 $\nu>0$, which depends 
 only on $(\boundx_a, \bounde_a,\sigma,\mu,h)$,
so that the derivative of $U$ satisfies the following inequality
\begin{equation}
\label{eq:derU}
\dot U \leq - \nu
\left(|\tilde x|^2 + |\tilde e| + |M^\top N_z\tilde\zeta|^2  \right)
\end{equation}
for all  
$(x,e)\in  \mathcal{S}(\boundx_b,\bounde_b)$, $t\in [0,T]$.

Next, to make
the right hand side of \eqref{eq:derU} negative definite,
consider the function
$$
V_{\zeta,s} := \tilde \zeta^\top (N_z+\kappa P_\zeta)\tilde \zeta,
$$ where $\kappa P_\zeta$ is defined in 
Lemma~\ref{lemma:hurwitzPhiM}
 and depends on $n_o$.
By using the inequality \eqref{app:ineq_PhiM}, we compute
its derivative. It satisfies
\begin{eqnarray*}
\dot V_{\zeta,s} &\leq &\displaystyle 
- \mu |M^\top N_z \tilde \zeta |^2-\kappa \tilde \zeta^\top N_z \tilde \zeta 
-2
\left[\tilde \zeta ^\top N_z  M \right]\tilde q(t,\tilde x,\tilde e)
-2 \kappa 
\left[\tilde \zeta ^\top N_z^{1/2}\right]\left[N_z^{-1/2} P_\zeta  M \right] \tilde q(t,\tilde x,\tilde e)
\\
&\leq &\displaystyle 
-\frac{\kappa}{2} \tilde \zeta^\top N_z \tilde \zeta 
+
\frac{1}{\mu }q(t,\tilde x,\tilde e)^2
+
2\kappa \left|N_z^{-1/2} P_\zeta  M\right|^2 q(t,\tilde x,\tilde e)^2
\\
&\leq &\displaystyle 
-
\dfrac{\kappa }{2}\tilde \zeta^\top N_z \tilde \zeta 
 + 
2\left[
\frac{1}{\mu }
+
2\kappa \left|N_z^{-1/2} P_\zeta  M\right|^2 
\right]
\left[\boundq_e(\boundx_b,\bounde_b)^2|\tilde e|^2
+
\boundq_x(\boundx_b,\bounde_b)^2 |\tilde x|^2\right],
\end{eqnarray*}
for all  
$(x,e)\in  \mathcal{S}(\boundx_b,\bounde_b)$, $t\in [0,T]$.
Hence, by letting 
$$
c := 
\dfrac{\nu}{2\left[
\frac{1}{\mu }
+
2\kappa \left|N_z^{-1/2} P_\zeta  M\right|^2 
\right]\max\{\boundq_e(\boundx_b,\bounde_b)^2,\boundq_x(\boundx_b,\bounde_b)^2 \}}
$$
we conclude, by using \eqref{eq:derU}, that
the derivative of $U_s := U + c V_{\zeta,s}$
satisfies
\begin{equation}
\label{eq:derUs}
\dot U_s \leq  - \dfrac\nu2 \left( 
|\tilde x|^2 +|\tilde e|^2
+ |M^\top N_z \tilde \zeta |^2
\right)
- \dfrac{c\kappa}2 \tilde\zeta^\top N_z \tilde \zeta 
\end{equation}
for all 
$(x,e)\in  \mathcal{S}(\boundx_b,\bounde_b)$, $t\in [0,T]$.
At this point is important to note for the exponential stability that, on the left hand side we have the time 
derivative of a function which,
although time dependent via its $V_x$ component, 
can be upperbounded by 
time independent positive definite quadratic forms of $(\tilde x,e,\tilde \zeta)$ and on the right hand side,
we have a time independent positive definite quadratic form.

Now, as a consequence of inequalities
\eqref{eq:boundx_bounde_1} and \eqref{eq:ineq_bounds_b},
 we also have
\begin{equation}
\label{eq:ineq_bounds_sets}
9\overline{P}_x \boundx_a^2 + 9 \bounde_a^2 + 9 h \boundzeta_a^2
 \;\leq \; \frac56 \min\left\{\underline{P_x} (\boundx_b-\boundx_a)^2,
(\bounde_b-\bounde_a)^2\right\}.
\end{equation}
Then,
with the definitions of $\boundx_b$ and 
$\bounde_b$ given in \eqref{eq:temp_bound_xe2},
let
\begin{eqnarray*}
\mathcal{O} & := &\{
(\tilde x, \tilde e, \tilde \zeta ) : |\tilde x|< 
\boundx_b-\boundx_a, \ 
  |\tilde e|< \bounde_b-\bounde_a\},
\\
&=& \left\{
(\tilde x, \tilde e, \tilde \zeta ) : |\tilde x|< 
\left(1+3\sqrt2 \sqrt{ \tfrac{\overline P_x}{\underline P_x}}\right)\boundx_a, \ 
  |\tilde e|< 5\bounde_b \right\}.
\end{eqnarray*}
Our interest in this set is coming from the following implication,
given by \eqref{eq:prop_exp_bound_periodic}, 
\begin{equation}
\label{eq:temp_cond_xe_O}
(\tilde x, \tilde e, \tilde \zeta ) \in \mathcal{O}
\quad
\Longrightarrow
\quad
\left\{
\begin{array}{l}
|x| \leq |\tilde x|+|x_p| < \boundx_b \ ,
\\
|e| \leq  |\tilde e|+|e_p| < \bounde_b \ .
\end{array}
\right.
\end{equation}
As a consequence, let $(\tilde x(t),\tilde e(t),\tilde \zeta (t))$ be an arbitrary  solution with $[0,\tau )$ as right 
maximal interval of definition into the set $\mathcal{O}$ and with initial conditions in
 the set $\N_\zeta$  defined in \eqref{eq:setN}, that is, 
$$
 N_\zeta=
\bigg\{ (x,e,\zeta) \ : \
|x-x_p(0)| \leq  3\boundx_a,
\; |e-e_p(0)| \leq  3\bounde_a, \;
\sqrt{(\zeta-\zeta_p(0))^\top N_z (\zeta-\zeta_p(0))} \leq 3\boundzeta_a
\bigg\}.
$$
 Note that by using the definitions of $\tilde x$, $\tilde e$,
 $\tilde\zeta$
and the bounds on the periodic solution
 $(x_p,e_p,\zeta_p)$
in \eqref{eq:prop_exp_bound_periodic}, 
$\N_z$ is a subset of
$ \{(\tilde x, \tilde e, \tilde \zeta): |\tilde x|\leq 4\boundx_a, |\tilde e|\leq 4\bounde_a\}$
and therefore, by using the definitions of  $\boundx_b$, and 
$\bounde_b$ given in \eqref{eq:temp_bound_xe2}, 
$\N_z$ is a subset of
$\mathcal{O}$.

Now, if $\tau$ is finite, we have
$$
\lim_{t\to \tau }|\tilde x (t)|= \boundx_b-\boundx_a, 
\quad \textnormal{or} \quad
\lim_{t\to \tau }|\tilde e (t)|= \bounde_b-\bounde_a, 
\quad \textnormal{or} \quad
\lim_{t\to \tau }\tilde \zeta(t) ^\top N_z  \tilde \zeta(t) = \infty.
$$
But because of  \eqref{eq:temp_cond_xe_O},
 we can use \eqref{eq:derU}
and \eqref{eq:setN}
to get 
 $$
 U(t) \leq U(0) \leq \overline{P}_x (3\boundx_a)^2 + 
 (3\bounde_a)^2 + h (3\boundzeta_a)^2
 $$
 for all $t\in [0,\tau )$,
 and in particular, by using the definition of $U$
 and inequality \eqref{eq:ineq_bounds_sets}, 
$$
\begin{array}{rcl}
|\tilde x(t)| & \leq &
\sqrt{\frac{1}{\underline P_x}U(t)} < 
 \boundx_b-\boundx_a 
\\
|\tilde e(t)| &\leq &  
  \bounde_b-\bounde_a
 \\
 \tilde \zeta(t) ^\top N_z  \tilde \zeta(t) & \leq &
 \displaystyle \frac{1}{h}
\min\left\{\underline{P_x} (\boundx_b-\boundx_a)^2,
(\bounde_p-\bounde_a)^2\right\}
\end{array}
$$
for all $t\in [0,\tau )$.
This contradicts the consequences of having $\tau $ finite. So
$\tau $ is infinite and with the inequality  \eqref{eq:derUs} and the 
remark following it, we conclude the periodic solution is asymptotically stable and locally exponentially 
stable.
Its domain of attraction contains the 
set $\N_\zeta$,
defined in \eqref{eq:setN}.
This concludes the proof.\qed
\end{proof}

\subsection{Proof of Item 1) of Theorem~\ref{thm1}}
To establish the result 
of item 1 of Theorem~\ref{thm1}, we rely on Propositions \ref{prop:periodic_solution} and 
\ref{prop:exponential_solution}.
We start by defining the numbers
$\boundx$, $\bounde$, $\boundz$, and $\sigma ^\star$ of the claim.
Let $\boundx_a, \bounde_a$ be given by Proposition~\ref{prop:exponential_solution}. Then
let $\boundx_p, \bounde_p,\boundr_x, \boundr_e, \boundzeta_p,\sigma ^\star_p$ be given by Proposition~\ref{prop:periodic_solution}, which thanks to Remark 
\ref{rem1}, can be assumed to satisfy
\begin{equation}
\label{20}
\boundx_p\leq \boundx_a
\quad ,\qquad 
\bounde_p\leq \bounde_a
\  .
\end{equation}
We define
\begin{equation}
\label{eq:boundxe_item1}
\boundx := \boundx_p, 
\qquad
\bounde := \bounde_p, 
\end{equation}
and 
\begin{equation}
\label{eq:boundz_item1}
\boundz :=  \boundzeta_p + \overline N_z\bounde, 
\end{equation}
with $\overline N_z$ defined in 
\eqref{eq:sequence_nzk}.
We choose $\mu\geq1$ arbitrary. Then, with
\begin{equation}
\label{eq:boundzeta_item1}
\boundzeta_a := \boundzeta_p 
+ 2\overline N_z\bounde ,
\end{equation}
Proposition~\ref{prop:exponential_solution} gives us $\sigma ^\star_a$. We select
$$
\sigma^\star := \max\{\sigma ^\star_p, \sigma ^\star_a\} .
$$

With these numbers defined, we 
consider system 
\eqref{eq:sys}, 
\eqref{eq:controller}
for some given $\sigma>\sigma^\star$,
and some given $n_o$.
We apply the change of coordinates 
\eqref{eq:change_of_coordinates_z_zeta}
to the closed-loop system \eqref{eq:sys}, 
\eqref{eq:controller}
 to obtain \eqref{eq:closed-loop}.
Proposition
\ref{prop:periodic_solution} guarantees the existence of a periodic solution
$(x_p,e_p, \zeta_p)$ to 
system \eqref{eq:closed-loop} and therefore a periodic solution
$(x_p,e_p,z_p)$ to  \eqref{eq:sys}-\eqref{eq:controller}, with
$$
z_p = \zeta_p + Me_p
.
$$
With (\ref{eq:defPhiM}) and (\ref{eq:sequence_nzk}), it satisfies
\begin{align}  \notag
\sup_{t\in[0,T]}|x_p(t)| & \;\leq  \;
\min\left\{\boundx_p,\frac{\boundr_x}{\sqrt{\sigma }} \right\}
\; \leq \;
\boundx, 
\\ \label{eq:boundx_item2}
\sup_{t\in[0,T]}|e_p(t)| &\;\leq\;
\min\left\{\bounde,\frac{\boundr_e}{{\sqrt \sigma}}\right\}
\;\leq \;\bounde, 
\\ \notag
\sup_{t\in[0,T]}\sqrt{\zeta_p(t)^\top N_z \zeta_p(t)} &\;\leq\;
\boundzeta_p , 
\\ \notag
\sup_{t\in[0,T]}\sqrt{z_p(t)^\top N_z z_p(t)}
 & \leq 
\sup_{t\in[0,T]}\sqrt{(\zeta_p(t)+M e_p(t))^\top N_z (\zeta_p(t) + Me_p(t)}
\\ \notag
&\leq 
\sup_{t\in[0,T]}\sqrt{\zeta_p(t)^\top N_z \zeta_p(t)}
+
\sqrt{M^\top N_z M }\sup_{t\in[0,T]} |e_p(t)|
\\ \notag
&\leq 
\boundzeta_p +  \overline{N}_z\bounde_p
= \boundz
\end{align}
This establishes the inequalities \eqref{eq:theorem_periodic_bound} and \eqref{eq:theorem_2_bound_e1}
with $\psi_1:= \boundr_e$.

Moreover Proposition~\ref{prop:exponential_solution} 
guarantees the periodic solution
$(x_p,e_p, \zeta_p)$ is locally exponentially stable with a
domain of attraction that includes 
the set
$\N_\zeta(\boundx_a, \bounde_a,\boundzeta_a)$ defined in (\ref{eq:setN}).
Hence the domain of attraction of $(x_p,e_p,z_p)$ contains the set
\\[1em]$\displaystyle 
\Big\{(x,e,z): \; |x-x_p(0)| \leq  3\boundx_a,
\quad |e-e_p(0)| \leq  3\bounde_a,
$\hfill\null\\\null\hfill$\displaystyle
\sqrt{(z-z_p(0)-M(e-e_p(0))^\top N_z (z-z_p(0)-M(e-e_p(0))} \leq 3\boundzeta_a
\Big\}.
$\\[1em]
But we have the implications
\begin{eqnarray*}
\left\{|x|\leq 2\boundx\quad \& \quad |x_p(0)|\leq \boundx\right\}&\quad \Rightarrow\quad &
|x-x_p(0)|\leq |x|+|x_p(0)|\leq 3 \boundx
\\
\left\{|e|\leq 2\bounde\quad \& \quad |e_p(0)|\leq \bounde\right\}&\quad \Rightarrow\quad &
|e-e_p(0)|\leq |e|+|e_p(0)|\leq 3 \bounde
\end{eqnarray*}
$\displaystyle 
\left\{\sqrt{z^\top N_z z}\leq 2\boundz\quad \& \quad \sqrt{z_p(0)^\top N_z z_p(0)}\leq \boundz
\quad \& \quad |e|\leq 2\bounde\quad \& \quad |e_p(0)| \leq  \bounde\right\}
$\hfill\null\\\null \quad  $\displaystyle
\Rightarrow\   
\sqrt{(z-z_p(0)-M(e-e_p(0))^\top N_z (z-z_p(0)-M(e-e_p(0))}
$\hfill\null\\\null\hfill$\displaystyle
\renewcommand{\arraystretch}{1.7}
\begin{array}[t]{cl}
 \leq &
\sqrt{(z-z_p(0))^\top N_z (z-z_p(0))} + 
\sqrt{M^\top N_z M}|e-e_p(0)|
\\
 \leq &\sqrt{z^\top N_z z} + \sqrt{z_p(0)^\top N_z z_p(0)}
+ \overline N_z 3 \bounde
\\
 \leq &3(\boundz + \overline N_z\bounde)
 = 3(\boundzeta_p + 2 \overline N_z\bounde)
\leq 3\boundzeta_a
\end{array}
$\\[1em]
Hence the set $\N(\boundx, \bounde,\boundz)$
defined in \eqref{eq:theorem_periodic_domain} is contained in the above set and therefore in the domain of attraction of $(x_p,e_p,z_p)$.\qed

\subsection{Proof of Items 2)-4) of Theorem~\ref{thm1}}
First, note that item 2) of Theorem~\ref{thm1}
is a straightforward consequence of the inequalities 
\eqref{eq:prop_periodic_bound_e} and 
\eqref{eq:prop_periodic_bound_x}
 claimed in the statement 
of Proposition~\ref{prop:periodic_solution}, that 
is $\psi_1 := \boundr_e$ and $\psi_x := \boundr_x$.

Then, the $C^2$ periodic solution $(x_p,e_p,\zeta_p)$ given by Item 1) satisfies (see \eqref{eq:closed-loop})
\begin{equation}
\label{21}
\begin{array}{rcl}
\dot e_p & = & -\sigma e_p + \mu M^\top N_z \zeta_p + q(t,x_p,e_p),
\\
\dot \zeta_p & = & (\Phi - \mu M M^\top N_z)\zeta_p - M 
q(t,x_p,e_p).
\end{array}
\end{equation}
and it is
the sum of its Fourier series, i.e.
\begin{align*}
e_p(t) =  \sum_{k\in \ZZ}e_{pk} \exp(\im k\tfrac{2\pi}{T} t),
\\
\zeta_p(t) =  \sum_{k\in \ZZ}\zeta_{pk} \exp(\im k\tfrac{2\pi}{T} t),
\\
q(t,x_p(t),e_p(t)) = \sum_{k\in \ZZ}q_{pk} \exp(\im k\tfrac{2\pi}{T} t),
\end{align*}
where the index $k$ in $\{0,\ldots,\infty \}$ denotes the 
$k$-th Fourier coefficient. Because of \eqref{21},
$e_{pk}, \zeta_{pk}$  satisfy
\begin{equation}
\label{eq:fourier_e_zeta}
\renewcommand{\arraystretch}{1.5}
\begin{array}{rcl}
\left(ki{\textstyle\frac{2\pi}{T}}  +\sigma \right) e_{pk}&=&\displaystyle 
\mu M^\top N_z \zeta_{pk}
+ q_{pk}
\  ,
\\
\left(ki{\textstyle\frac{2\pi}{T}}  - \Phi +
\mu M M^\top N_z\right)\zeta_{pk} &=& - M q_{pk}
\end{array}
\end{equation}
When we are interested in expressing 
the $\zeta_p$ dynamics oscillator
by oscillator, it is appropriate to exhibit the corresponding components of $\zeta _{pk}$ as
$$
\zeta_{pk}=(\zeta_{pk0},\ldots ,\zeta_{pk\ell},\ldots\zeta_{pkn_o}),
$$
so where the index $\ell$, or $m$ below, in $\{0,\ldots,n_o\}$
refers to the $\ell$-th component of $\zeta_{pk}$.
With \eqref{eq:defPhiM}, \eqref{eq:fourier_e_zeta} becomes
\begin{equation}
\label{22}
\renewcommand{\arraystretch}{1.5}
\begin{array}{rcl}
(ki{\textstyle\frac{2\pi}{T}}  +\sigma ) e_{pk}&=&\displaystyle 
\mu \sum_{m=0}^{n_o} M_m^\top N_{zm} \zeta _{pkm} + q_{pk}
\  ,
\\\displaystyle 
   \left(\begin{array}{cc}
ki{\textstyle\frac{2\pi}{T}} &   \omega _\ell  \\ - \omega _\ell & k i {\textstyle\frac{2\pi}{T}}  
\end{array}\right)
\zeta _{pk \ell}+ \mu M_\ell
\sum_{m=0}^{n_o} M_m^\top N_{zm} \zeta _{pkm}&=& -M_\ell q_{pk}
\qquad \forall \ell\in \{0,1,\dots,n_o\}
\  .
\end{array}
\end{equation}
Assume that for some $\ell$ in $\{0,1,\ldots,n_o\}$, we have
$$
\omega _\ell\;=\; \mathfrak{K} _\ell {\textstyle\frac{2\pi}{T}}
$$
where $\mathfrak{K} _\ell$ is an integer. With (\ref{22}), this implies in particular
$$
\mathfrak{K} _\ell  {\textstyle\frac{2\pi}{T}}
 \left(\begin{array}{cc}
 \im &   1 \\ - 1 &  \im  
\end{array}\right)
\zeta _{p\mathfrak{K} _\ell \ell}\;=\; -M_\ell\left[ \mu 
\sum_{m=0}^{n_o} M_m^\top N_{zm} \zeta _{p\mathfrak{K} _\ell m}+ q_{p\mathfrak{K} _\ell}
\right]
\;=\; -M_\ell (k\im {\textstyle\frac{2\pi}{T}}  +\sigma ) e_{p\mathfrak{K} _\ell}
$$
Then, with the identities
$$
 \left(\begin{array}{cc}
 \im &   -1  
\end{array}\right)
 \left(\begin{array}{cc}
 \im &   1 \\ - 1 &  \im  
\end{array}\right)\;=\; 0
\quad ,\qquad 
 \left(\begin{array}{cc}
 \im &   -1  
\end{array}\right) M_\ell\;=\;  \left(\begin{array}{c}
 \im \\ 0
\end{array}\right),
$$
we obtain finally
$$
0\;=\; e_{p\mathfrak{K} _\ell}\;=\; \int_{0}^T \cos(\mathfrak{K}_\ell\tfrac{2\pi}T t)e_p(t)
+ \im\,  \int_{0}^T \sin(\mathfrak{K}_\ell\tfrac{2\pi}T t)e_p(t) . 
$$
This is \eqref{eq:FourierCoeff} in Item 3)

To establish Item 4), we note that \eqref{eq:fourier_e_zeta} gives
$$
\zeta _{pk}=
-
\left(k\im{\textstyle\frac{2\pi}{T}}  - \Phi +
\mu M M^\top N_z\right)^{-1}
Mq_{pk}
= -\frac{ ( k\im{\textstyle\frac{2\pi}{T}}  I-\Phi)^{-1} M}{1+\mu  M^\top N_z ( k\im{\textstyle\frac{2\pi}{T}}  I-\Phi)^{-1} M} q_{pk}
\  ,
$$
and therefore
\begin{align*}
e_{pk} & =  (q_{pk}+\mu  M^\top N_z \zeta _{pk})
\frac{1}{ki{\textstyle\frac{2\pi}{T}}  +\sigma } 
 = 
\frac{1}{
1+\mu  M^\top N_z ( k\im{\textstyle\frac{2\pi}{T}}  I-\Phi)^{-1} M}
\frac{1}{k\im{\textstyle\frac{2\pi}{T}}  +\sigma } 
q_{pk}.
\end{align*}
With the definitions \eqref{eq:defPhiM}, where $\omega _0=0$, we have
\begin{eqnarray*}
M^\top N_z ( k\im{\textstyle\frac{2\pi}{T}}  I-\Phi)^{-1} M&=&
\frac{n_{z0}}{k\im{\textstyle\frac{2\pi}{T}}}
+
\sum_{\ell=1}^{n_o}
\left(\begin{array}{cc}
1 & 0
\end{array}\right)n_{z\ell}
\left(\begin{array}{cc}
k\im {\textstyle\frac{2\pi}{T}} & \omega _\ell
\\
-\omega _\ell & k\im {\textstyle\frac{2\pi}{T}}
\end{array}\right)\left(\begin{array}{c}
1 \\ 0
\end{array}\right)
\\
&=&\im
k {\textstyle\frac{2\pi}{T}}
\sum_{\ell=0}^{n_o}\frac{n_{z\ell}}{\omega _\ell^2-k^2 [{\textstyle\frac{2\pi}{T}}]^2}
\end{eqnarray*}
and therefore, 
\begin{equation}
\label{eq:gain_ek}
|e_{pk}|^2\;=\; \frac{1}{1+\mu ^2\left[
k {\textstyle\frac{2\pi}{T}}
\displaystyle \sum_{\ell=0}^{n_o}
\textstyle \frac{n_{z\ell}}{\omega _\ell^2-k^2 [{\textstyle\frac{2\pi}{T}}]^2}
\right]^2}
\frac{|q_{pk}|^2}{k^2[{\textstyle\frac{2\pi}{T}}]^2  +\sigma ^2}
\end{equation}
Note that the first factor of the right hand side comes with
the presence of the internal model. Indeed, without this model,
we would get
$$
|e_{pk}|^2\;=\; 
\frac{|q_{pk}|^2}{k^2[{\textstyle\frac{2\pi}{T}}]^2  +\sigma ^2}.
$$
So now, if we suppose
that all 
$\omega_\ell$ 
in \eqref{eq:controller}
are selected
as 
$$
\omega_\ell = \ell \frac{2\pi}{T}, 
\qquad \forall \;  \ell \in \{1, \ldots, n_o\}.
$$
then, from \eqref{eq:gain_ek}, we obtain
$$
\begin{array}{rcll}
|e_{pk}|^2&=&0\qquad  & \forall |k|\leq n_o,
\\[.5em]
&\leq &
\frac{1}{(n_o+1)^2[{\textstyle\frac{2\pi}{T}}]^2  +\sigma ^2}
|q_{pk}|^2
\qquad & \forall |k|\geq n_o+1.
\end{array}
$$
Parseval's theorem then gives us
\begin{eqnarray*}
\dfrac{1}{T}\int_{0}^T |e_p(t)|^2 dt &=&
\sum_{k\in\ZZ} |e_{pk}|^2,
\\
&\leq &\frac{1}{(n_o+1)^2[{\textstyle\frac{2\pi}{T}}]^2  +\sigma ^2}
\sum_{k:|k|\geq n_o+1}|q_{pk}|^2
\\
&\leq &\frac{1}{(n_o+1)^2[{\textstyle\frac{2\pi}{T}}]^2  +\sigma ^2}
\dfrac{1}{T}\int_{0}^T |q(t,x_p(t),e_p(t))|^2 dt
\\
&\leq &\frac{1}{(n_o+1)^2}
\frac{\boundq(\boundx_p,\bounde_p)^2}{[{\textstyle\frac{2\pi}{T}}]^2 }
\end{eqnarray*}
This is \eqref{eq:theorem_2_bound_e2}, 
with $\psi_2:= [{\textstyle\frac{T}{2\pi}}]^2 \boundq(\boundx_p,\bounde_p)^2$.\qed

\section{Proof of Theorem~\ref{thm2}}
\label{sec:proof2}
\subsection{Preliminaries}
A key feature of the infinite dimensional dynamical regulator \eqref{eq:controller}
is that, when $e=0$, it 
can reproduce any desired
$C^1$ $T$-periodic function, as stated by the following 
lemma.

\begin{lemma}\label{lemma:friend}
For any $c\in \C^1_T(\RR)$
and any $\mu\geq1$, there exists 
$z_p \in \cL^2_{N_z}$ such that the solution 
to 
\begin{subequations}\label{eq:regulator_equation}
\begin{equation}
\label{eq:regulator_equation_z}
\dot z = \Phi z, 
\qquad z(0)= z_p
\end{equation}
is in $\C^1_T(\cL^2_{N_z})$ and satisfies
\begin{equation}
\label{eq:friend}
\mu M^\top N_z z(t) = c(t)
\qquad
\forall\ t\geq0
\end{equation}
\end{subequations}
with $\Phi, M, N_z$ defined in
\eqref{eq:defPhiM-inf}.
\end{lemma}

\begin{proof}
Since $c$ is $C^1$ and $T$-periodic, it can be expressed
as the sum of its
Fourier series as
$$
c(t) = c_0 + \sum_{k=1}^{\infty} c_k^c \cos(k\omega t)
+ c_k^s \sin(k\omega t), 
\qquad \omega = \dfrac{2\pi}{T}.
$$
Then, let the $k$-th component  $z_{pk}$ of the
 vector $z_p$ be
 selected as 
 $$
 z_{p0} = \dfrac{c_0}{\mu}, \qquad
 z_{pk}=  
 \dfrac{1}{\mu n_{zk}}  c_k, 
 \qquad
c_k =  \begin{pmatrix}
  c_k^c \\ c_k^s
\end{pmatrix}  
\qquad \forall \; k = 1, \ldots, \infty.
 $$
To verify that $z_p$ is in the set $\cL^2_{N_z}$ as defined in 
definition \eqref{eq:def_space_L2Nz}, we compute
\begin{equation}
\label{24}
\sum_{k=1}^\infty n_{zk} |z_{pk}|^2
 = 
\dfrac{1}{\mu^2} \sum_{k=1}^\infty  
\dfrac{1}{n_{zk}}|c_k|^2
\leq 
\dfrac{1}{\mu^2}
\sup_{k\in\NN_{>0}}
\dfrac{1}{k^2n_{zk}}
 \sum_{k=1}^\infty  
k^2|c_k|^2.
\end{equation}
Now, $\sup_{k\in\NN_{>0}} \dfrac{1}{k^2n_{zk}}$ is finite by assumption
\eqref{23}
as well as
$k|c_k|$ is square summable
since $c$, being $C^1$, $\frac{dc}{dt}$ is square integrable.

To verify \eqref{eq:friend} holds, we remind that,
component-wise, the elements of $z$ are given by the solution of the differential equations 
$$
\begin{array}{rclrcl}
\dot z_0 &= & 0,  
&z_0(0) & = &z_{p0}
\\[.5em]
\dot z_k &=& \Phi_k z_k, \qquad  
&z_k(0) & = &z_{pk} \qquad 
\forall \; k=1, \ldots, \infty,
\end{array}
$$
with $\Phi_k$ defined in
\eqref{eq:defPhiM-inf}.
Hence
$$
\begin{array}{rcl}
z_0 & = & z_{p0},
\\[.5em]
z_k(t) &=& 
\left(\begin{array}{cc}
\cos (k\omega t) 
&  \sin (k\omega t)
\\
-\sin (k\omega t)
&
\cos (k\omega t) 
\end{array}\right)
 z_{pk}
\qquad 
\forall \; k=1, \ldots, \infty.
\end{array}
$$
This implies $z\in \C^1_T(\cL^2_{N_z})$,
Then,
by using the definition 
of $M, N_z$ in
\eqref{eq:defPhiM},
we obtain
\begin{align*}
\mu M^\top N_z z(t) & =  
\mu  z_0(t)
+
\mu \sum_{k=1}^\infty n_{zk} M_k^\top \exp(\Phi_k t) z_{pk} 
\\
  & = 
  \mu  z_{p0}(t)
+
\mu \sum_{k=1}^\infty  n_{zk} [\cos (k\omega t) z_{pk}^c
+  \sin (k\omega t) z_{pk}^s]
\\
& = 
c_0 + 
\sum_{k=1}^{\infty} c_k^c \cos(k\omega t)
+ c_k^s \sin(k\omega t)
 = c(t)
\end{align*}
 for all $t\geq0$.
 This concludes the proof.\qed
\end{proof}

Now, let  us define 
\begin{equation}
\label{eq:friend_inf}
c(t) := -q(t,0,0).
\end{equation}
By construction $c\in\C^1_T([0,T];\RR)$.
Hence, let  $z_p$ be given with such $c$
by Lemma~\ref{lemma:friend} and
consider the following change of coordinates
\begin{equation}
\label{eq:change_of_coordinates_z_zeta_inf}
z\mapsto \zeta := z - z_p(t)- Me
\end{equation}
which
transforms
the closed-loop system 
\eqref{eq:sys}, \eqref{eq:controller}
into
\begin{equation}\label{eq:closed-loop-inf}
\begin{array}{rcl}
\dot x & = & F(t)x  
+ \delta f(t,x,e)
\\
\dot e & = & \Delta(t,x,e) - \sigma e + \mu M^\top N_z \zeta
\\
\dot {\zeta} & = &(\Phi - \mu M M^\top N_z)\zeta  - 
M\Delta(t,x,e),
\end{array}
\end{equation}
where $F,\delta f$ are defined as in \eqref{eq:f_decomposition}, 
namely
$$
\delta f(t,x,e):= f(t,x,e) - F(t)x, 
\qquad
F(t) :=  \frac{\partial f}{\partial x}(t,0,0),
$$
and $\Delta$ is a $C^1$ $T$-periodic function defined as
$$
\Delta(t,x,e) : = q(t, x, e)- q(t,0,0).
$$
and satisfying
\begin{equation}
\label{eq:boundDelta-inf}
|\Delta(t,x,e)| \leq \boundq_e(\boundx,\bounde)|e|
+
\boundq_x(\boundx,\bounde)|x|
\end{equation}
for any $(x,e)\in\cS_T(\boundx,\bounde)$.

The metric space $\cL^2_{N_z}$ being complete
and all functions being locally Lipschitz, 
 we are guaranteed about existence of solutions to \eqref{eq:closed-loop-inf}, as stated by the following
known result on existence of solutions
for differential equations in Banach spaces proved for example in
\cite[\S 1.1]{Deimling},
\cite[Theorem  1.8.3]{Cartan}
or
\cite[Proposition VI.1.2 \& Theorem VI.3.1]{Martin}.

\begin{proposition}\label{prop:existence_solution_inf}
Let $D$ be an arbitrary open bounded subset of $\RR^n\times\RR\times \times \cL^2_{N_z}$, then for any
$(x(t_0)$, $e(t_0)$, $\zeta (t_0))$ in $D$ and any $t_0$, system \eqref{eq:closed-loop-inf} has a unique solution 
$(x(t)$, $e(t)$, $\zeta (t))$ with values in $D$, defined 
either on $[t_0,+\infty [$ or only on an interval $[t_0,t_0+\bar \tau [$ which is maximal in the sense that
\begin{equation}
\label{bartaufini}
\lim_{t\to\bar \tau } d( (x(t),e(t),\zeta (t))\,  ,\,  \partial D)\;=\; 0,
\end{equation}
where $d$ is the distance on $\RR^n\times\RR \times \cL^2_{N_z}$ and $\partial D$ is the boundary of 
$D$.
\end{proposition}

\subsection{Stability Analysis}
 
 We have the following proposition.

\begin{proposition}\label{prop:exponential_solution_inf}
Given any triplet $(\overline P_x,\underline P_x,\alpha)$
 there exist
strictly positive
real numbers  $\boundx_\infty, \bounde_\infty>0$
and, for any $\boundzeta_\infty>0$
and $\mu\geq1$, 
there exists $\sigma^\star_\infty>0$
such that, the following holds.
\begin{enumerate}
\item[(i)] Solutions of \eqref{eq:closed-loop-inf}
starting from the set 
\begin{equation}
\label{eq:Omega_init}
\Omega_{\init}=
\Big\{(x,e,\zeta)\in \RR^n\times \RR \times \cL^2_{N_z} : \; |x| \leq  3\boundx_a,
\; |e| \leq 3\bounde_a, \;
\|\zeta\|_{N_z} \leq 3\boundzeta_a
\Big\}
\end{equation}
are complete and bounded in positive times and
 $\zeta(t)\in \cL^2_{N_z}$ for all $t\geq0$.
\item[(ii)] The origin of \eqref{eq:closed-loop-inf}
is Lyapunov stable with domain of attraction 
containing
$\Omega_{\init}$.
\item[(iii)] Solutions starting from 
$\Omega_{\init}$ satisfies
$$
\lim_{t\to\infty} x(t) = 0, 
\qquad
\lim_{t\to\infty}e(t) = 0.
$$
\end{enumerate}
\end{proposition}

\begin{proof}~We follow similar steps of the proof 
of Proposition~\ref{prop:exponential_solution}.
First, let us define
the non-decreasing function $\rho_0^\infty$ as
\begin{equation}
\label{eq:rho_0_inf}
\rho_{0}^\infty(\boundx) :=  \dfrac{\overline P_x}2
\boundf_{xx}(\boundx,0)\boundx \qquad s\geq0.
\end{equation}
Let $\boundx_b$, $\bounde_b$ be parametrized as follows
\begin{equation}
\label{eq:temp_bound_xe_inf}
\boundx_b:= \left(
3{\sqrt2} \sqrt{\frac{\overline P_x}{\underline P_x}}+ 2 \right)\boundx_a, 
\qquad
\bounde_b := 6 \bounde_a.
\end{equation}
and let 
$\boundx_a, \bounde_a$ be any pair of positive real numbers 
satisfying
\begin{equation} 
\label{eq:bound_xe_inf}
\begin{array}{rcl}
9\boundx_\infty^2  \overline P_x + 9 \bounde_\infty^2 
& \leq &
\dfrac12\min\big\{\underline P_x \boundx_b^2,\; \bounde_b^2\big\},
\\[1em]
\rho_{0}^\infty(\boundx_b) & \leq &\tfrac{1}{3}\alpha \underline P_x
\end{array}
\end{equation}
Now, given any $\boundzeta_a>0$,
 let $h>0$ be
 selected such that
\begin{equation} 
\label{eq:bound_h_inf}
h \leq  \dfrac{1}{3} 
\min\left\{
\dfrac{1}{(3\boundzeta_a)^2}
\min\big\{\underline P_x \boundx_b^2,\; \bounde_b^2\big\},
\;
\mu \dfrac{\alpha\underline P_x}{
\boundq_x(\boundx_b,\bounde_b)^2}
\right\}.
\end{equation}
Finally, let $\sigma_\infty^\star >0$ be chosen as 
\begin{equation}
\label{eq:bound_sigma_inf}
\sigma_\infty^\star :=
\frac{
\left[\overline{P}_x \boundf_e(\boundx_b, \bounde_b)
+\boundq_x(\boundx_b,\bounde_b)\right]^2
}{
2\left[\alpha \underline P_x- \rho_0(\boundx_b) - 
\tfrac{h}{\mu}\boundq_x(\boundx_b, \bounde_b)^2\right]
}+\frac{\mu }{2h}+  \boundq_e(\boundx_b, \bounde_b)
+\dfrac{h}{\mu}\boundq_e(\boundx_b, \bounde_b)^2.
\end{equation}
Then, by following the same steps, consider the functions
$$
V_x := x^\top P_x
(t)
 x,  
\qquad
V_e := e^2, 
\qquad
V_{\zeta}:= \zeta^\top N_z \zeta.
$$
By using \eqref{eq:Pass2}, 
and  \eqref{app:bound_deltaf},
their derivatives satisfy,
for all $(x,e)$ in $\cS(\boundx_b,\bounde_b)$, $t\in[0,T]$,
\begin{align*}
\dot V_x  \leq &  -2 \alpha x^\top P_x(t)x + 2x^\top P_x(t)\delta f(t,x,e)  
\\
\leq &  -2 \alpha \underline P_x |x|^2 +
2 |x| \overline P_x\boundf_e(\boundx_b,\bounde_b)\,  |e|
+
\overline P_x \boundf_{xx}(\boundx_b,0)\,  | x|^3 
\\
\leq &  -2 [\alpha \underline P_x   - \rho_0^\infty(\boundx_b)]|x|^2 +
2  \overline P_x\boundf_e(\boundx_b,\bounde_b)\,  |e| \, |x|
\end{align*}
with $\rho_0^\infty$ defined as in \eqref{eq:rho_0_inf},
and 
\begin{align*}
\dot V_e
\leq &
-2\sigma |e|^2 + 2\mu M^\top N_\zeta \zeta e + 
2\boundq_e(\boundx_b,\bounde_b)|e|^2 + 
2\boundq_x(\boundx_b,\bounde_b)|x| \, |e|,
\\[.5em]
\dot V_\zeta \leq  &
- \mu |M^\top N_\zeta \zeta|^2 + 
\dfrac{2}{\mu} 
\boundq_e(\boundx_b,\bounde_b)^2|e|^2 + 
\boundq_x(\boundx_b,\bounde_b)^2|x|^2.
\end{align*}
Therefore, 
collecting all the inequalities
together, 
and by defining
$$
U :=  V_x + V_e + h V_\zeta, 
\qquad  \chi := (| x|, | e|,
|M^\top N_z \zeta|)^\top ,
$$
with $h$ selected as in \eqref{eq:bound_h_inf},
the time derivative of the Lyapunov function $U$
satisfies
$$
\dot U \leq - \chi^\top \R^\infty(\boundx_b, \bounde_b, \sigma, \mu, b) \chi
\qquad \forall (x,e,\zeta)\in \cS_T(\boundx_b, \bounde_b)\times \cL^2_{N_z},
$$
with $\R^\infty$ defined as  
\begin{equation}
\label{eq:R-lyap1_inf}
\R^\infty(\boundx_b, \bounde_b,\sigma,\mu,h) := \begin{pmatrix}
2(\alpha \underline{P}_x -\rho_{11}^\infty(\boundx_b,\bounde_b,\mu,h)) && 
 -\rho_{12}^\infty(\boundx_b,\bounde_b) && 0 
\\
-\rho_{12}^\infty(\boundx_b,\bounde_b) && 
2(\sigma-\rho_{22}^\infty(\boundx_b,\bounde_b,\mu,h)) && -\mu
\\
 0&& -\mu && \mu h
\end{pmatrix},
\end{equation}
and where the functions $\rho_{11}^\infty$,
$\rho_{12}^\infty$ and $\rho_{22}^\infty$ are 
non-decreasing functions
defined as
$$
\label{eq:R-lyap2_inf}
\begin{array}{rl}
\rho_{11}^\infty(\boundx_b,\bounde_b,\mu,h)  := & 
 \rho_0^\infty(\boundx_b) + 
\tfrac{h}{\mu}\boundq_x(\boundx_b, \bounde_b)^2,
\\
\rho_{12}^\infty(\boundx_b,\bounde_b) : = & 
 \overline{P}_x 
\boundf_e(\boundx_b, \bounde_b)
+\boundq_x(\boundx_b,\bounde_b),
\\
\rho_{22}^\infty(\boundx_b,\bounde_b,\mu,h)  : = &
  \boundq_e(\boundx_b, \bounde_b)
+\tfrac{h}{\mu}\boundq_e(\boundx_b, \bounde_b)^2 .
\end{array}
$$
According to Sylvester's criterion, the matrix $\R^\infty$
is positive definite if and only if the leading principal minors
are all strictly positive, that is, we need to satisfy
\begin{eqnarray}
\label{16_inf}
\alpha \underline P_x - \rho_{11}^\infty& >0,
\\
\label{17_inf}
4(\alpha \underline P_x - \rho_{11}^\infty)(\sigma-\rho_{22}^\infty)-(\rho_{12}^\infty)^2 &>0,
\\
\label{18_inf}
2(\alpha \underline{P}_x -\rho_{11}^\infty)
\left[2(\sigma-\rho_{22}^\infty) 
 h
-
\mu
\right]
-
h
(\rho_{12}^\infty)^2
&>0,
\end{eqnarray}
where the arguments of the functions $\rho_{ij}^\infty$
have been omitted for compactness.
These inequalities are always satisfied
for any $\sigma>\sigma ^\star_\infty$ since
\begin{itemize}
\item[$\bullet$]
\eqref{16_inf} is implied by \eqref{eq:bound_xe_inf} and \eqref{eq:bound_h_inf}.
\item[$\bullet$]
\eqref{17_inf} is implied by \eqref{16_inf} and \eqref{18_inf}.
\item[$\bullet$]
\eqref{18_inf} is implied by \eqref{eq:bound_sigma_inf}.
\end{itemize}
As a consequence, for any $\sigma>\sigma ^\star_\infty$,
there exists a  positive real number
 $\varepsilon>0$
such that 
\begin{equation}
\label{eq:ineqU_inf}
\dot U \leq -\varepsilon(|x|^2 +|e|^2 + |M^\top N_z\zeta|^2)
\qquad
\forall (x,e,\zeta)\in \cS(\boundx_b, \bounde_b)\times \cL^2_{N_z}.
\end{equation}
Now, let  $\overline{\upsilon}$ be the strictly positive real number defined as
$$
\overline{\upsilon} := \min_{(x,e,\zeta)\not \in \Omega_{\init}}
|x|^2 \overline{P}_x + | e|^2 + h\zeta ^\top N_z \zeta 
 \leq (3\boundx_a)^2 \overline{P}_x +
  (3\bounde_a)^2 + h (3\boundzeta_a)^2
$$
and define
$$
\Omega_{\max} := \big\{(x,e,\zeta)\in \RR^n\times \RR\times \cL^2_{N_z}:
U(x,e,\zeta,t)\leq \overline{\upsilon}, \ t\in [0,T]
\big\}.
$$
In view of the choices of $\boundx_a,\boundx_b,h$ 
in
\eqref{eq:bound_xe_inf}, \eqref{eq:bound_h_inf},
we have  
$\overline{\upsilon}\leq \tfrac56\min\big\{\underline P_x \boundx_b^2,\; \bounde_b^2\big\}$
 and therefore it is straightforward to see that
$$
\Omega_{\max} \subset \cS(\boundx_b, \bounde_b)\times \cL^2_{N_z}.
$$
Now let $D$ be an open bounded superset of $\Omega_{\max}$
contained in  $\cS(\boundx_b, \bounde_b)\times \cL^2_{N_z}$.
Let $(x(t),e(t),\zeta (t))$ be the unique solution to 
system \eqref{eq:closed-loop-inf},
established by Proposition~\ref{prop:existence_solution_inf},
 with initial condition in $\Omega_{\init}$
and right maximally defined on $[0,\bar \tau [$ on $D$.
With \eqref{eq:ineqU_inf}, we have established
$$
U(x(t),e(t),\zeta(t),t)\leq  U(x(0),e(0),\zeta(0),0) < \overline{\upsilon}
\qquad  \forall t \in [0,\bar \tau [,
$$
and therefore
$$
(x(t),e(t),\zeta(t)) \in \Omega_{\max} 
\  ,\quad 
\zeta(t)\in \cL_{N_z}^2
\qquad  \forall t \in [0,\bar \tau [.
$$
But $\Omega_{\max}$ being a closed subset of the open set $D$, we have
$$
d(\Omega _{\max},\partial D) >0
$$
and therefore a contradiction with (\ref{bartaufini}) if $\bar \tau $ is finite.
So the solution $(x(t),e(t),\zeta (t))$ is defined on $[0,+\infty [$ and takes it values in $\Omega _{\max}$.
It follows that the function $t\mapsto ((x(t),e(t),\zeta (t),M^\top N_z\zeta (t))
\in \RR^n\times \RR\times  \cL_{N_z}^2 \times\cL_{N_z}^2$ is bounded
and
$$
|q(t,x(t),e(t))|\leq \boundq(\bar\boundx,\bar\bounde)\qquad \forall t\in [0,+\infty [
\  .
$$
This shows item (i).

Next, consider the set of functions 
$(x,e,\zeta)$ so that $\overline U$ is bounded, with $\overline U$
defined as
$$
\overline U(x,e,\zeta):= |x|^2 \overline P_x + |e|^2 
+ h \zeta^\top N_z \zeta.
$$
We have that $\sqrt{\overline U}$ defines a norm.
Now, for any $\epsilon>0$ small enough, 
let  $\delta = \epsilon \tfrac{\underline P_x }{\overline P_x}$, 
and consider any initial condition 
$(x(0),e(0),\zeta(0))\in\Omega_{\init}$
satisfying
$$
\overline U(x(0),e(0),\zeta(0))\leq \delta.
$$
Using \eqref{eq:ineqU_inf} along 
solutions, we obtain
$$
\begin{array}{rcl}
\overline U(t) &\leq & \dfrac{\overline P_x}{\underline P_x} U(t)
\; \leq  \; \dfrac{\overline P_x}{\underline P_x} U(0)
\; \leq \;  \dfrac{\overline P_x}{\underline P_x}\overline U(0)
\; \leq \;\epsilon.
\end{array}
$$ 
This shows that that the origin of \eqref{eq:closed-loop-inf}
is Lyapunov stable 
with a domain of attraction which includes
the set $\{(x,e,\zeta):\overbar U(0) \leq \delta_{\max}\}$, 
where  $\delta_{\max}>0$ is 
the smallest $\delta$ such that 
$\overbar U(x^*, e^*, \zeta^*) > \delta$ implies 
$(x^*, e^*, \zeta^*)\not\in\Omega_{\init}$.
This shows item (ii) of the theorem.

Finally, by using  \eqref{eq:ineqU_inf} along 
solutions, we obtain
$$
 \lim_{\tau \to \infty }
\int_0^\tau 
\varepsilon(|x(t)|^2 + e^2(t) )
dt\; \leq \; U(0)\,.
$$
On the left hand side each term in the integrand is a function with 
non negative values, and, 
since, according to \eqref{eq:closed-loop-inf},
the function $t\mapsto (\dot x(t),\dot e(t))$ is bounded, the function $t\mapsto  (x(t),e(t)$
is Lipschitz (and therefore uniformly continuous).
Hence by applying Barbalat's lemma 
we get
$$
\lim_{t\to\infty} |e(t)| = 0\;, \qquad \lim_{t\to\infty}|x(t)|=0 \;, 
$$
which shows item (iii) and concludes the proof.
\qed
\end{proof}

\subsection{Proof of Theorem~\ref{thm2}}
The statement of Theorem~\ref{thm2} follows
by combining the results of
Lemma~\ref{lemma:friend} and 
 Propositions \ref{prop:existence_solution_inf}
 and \ref{prop:exponential_solution_inf}.
To this end, let  $c(t)= -q(t,0,0)$ be given as in 
\eqref{eq:friend_inf} and apply Lemma~\ref{lemma:friend}
to obtain the initial condition $z_p\in \cL^2_{N_z}$. 
By further using \eqref{24} and the properties of 
$n_{z\ell}$ in \eqref{23},
we also have 
$\|z_p\|_{N_z}\leq \boundz_\infty$, 
with $\boundz_\infty = \dfrac{1}{\mu^2n_{z1}} \boundq_t(0,0)$.
Now, let $\bounde$
be given  by Proposition~\ref{prop:exponential_solution_inf}, 
and select $\boundzeta_\infty  = 2\boundz_\infty + 2\overline N_z \bounde_\infty$.
With such choice, we ensure that both initial conditions 
$z(0) = 0$ and $z(0) = z_p$ belong the set
$\Omega_{\init}$ defined in \eqref{eq:Omega_init}
and are therefore included 
 in the domain 
of attraction
in view of item (ii) of 
Proposition~\ref{prop:exponential_solution_inf}.
Indeed,  following similar computations
to those used in the proof of item 1) of Theorem~\ref{thm1},
we can show that
$$
\Big\{(x,e,z)\in \RR^n\times \RR\times \cL^2_{N_z}: |x|\leq 2\boundx_\infty, \
|e|\leq 2\bounde_\infty, \ \|z\|_{N_z}\leq 2\boundz_\infty \Big\} 
\ \subseteq \ \Omega_{\init}.
$$
Note that such initial conditions are of particular interest, since $z(0)=0$
corresponds to the origin of the control law (which is a quite natural
choice when the steady-state solution for $z$ is unknown), while $z(0)= z_p$ corresponds
to the origin of the closed-loop system \eqref{eq:closed-loop-inf}.
Now, by applying Proposition~\ref{prop:existence_solution_inf}
and \ref{prop:exponential_solution_inf}, we are guaranteed
that,  for $\sigma>\sigma^\star_\infty$,
the solution to system \eqref{eq:closed-loop-inf} to be complete
forward in time. Furthermore, as 
 $\zeta(t)$ is in $\cL_{N_z}^2$ for all $t\geq 0$, 
 and $M$ is in  $\cL_{N_z}^2$ by design, 
 we conclude that also  $z(t)$ is in $\cL_{N_z}^2$ for all $t\geq 0$.
 Moreover, $x$ and $e$ converges asymptotically to zero.

Finally, 
the last part of the proof consists in showing that 
the domain of attraction given by Theorem~\ref{thm1}
with a finite-dimensional regulator, 
is not reduced  in terms of $(x,e)$-coordinates
  with the infinite-dimensional regulator.
  Moreover, for such a bound, we want also 
$\sigma^\star_\infty< \sigma^\star$.
In turn, by entering the details of the proof of 
Theorem~\ref{thm1}, 
this consists in showing that for
 Proposition~\ref{prop:existence_solution_inf}
we can select the bounds 
$\boundx_a, \bounde_a$
and $\boundzeta_a$
given by 
Proposition~\ref{prop:exponential_solution}, 
and then compare $\sigma_a^\star$  with the resulting 
 $\sigma^\star_\infty$.
To this end, 
it suffices first to note that the bounds 
in \eqref{eq:bound_xe_inf} are 
less conservatives than 
\eqref{eq:boundx_bounde} for two main reasons:
in the first pair of equations, the right term is 
less conservative since we have 
$\boundx_b$ (resp. $\bounde_b$)
and not $\boundx_b-\boundx_a$ (resp. 
 $\bounde_b-\bounde_a$);
the second condition is less conservative in view of the definition of
$\rho_0^\infty$ 
given in \eqref{eq:rho_0_inf} 
compared to that of $\rho_0$ in \eqref{eq:def_rho_0}. 
As a consequence, any selection of
$\boundx_a, \bounde_a, \boundzeta_a$ given by 
Proposition~\ref{prop:exponential_solution}
is a feasible choice for the statement of 
Proposition~\ref{prop:exponential_solution_inf}.
This ensure that the domain of attraction 
in terms of $(x,e)$ coordinates
for the infinite-dimensional regulator contains 
the domain of attraction obtained in the finite-dimensional case, 
and therefore that the bounds $\boundx, \bounde, \boundz$
of Theorem~\ref{thm2} can be taken from Theorem~\ref{thm1}.
Hence, to complete that statement of the theorem, 
it suffices to show that for a given set of values, 
$\sigma_\infty^\star < \sigma_a^\star$
given by Proposition~\ref{prop:exponential_solution}.
Again, this can be easily proved by 
noting that the definition of 
$h$ in \eqref{eq:bound_h_inf}
is less conservative than the bound 
\eqref{eq:ineq_bounds_b}, 
and therefore, for a given $h$, 
we have
$\sigma_\infty^\star<\sigma_a^\star$
since $\rho_0^\infty(\boundx_b)< \rho_0(\boundx_a, \bounde_b)$, 
see \eqref{eq:ineq_bounds_sigma} and
\eqref{eq:bound_sigma_inf}.
\qed

\section{Supplementary Material: a Technical Lemma}

\begin{lemma}
\label{lemma:transfer_function}
Let $\mu>0$ and  let $n_{z\ell }$  be a sequences of positive real numbers 
satisfying \eqref{eq:sequence_nzk}.
 There exist strictly positive real numbers $\kappa_0$ and $\kappa_1$
independent of $n_o$  such that  
the function $\mathcursive{T}(x)$ defined  as in 
\eqref{eq:zeta_pk_simplified}, i.e. 
\begin{equation*}
\mathcursive{T}(x)\;=\; 
\dfrac{\displaystyle
\sum_{\ell=0}^{n_o} 
\frac{\ell^2+x^2}{(\ell^2 - x^2)^2}n_{z\ell}
}
{\displaystyle
1 + \mu^2   x^2
\left(
\sum_{\ell=0}^{n_o}
 \frac{1}{\ell^2-x^2} n_{z\ell}
\right)^2
}
\end{equation*}
satisfies
\begin{equation}
\label{32}
\mathcursive{T}(x)
\; \leq  \;
\kappa_0 + \kappa_1 x^2
\qquad
\forall \; x \in \RR, \; n_o \in \NN, \;  \mu\geq 1.
\end{equation}
\end{lemma}

\begin{proof}{}
Seemingly the function $\mathcursive{T}$ has singularities at $x=\ell$. But having~:
$$
\mathcursive{T}(x)\;=\; 
\frac{\displaystyle 
\sum_{\ell=0}^{n_o}
\left(\ell  ^2 + x^2\right)\,  n_{z\ell}
\left[\prod_{m=0,\neq \ell}^{n_o}\left(
m^2-x^2\right)\right]^2
}{\displaystyle 
\left[\prod_{m=0}^{n_o}\left(
m^2-x^2\right)\right]^2
+
\mu ^2 x^2\left[\sum_{\ell=0}^{n_o}
n_{z\ell}
\prod_{m=0,\neq \ell}^{n_o}\left(
m^2-x^2\right)\right]^2
}\,
$$
we conclude that these singularities are fictitious and actually $\mathcursive{T}$ is a $C^\infty $ function
of $x$ on $\RR$. Moreover we have 
$$
\mathcursive{T}(\ell)
\; = \; \frac{2}{\mu ^2 n_{z\ell}}
\qquad \forall \ell\in  \{0,\ldots,n_o\}
\ .
$$
With (\ref{23}) this yields
\begin{equation}
\label{eq:bound_zetaNzeta_exact}
\mathcursive{T}(\ell)
\; \leq  \; \frac{2\ell^2}{\mu ^2 n_{z1}}
\qquad \forall \ell\in  \{0,\ldots,n_o\}
\ .
\end{equation}

With this at hand, we continue in studying 
$\mathcursive{T}$ for $x$ non integer.

To isolate the (fictitious) 
singularities at $x=m$ or $x=m+1$, with $m\in \{0,\ldots, n_o-1\}$, we consider the sums without these values 
by defining
\begin{eqnarray*}
S_m(x)&:=& \underline{S}_m(x) - \overline{S}_m(x) 
\\
\underline{S}_m(x)&:=&\sum_{\ell=0}^{m-1}
\frac{1}{x^2-\ell^2}\frac{n_{z\ell}}{n_{zm}}  
\quad ,\qquad 
\underline{S}_0(x)\;=\; 0
\ ,
\\
\overline{S}_m(x) &:=&\sum_{\ell=m+2}^{n_0}
\frac{1}{\ell^2-x^2}\frac{n_{z\ell}}{n_{zm}}  
\quad ,\qquad 
\overline{S}_{n_o-1}(x)
\;=\; 0 \ ,
\end{eqnarray*}
and
$$
Q_m(x) \; := \;\sum_{\ell=0, \neq m, \neq m+1}^{n_0}
\frac{x^2+\ell^2}{(x^2-\ell^2)^2}n_{z\ell}  \ .
$$
To simplify these expressions, we introduce the notation
$$
\mathcursive{r}(\ell) \; := \; \dfrac{n_{z(\ell+1)}}{n_{z\ell}} \ .
$$
With (\ref{27}) and (\ref{23}), we have
\begin{equation}
\label{boundR}
\frac{1}{4} \leq \frac{\ell^2}{(\ell+1)^2}\leq \mathcursive{r}(\ell) < 1
\qquad \forall 0 < \ell
\  .
\end{equation}
Then
\eqref{eq:zeta_pk_simplified}
is simplified into
\begin{equation}
\label{eq:zeta_pk_simplified3}
\mathcursive{T}(x)
 \; =  \; F_{n_o}(x)+ G_{n_o}(x) 
\end{equation} 
where
\begin{equation}
\label{eq:Fno}
F_{n_o}(x) \; := \; 
\dfrac{
n_{zm}
\left[\dfrac{x^2+m^2}{(x^2-m^2)^2}
+\dfrac{x^2+(m+1)^2}{((m+1)^2-x^2)^2}\mathcursive{r}(m)
\right]
}{1+\mu^2 x^2 n_{zm}^2
\left(S_m(x) + \dfrac{1}{x^2-m^2} 
-
\dfrac{1}{{(m+1)}^2 - x^2}\mathcursive{r} (m) 
\right)^2 
}\ ,
\end{equation}
\begin{equation}
\label{eq:Gno}
G_{n_o}(x) \; := \; 
\dfrac{Q_m(x)}{1+\mu^2 x^2 n_{zm}^2
\left(S_m(x) + 
\dfrac{1}{x^2-m^2} 
-
\dfrac{1}{{(m+1)}^2 - x^2}\mathcursive{r} (m) 
\right)^2 
} \ .
\end{equation}
Let now study the behavior of the functions 
$F_{n_o}, G_{n_o}$.
We divide the analysis into three cases,
where, in 
any case,
$x$ is not an integer
\begin{itemize}
\item[\textbf{a)}] 
$0 < x < n_o$,
\item[\textbf{b)}] 
$n_o < x <n_o+1$,
\item[\textbf{c)}] $n_o+1 \leq x$.
\end{itemize}

\paragraph*{Case a) $0 < x < n_o$.}~ 
Then there exists $m\in \{0, \ldots, n_o-1\}$ such that $x\in(m, m+1)$.
First, we look for a bound for
$
\overline{S}_m(m).
$
With (\ref{26}), this gives
$$
\overline{S}_m(m)\leq m\sum_{\ell=m+2}^{n_o}
\frac{1}{\ell^2-m^2}\frac{1}{\ell}
$$
We have
$$
\frac{1}{\ell^2-m^2}\frac{1}{\ell}
\;=\; \frac{1}{m^2}\left[\frac{1}{2(\ell -m)}+\frac{1}{2(\ell +m)}-\frac{1}{\ell}\right]
$$
Since the function $s\mapsto \frac{1}{2(s -m)}+\frac{1}{2(s +m)}-\frac{1}{s}$ is decreasing on 
$[\frac{m}{\sqrt{3}},+\infty [$, we get
\begin{eqnarray*}
\overline{S}_m(m)
&\leq &\frac{1}{m}\int_{m+1}^\infty 
\left[\frac{1}{2(s -m)}+\frac{1}{2(s +m)}-\frac{1}{s}\right] ds
\\
&\leq &\frac{1}{m}\  
\left.\log\left(\frac{\sqrt{s^2-m^2}}{s}\right)\right|_{m+1}^\infty 
\\
&\leq &
\frac{1}{m}\  
\log\left(\frac{m+1}{\sqrt{2m+1}}\right)\; \leq \; \frac{1}{m}\  
\log\left(\sqrt{m+1}\right)\; \leq \; \frac{1}{2}
\  .
\end{eqnarray*}
So the sequence $\overline{S}_m(m)$ is bounded independently of $n_o$.

Then, we look for a bound for
$\underline{S}_m(m)$. For $m\geq 2$, we have the decomposition
$$
\underline{S}_m(m) =  \sum_{\ell=0}^{m-1}
\frac{1}{m^2-\ell^2}\frac{n_{z\ell} }{n_{zm}}
$$
With (\ref{23}), this gives
$$
\underline{S}_m(m) \leq 
\frac{1}{m^2}\frac{n_{z0} }{n_{zm}}
+
\frac{1}{m^2-1}\frac{n_{z1} }{n_{zm}}
+
\sum_{\ell=2}^{m-2}
\frac{1}{m^2-\ell^2}\frac{m^2}{\ell^2}
+
\frac{1}{m^2-(m-1)^2}\frac{n_{z(m-1)} }{n_{zm}}
$$
Since the function $s\mapsto  \frac{1}{m^2-s^2}\frac{m^2}{s^2}$ is decreasing on
$]0,\frac{m}{\sqrt{2}}]$ and increasing on $[\frac{m}{\sqrt{2}},m[$, we have
$$
\sum_{\ell=2}^{m-2}
\frac{1}{m^2-\ell^2}\frac{m^2}{\ell^2}\leq 
\int_1^{m-1}
\frac{1}{m^2-s^2}\frac{m^2}{s^2} ds
$$
Then, with
$$
\frac{1}{m^2-s^2}\frac{m^2}{s^2} \;=\; 
\frac{1}{s^2}+\frac{1}{2m}\left(\frac{1}{m-s}+\frac{1}{m+s}\right)
$$
we obtain
\begin{eqnarray*}
\int_1^{m-1}
\frac{1}{m^2-s^2}\frac{m^2}{s^2} ds
&=&
\left.\left[
-\frac{1}{s}
+\frac{1}{2m}\log\left(\frac{m+s}{m-s}\right)\right]\right|_1^{m-1}
\\
&=&
1-\frac{1}{m-1}
+\frac{1}{2m}
\log\left(\frac{(2m-1)(m-1)}{m+1}\right)
\end{eqnarray*}
This yields
\begin{eqnarray*}
\underline{S}_m(m) 
&\leq &
\frac{1}{m^2}\frac{n_{z0} }{n_{zm}}
+
\frac{1}{m^2-1}\frac{n_{z1}}{n_{zm}}
+
\left[
\frac{m-2}{m-1}
+\frac{1}{2m}
\log\left(\frac{(2m-1)(m-1)}{m+1}\right)
\right]
 +
\frac{1}{m^2-(m-1)^2}\frac{n_{z(m-1)} }{n_{zm}}
\  ,
\\
&\leq &\frac{1}{m^2}\frac{n_{z0} }{n_{zm}}
+
\frac{1}{m^2-1}\frac{n_{z1}}{n_{zm}}
+
2
 +
\frac{1}{m^2-(m-1)^2}\frac{n_{z(m-1)} }{n_{zm}}
\end{eqnarray*}
But with (\ref{23}), (\ref{27}) and $m\geq 2$, we have
$$
\frac{1}{m^2}\frac{n_{z0} }{n_{zm}} \leq \frac{n_{z0} }{n_{z1}} 
\quad ,\qquad 
\frac{1}{m^2-1}\frac{n_{z1}}{n_{zm}} \leq \frac{n_{z0} }{3n_{z2}}
\quad ,\qquad 
\frac{1}{m^2-(m-1)^2}\frac{n_{z(m-1)} }{n_{zm}}\leq \frac{m^2}{(m-1)^2(2m-1)}\leq \frac{4}{3}
$$
Hence,  $\underline{S}_m(m)$ is bounded 
independently of $n_o$ as
$$
\underline{S}_m(m) 
\; \leq \; \frac{n_{z0} }{n_{z1}}
+
\frac{n_{z0} }{3n_{z2}}
+
\frac{10}{3}
\  .
$$
In conclusion we have obtained the real number
\begin{equation}
\label{29}
\mathbf{S}\;=\; \frac{n_{z0} }{n_{z1}}
+
\frac{n_{z0} }{3n_{z2}}
+
\frac{10}{3}
\  ,
\end{equation}
independent of $n_o$, such that
\begin{equation}
\label{76}
|S_m(x)| 
\leq \left|\sum_{\ell=0}^{m-1}
\frac{1}{m^2-\ell^2}\frac{n_{z\ell} }{n_{zm}}
-
 \sum_{\ell=m+2}^{n_o}
\frac{1}{\ell^2-m^2}\frac{n_{z\ell} }{n_{zm}}\right|
\leq \mathbf{S}
\qquad \forall m\in \{0,\dots,n_o-1\}\  ,\quad \forall n_o
\end{equation}
Finally we look for a bound for $Q_m(x)$.
We decompose, with no care about the boundary effects for $m<2$ or $m>n_o-3$,
\\[1em]\null \quad $\displaystyle 
Q_m(x)\;=\; 
\sum_{\ell=0}^{m-2}
\frac{x^2+\ell^2}{(x^2-\ell^2)^2}n_{z\ell}
$\hfill\null\\[-0.5em]\null\hfill$\displaystyle
+
\frac{x^2+(m-1)^2}{(x^2-(m-1)^2)^2}n_{z(m-1)}
+
\frac{x^2+(m+2)^2}{(x^2-(m+2)^2)^2}n_{z(m+2)}
$\hfill\null\\[-0.7em]\null\hfill$\displaystyle
+
\sum_{\ell=m+3}^{n_0}
\frac{x^2+\ell^2}{(x^2-\ell^2)^2}n_{z\ell}
$\null\qquad \null \\[1em]
With (\ref{27}),
since the function $s\mapsto \frac{s^2+1}{(s^2-1)^2}$ is increasing for $s<1$ and decreasing 
for $1<s$, and $x$ is in $]m,m+1[$, we get, for $m\geq 1$,
$$
\frac{Q_m(x)}{n_{z0}}\; \leq \; 
\int_0^{m-1}
\frac{m^2+\ell^2}{(m^2-\ell^2)^2} d\ell
+ \frac{2m^2-2m+1}{(2m-1)^2}
+ \frac{2m^2+5m+5}{(2m+3)^2}
+
\int_{m+2}^{\infty }
\frac{(m+1)^2+\ell^2}{((m+1)^2-\ell^2)^2} d\ell
$$
Using the identity
$$
\frac{s^2 + a^2}{(s^2-a^2)^2}
\;=\;
 -\frac{d}{ds}\left\{\frac{s}{s^2-a^2}\right\}
$$
we obtain
$$
\int_0^{m-1}
\frac{m^2+\ell^2}{(m^2-\ell^2)^2} d\ell\;=\; \frac{m-1}{m^2-(m-1)^2} \;=\; \frac{m-1}{2m-1} 
$$
$$
\int_{m+2}^{\infty }
\frac{(m+1)^2+\ell^2}{((m+1)^2-\ell^2)^2} d\ell
\;=\; \frac{m+2}{((m+2)^2-(m+1)^2)}\;=\;  \frac{m+2}{2m+3}
$$
and therefore
$$
\frac{Q_m(x)}{n_{z0}}\; \leq \; 
\frac{m-1}{2m-1} 
+ \frac{2m^2-2m+1}{(2m-1)^2}
+ \frac{2m^2+5m+5}{(2m+3)^2}
+ \frac{m+2}{2m+3}
\; \leq \; \frac{7}{2}
$$
We have also
$$
Q_0(x) \; = \;\sum_{\ell=2}^{n_0}
\frac{x^2+\ell^2}{(x^2-\ell^2)^2}n_{z\ell}
\; \leq \; n_{z2}\,  \sum_{\ell=2}^{n_0}
\frac{1+\ell^2}{(1-\ell^2)^2}
\; \leq \; n_{z2}\, \left[
\frac{5}{9}
+
\int_2^\infty \frac{1+\ell^2}{(1-\ell^2)^2} d\ell
\right]
\; \leq \; n_{z2}\, \frac{11}{9},
$$
implying $Q_0(x)\leq Q_m(x)$ with (\ref{27}).
This gives readily the following bound for $|G_{n_o}(x)| $
\begin{equation} 
\label{eq:boundGno}
|G_{n_o}(x)| \; \leq \; Q_m(x)\; \leq \; \frac{7}{2}\,  n_{z0} 
\end{equation}
As desired this bound is independent of $n_o,m$.
\par\vspace{1em}

With all this,
to obtain a bound for $|F_{n_o}(x)|$, we start by simplifying the expression of the denominator in 
  (\ref{eq:Fno}). For this, we define the following  change of variables 
$$
z =  \Psi(x)
:=
- \dfrac{1}{X- m^2} + \dfrac{\mathcursive{r}(m)}{(m+1)^2 - X},
$$
where $X = x^2$. Since $x\in(m, m+1)$, we decompose the
interval ${\mathfrak{X}} :=(m^2, (m+1)^2)$
as
$
{\mathfrak{X}} 
 = 
 \mathfrak{X}_- \cup \mathfrak{X}_0 \cup \mathfrak{X}_+
$
where 
$$
\begin{array}{rcl}
\mathfrak{X}_- &=& (m^2, X_a ]
\\
\mathfrak{X}_0 &=& [X_a , X_b
]
\\
\mathfrak{X}_+ &=& [ X_b, (m+1)^2)
\end{array}
\qquad
\begin{array}{rclcl}
X_a &= &(1-\lambda)m^2 + \lambda (m+1)^2 &= & m^2 + \lambda (2m+1)
\\[.5em]
X_b & = &  \lambda m^2 + (1-\lambda)(m+1)^2 & = &m^2 + (1-\lambda )(2m+1)
\end{array}
$$
with $\lambda \in (0,\tfrac12]$ to be selected with in mind the fact that when $\lambda $ is close to $0$,
$X_a$ is close to $m^2$ and $X_b$ close to 
$(m+1)^2$.
These sets are mapped, via $\Psi$, into the sets 
$$
\begin{array}{rcl}
\mathfrak{Z}_- &=& \left(-\infty,  \mathfrak{z}_a
\right]
\\[.5em]
\mathfrak{Z}_0 &=& \left[\mathfrak{z}_a, 
\ 
\mathfrak{z}_b
\right]
\\[.5em]
\mathfrak{Z}_+ &=& \left[ \mathfrak{z}_b, \
+\infty \right)
\end{array}
\qquad
\begin{array}{rcl}
\mathfrak{z}_a & = & 
- \dfrac{(1-\lambda){\mathfrak{m}_{p}} - \lambda \mathcursive{r}(m){\mathfrak{m}_{m}}}{\lambda(1-\lambda)
{\mathfrak{m}_{p}} {\mathfrak{m}_{m}}}
\\
\mathfrak{z}_b & = &
\left( \dfrac{\mathcursive{r}(m)}{1+\mathcursive{r}(m)} - \lambda
\right)\dfrac{1+\mathcursive{r}(m)}{\lambda(1-\lambda)}
\dfrac{1}{\mathfrak{m}_{m}}
\end{array}
$$
with the notation 
$$
{\mathfrak{m}_{p}} = [(m+1)^2+m^2]
\qquad
{\mathfrak{m}_{m}} = [(m+1)^2-m^2] = 2m+1
$$
Note that function $\Psi$ is injective
and its  inverse transformation is continuous and given by 
\begin{eqnarray}
\label{82}
X(z)
&= &\dfrac{{\mathfrak{m}_{p}} z
 - (1 + \mathcursive{r}(m))
+
\sqrt{
\Big(z {\mathfrak{m}_{m}} + (1-\mathcursive{r}(m))\Big)^2
+4 \mathcursive{r}(m)
}}{2z}
\qquad\mbox{if}\  z\neq 0
\\\nonumber
X(0)&=&
\dfrac{(m+1)^2+ m^2 \mathcursive{r}(m)}{1+\mathcursive{r}(m)}
\end{eqnarray}
%
About the choice of $\lambda $, we gather here all the constraints we shall impose later on.
\begin{enumerate}
\item We want $X(0)$ be in  ${\mathfrak{X}}_0$,
or equivalently
$\mathfrak{z}_a < 0 < \mathfrak{z}_b$. This motivates the following inequalities
\begin{equation}
\label{lambda0}
X_a =
m^2 + \lambda(2m+1)
\ < \ 
m^2 + (2m+1)
\dfrac{ 1}{1+\mathcursive{r}(m)}
\ 
< \
m^2 + (1-\lambda)(2m+1)
= X_b
\end{equation}
This
is always verified for any $m\geq1$ if we select 
$$
\lambda \ < \ 
\dfrac{ 1}{1+\mathcursive{r}(m)} \ <
\
1-\lambda
\  .
$$
With (\ref{boundR}), that is
\begin{equation}
\label{lambda1}
\lambda \ < \  \dfrac{1}{5} 
\ .
\end{equation}
\item
With \eqref{76}, the triangle inequality gives
$$
\left|S_m(x) + z \right| \geq |z|-\mathbf{S}
$$
We want to have
\begin{equation}
\label{77}
\left|S_m(x) + z \right| \geq |z|-\mathbf{S}\geq \frac{1}{\sqrt{2}} |z|
\qquad \forall z\in \mathfrak{Z}_-\cup \mathfrak{Z}_+
\end{equation}
We want also
\begin{equation}
\label{80} 
\frac{(1+\mathcursive{r}(m))^2}{2\mathcursive{r}(m){\mathfrak{m}_{m}}} 
\leq |z|
\qquad \forall z\in \mathfrak{Z}_-\cup \mathfrak{Z}_+
\end{equation}
This holds if
$$
\max\left\{
(2+\sqrt{2})\mathbf{S}
\,  ,\,  
\frac{25}{8} 
\right\}
\leq \min\{-\mathfrak{z}_a,\mathfrak{z}_b\}
$$
Let
\begin{equation}
\label{30}
\mathfrak{a} = \max\left\{
(2+\sqrt{2})\mathbf{S}
\,  ,\,  
\frac{25}{8} 
\right\}.
\end{equation}
With (\ref{29}), this number does not depend on $n_o$.
We have $-\mathfrak{z}_a\geq  \mathfrak{a}$ if
$$
\dfrac{(1-\lambda){\mathfrak{m}_{p}} - \lambda \mathcursive{r}(m){\mathfrak{m}_{m}}}{\lambda(1-\lambda)
{\mathfrak{m}_{p}} {\mathfrak{m}_{m}}} 
\geq \mathfrak{a}
$$
which gives
$$
\lambda^2 - \left( 1+
\dfrac{\mathfrak{m}_{p} +\mathcursive{r}(m){\mathfrak{m}_{m}} }{\mathfrak{m}_{p}\mathfrak{m}_{m}\mathfrak{a}}
\right)\lambda + \dfrac{1}{\mathfrak{m}_{m}\mathfrak{a}}
\geq 0
$$
and by recalling that $\lambda\in(0,\tfrac12]$, 
we obtain that the inequality is verified
for $\lambda\in(0,\lambda_a]$
when
\begin{equation*}
\label{lambda_a}
 \lambda_a = \dfrac{1}{2} \left( 1+ \dfrac{\mathfrak{m}_{p} +\mathcursive{r}(m){\mathfrak{m}_{m}} }{\mathfrak{m}_{p}\mathfrak{m}_{m}\mathfrak{a}}\right)
- \dfrac12\sqrt{
\left(1+\dfrac{\mathfrak{m}_{p} +\mathcursive{r}(m){\mathfrak{m}_{m}} }{\mathfrak{m}_{p}\mathfrak{m}_{m}\mathfrak{a}}
\right)^2 - \dfrac{4}{\mathfrak{m}_{m}\mathfrak{a}}}
\end{equation*}

Similarly, we have $\mathfrak{z}_b \geq \mathfrak{a}$ if$$
\left( \dfrac{\mathcursive{r}(m)}{1+\mathcursive{r}(m)} - \lambda
\right)\dfrac{1+\mathcursive{r}(m)}{\lambda(1-\lambda)}
\dfrac{1}{\mathfrak{m}_{m}}
\geq \mathfrak{a}
$$
which gives
$$
\lambda^2 - \left( 1+
\dfrac{1+\mathcursive{r}(m)}{\mathfrak{a} \mathfrak{m}_{m}}
\right) \lambda
+ \dfrac{\mathcursive{r}(m)}{\mathfrak{a} \mathfrak{m}_{m}}
\geq0
$$
which is satisfied for all $\lambda\in(0,\lambda_b]$
when
\begin{equation*}
\label{lambda_b}
 \lambda_b = \dfrac{1}{2} \left( 1+
\dfrac{1+\mathcursive{r}(m)}{\mathfrak{a}\mathfrak{m}_{m}}\right)
- \dfrac12\sqrt{
\left(1+
\dfrac{1+\mathcursive{r}(m)}{\mathfrak{m}_{m}\mathfrak{a}}
\right)^2 - \dfrac{\mathcursive{r}(m)}{\mathfrak{m}_{m}\mathfrak{a}}}
\end{equation*}

Now, we look for a lower bound for $\lambda_a,\lambda_b$. 
To this end, note that they
satisfy
\begin{align*}
\lambda_i & = \frac{1}{2}\left[T_{1i} - \sqrt{T_{1i}^2-T_{2i}} 
\right]
= \frac12\dfrac{T_{1i}^2 - \left(\sqrt{T_{1i}^2-T_{2i}}\right)^2}{T_{1i} + \sqrt{T_{1i}^2-T_{2i}} }
\\
 & = \frac12 \dfrac{T_{2i}}{T_{1i} + \sqrt{T_{1i}^2-T_{2i}}}
 \\
 & \geq    
\frac14 \dfrac{T_{2i}}{T_{1i}}
\end{align*}
where  
$$
\begin{array}{rclrcl}
T_{1a} & = & 
1+ \dfrac{\mathfrak{m}_{p} +\mathcursive{r}(m){\mathfrak{m}_{m}} }{\mathfrak{m}_{p}\mathfrak{m}_{m}\mathfrak{a}},
 \qquad &
T_{2a} & = &\dfrac{4}{\mathfrak{m}_m \mathfrak{a}},
\\
T_{1b} & = &1+\dfrac{1+\mathcursive{r}(m)}{\mathfrak{m}_m \mathfrak{a}},
&
T_{2b} & = & \dfrac{\mathcursive{r}(m)}{\mathfrak{m}_m \mathfrak{a}}.
\end{array}
$$
By further using 
$$
T_{1a} = 
1 + \dfrac{1}{\mathfrak{m}_{m}\mathfrak{a}} +
\dfrac{\mathcursive{r}(m)}{\mathfrak{m}_{p}\mathfrak{a}} 
\leq 1 + \dfrac{1}{\mathfrak{a}} + \dfrac{1}{\mathfrak{a}}
=     \dfrac{2+\mathfrak{a}}{\mathfrak{a}}
$$
and similarly, 
$$
T_{1b} \leq   \dfrac{2+\mathfrak{a}}{\mathfrak{a}}
$$
we finally obtain
\begin{align*}
\lambda_a  & \geq \dfrac{1}{2+\mathfrak{a}} \dfrac{1}{\mathfrak{m}_{m}},
\\
\lambda_b & \geq 
\frac14 
\dfrac{\mathcursive{r}(m)}{2+\mathfrak{a}} \dfrac{1}{\mathfrak{m}_{m}}
\geq  
\frac{1}{16}
\dfrac{1}{2+\mathfrak{a}} \dfrac{1}{2m+1}.
\end{align*}
Since $\mathfrak{a}$ does not depend on $m$ nor $n_o$, 
by selecting $\lambda $, depending on $m\geq 1$ and satisfying
\begin{equation}
\label{31}
\frac{1}{5} \geq  \lambda  \geq \dfrac{\varpi}{m}, 
\qquad 
\varpi := \dfrac{1}{48(\mathfrak{a}+2)}
\end{equation}
we have always satisfied
$\lambda \in (0, \min\{\lambda_a,\lambda_b,\tfrac{1}{5}\}]$.
Note that, according to (\ref{30}) and (\ref{29}), $\varpi$ does not depend on $m$ or $n_o$.

\item
We want to have
\begin{equation}
\label{79}
\sqrt{(z{\mathfrak{m}_{m}} + (1 -\mathcursive{r}(m)))^2 + 4 \mathcursive{r}(m) }  \geq  (1 +\mathcursive{r}(m))
\qquad \forall z\in \mathfrak{Z}_-\cup \mathfrak{Z}_+
\end{equation}
This holds if $z$ is not in
$\left]-2\frac{1-\mathcursive{r}(m)}{\mathfrak{m}_{m}},0\right[$ and therefore if we have
$$
- \dfrac{(1-\lambda){\mathfrak{m}_{p}} - \lambda \mathcursive{r}(m){\mathfrak{m}_{m}}}{\lambda(1-\lambda)
{\mathfrak{m}_{p}} {\mathfrak{m}_{m}}}=\mathfrak{z}_a \leq -2\frac{1-\mathcursive{r}(m)}{\mathfrak{m}_{m}}
$$
This inequality reduces to
$$
2[1-\mathcursive{r}(m)] {\mathfrak{m}_{p}}\lambda ^2
-
\left({\mathfrak{m}_{p}}
+
\mathcursive{r}(m){\mathfrak{m}_{m}}
+
2[1-\mathcursive{r}(m)] {\mathfrak{m}_{p}}
\right)
 \lambda 
+
{\mathfrak{m}_{p}} 
\geq 0
$$
It is satisfied for all 
$\lambda\in(0,\min\{\tfrac12,\lambda_c\}]$
with 
\begin{equation*}
\label{lambda_c}
 \ \lambda_c = \dfrac{1}{2} \left( 1+
\dfrac{\mathfrak{m}_{p} +\mathcursive{r}(m){\mathfrak{m}_{m}} }{2[1-\mathcursive{r}(m)]\mathfrak{m}_{p}}
\right)
- \dfrac12\sqrt{
\left(1+ \dfrac{\mathfrak{m}_{p} +\mathcursive{r}(m){\mathfrak{m}_{m}} }{2[1-\mathcursive{r}(m)]\mathfrak{m}_{p}}
\right)^2 - \dfrac{2}{1-\mathcursive{r}(m)}}.
\end{equation*}
By following similar computations to those made for 
$\lambda_a, \lambda_b$, previous expression gives also  
$$
\lambda_c  \geq 
\frac{1}{2[1-\mathcursive{r}(m)]}
\frac{1}{ 1+ \dfrac{\mathfrak{m}_{p} +\mathcursive{r}(m){\mathfrak{m}_{m}} }{
2
[1-\mathcursive{r}(m)]\mathfrak{m}_{p}}}
\geq 
\frac{1}{ 3 -\mathcursive{r}(m)[2-\frac{\mathfrak{m}_{m}}{\mathfrak{m}_{p}}]}
\geq 
 \dfrac{1}{3}.
$$
\end{enumerate}
Collecting all these constraints, we conclude 
that $\lambda$ can be selected in the interval
\begin{equation}
\label{eq:ineq_lambda}
\dfrac{\varpi}{m} 
=
\dfrac{1}{48(\mathfrak{a}+2)}\frac{1}{m}
\leq \lambda \leq \frac{1}{5}.
\end{equation}
After these preliminaries, we are ready to derive bounds for $|F_{n_o}(x)|$.
For the case where
$x$ is in $\mathfrak{X}_0$ or equivalently $z$ in $\mathfrak{Z}_0$, we have
$$
(x^2-m^2)^2\geq (X_a-m^2)^2
\quad ,\qquad 
((m+1)^2-x^2)^2 \geq ((m+1)^2-X_b)^2
\  .
$$
This gives
\begin{eqnarray}
\nonumber
|F_{n_o}(x)|
&\leq &
n_{zm}
\left[\dfrac{x^2+m^2}{(x^2-m^2)^2}
+\dfrac{x^2+(m+1)^2}{((m+1)^2-x^2)^2}\mathcursive{r}(m)
\right]
\\\nonumber
 & \leq  &
2n_{zm} (m+1)^2
\left[\dfrac{1}{(X_a-m^2)^2}
+\dfrac{1}{((m+1)^2-X_b)^2}\mathcursive{r}(m)
\right]
\\\nonumber
& \leq &
2n_{zm} (m+1)^2 
\left[
\dfrac{1}{\lambda^2 (2m+1)^2}
+ \dfrac{1}{(1-\lambda)^2 (2m+1)^2}
\right]
\\ \notag
&\leq & 
2n_{z0}   \left(\dfrac{m+1 }{2m+1}\right)^2\left[\frac{1}{\lambda ^2}+\frac{1}{(1-\lambda )^2}\right]
\end{eqnarray}
So, with (\ref{eq:ineq_lambda}), we get
\begin{equation}
\label{boundFno0}
|F_{n_o}(x)| \; \leq \; \dfrac {4n_{z0}}{\varpi^2}  
m^2 
\qquad
\forall \ x\in \mathfrak{X}_0
\  .
\end{equation}
For the case where $x$ is in $ \mathfrak{X}_- \cup \mathfrak{X}_+$, or equivalently $z$ is in 
$\mathfrak{Z}_-\cup \mathfrak{Z}_+$, we have, with (\ref{82}),
$$
x^2 -m^2  \  =   \  \dfrac{z {\mathfrak{m}_{m}} 
 -c(z)}{2z}\ ,
\qquad
(m+1)^2 -x^2  \ = \  \dfrac{z {\mathfrak{m}_{m}} +
c(z)}{2z} \ ,
$$
with the notation 
$$
c(z) = 1 +\mathcursive{r}(m)-\sqrt{(z {\mathfrak{m}_{m}} + (1 -\mathcursive{r}(m)))^2 + 4 \mathcursive{r}(m) } 
\  .
$$
where $c(z)$ is non positive according to (\ref{79}).
So, by using $\mathcursive{r}(m)\leq 1$ and
$m < x < m+1$,
and ${\mathfrak{m}_{m}} = 2m+1$,
we obtain
\begin{eqnarray*}
|{F}_{n_o}(x)| & \leq  &
 \dfrac{
8n_{zm}(m+1)^2 
\left[\dfrac{z^2}{(z {\mathfrak{m}_{m}}  - c(z))^2}
+\dfrac{z^2}{(z {\mathfrak{m}_{m}}  + c(z))^2}
\right]
}{1+\mu^2 m^2 n_{zm}^2
\left(S_m(x) + z
\right)^2 
}
\end{eqnarray*}
But (\ref{77}) gives
\begin{eqnarray*}
\dfrac{
n_{zm}(m+1)^2 z^2
}{1+\mu^2 m^2 n_{zm}^2
\left(
S_m(x)
 + z
\right)^2 
}&\leq&
n_{z0}\qquad \mbox{if}\  m=0
\\
 &\leq &\dfrac{
2n_{zm}(m+1)^2 z^2
}{
2+\mu^2 m^2 n_{zm}^2z^2 
}
\qquad \forall z\in \mathfrak{Z}_-\cup \mathfrak{Z}_+
\qquad \mbox{if}\  m\geq 1
\\
&\leq &
\frac{1}{n_{zm}}\dfrac{(m+1)^2}{m^2}\frac{2}{\mu ^2}
\qquad \forall z\in \mathfrak{Z}_-\cup \mathfrak{Z}_+
\qquad \mbox{if}\  m\geq 1
\end{eqnarray*}
This yields
\begin{eqnarray*}
|{F}_{n_o}(x)| 
& \leq &
\frac{1}{n_{zm}}\frac{(m+1)^2}{m^2}\frac{32}{\mu ^2}
\dfrac{(z {\mathfrak{m}_{m}})^2  + c(z)^2}
{((z {\mathfrak{m}_{m}})^2  - c(z)^2)^2}
\qquad \forall z\in \mathfrak{Z}_-\cup \mathfrak{Z}_+
\\
|{F}_{n_o}(x)| & \leq  &
\frac{1}{n_{zm}}\frac{32}{\mu ^2}\frac{1}{z^2 {\mathfrak{m}_{m}}^2}
\dfrac{1+ [\tfrac{c(z)}{z {\mathfrak{m}_{m}}}]^2}
{(1  - [\tfrac{c(z)}{z {\mathfrak{m}_{m}}}]^2)^2}
\qquad \forall z\in \mathfrak{Z}_-\cup \mathfrak{Z}_+
\end{eqnarray*}
Since the function $\chi^2 \mapsto \frac{1+\chi^2}{(1-\chi^2)^2}$ is non decreasing on 
$[0,1 [$ we go on by looking for an upperbund for $\left|\tfrac{c(z)}{z {\mathfrak{m}_{m}}}\right| $.
With (\ref{79}), we have
$$
\left|\tfrac{c(z)}{z {\mathfrak{m}_{m}}}\right| \;=\; 
\frac{
\sqrt{(z{\mathfrak{m}_{m}} + (1 -\mathcursive{r}(m)))^2 + 4 \mathcursive{r}(m) }  - (1 +\mathcursive{r}(m))
}{
|z|{\mathfrak{m}_{m}} 
}
$$
Then, the inequality
$$
\sqrt{1+a}\leq 1+\frac{a}{2}\qquad \forall a\geq -1
$$
allows us to write
$$
\left|\tfrac{c(z)}{z {\mathfrak{m}_{m}}}\right| 
\; \leq \; 
1+\frac{1-\mathcursive{r}(m)}{z{\mathfrak{m}_{m}}}+\frac{(1+\mathcursive{r}(m))^2}{2z^2{\mathfrak{m}_{m}}^2}
- \frac{1+\mathcursive{r}(m)}{|z|{\mathfrak{m}_{m}}}\,.
$$
But on one hand, we have
$$
\renewcommand{\arraystretch}{2}
\begin{array}{rclcl}
\displaystyle 
\frac{1-\mathcursive{r}(m)}{z{\mathfrak{m}_{m}}}
- \frac{1+\mathcursive{r}(m)}{|z|{\mathfrak{m}_{m}}} &=&
\displaystyle -\frac{2}{|z|{\mathfrak{m}_{m}}} 
\; \leq \; -\frac{2\mathcursive{r}(m)}{|z|{\mathfrak{m}_{m}}} 
&\mbox{if}&\displaystyle  z < 0
\  ,
\\
&=&\displaystyle 
-\frac{2\mathcursive{r}(m)}{|z|{\mathfrak{m}_{m}}} 
&\mbox{if}&\displaystyle  0 < z\ .
\end{array}
$$
On another hand, we have~:
$$
\frac{(1+\mathcursive{r}(m))^2}{2z^2{\mathfrak{m}_{m}}^2}  \leq  
\frac{\mathcursive{r}(m)}{|z|{\mathfrak{m}_{m}}}
\qquad \mbox{if}\quad 
\frac{(1+\mathcursive{r}(m))^2}{2\mathcursive{r}(m){\mathfrak{m}_{m}}}  \leq  
|z| \ .
$$
This yields
\begin{eqnarray*}
\left|\tfrac{c(z)}{z {\mathfrak{m}_{m}}}\right| 
\; \leq \;  1- \frac{\mathcursive{r}(m)}{|z|{\mathfrak{m}_{m}}}
\qquad \mbox{if}\quad  \frac{(1+\mathcursive{r}(m))^2}{2\mathcursive{r}(m){\mathfrak{m}_{m}}}  \leq  |z| \ .
\end{eqnarray*}
So, with (\ref{80}), we have obtained
$$
|{F}_{n_o}(x)| \;  \leq  \; 
\frac{1}{n_{zm}}\frac{32}{\mu ^2}\frac{1}{z^2 {\mathfrak{m}_{m}}^2}
\dfrac{1+ [1- \frac{\mathcursive{r}(m)}{|z|{\mathfrak{m}_{m}}}]^2}
{(1  - [1- \frac{\mathcursive{r}(m)}{|z|{\mathfrak{m}_{m}}}]^2)^2}
\;=\; 
\frac{1}{n_{zm}}\frac{32}{\mu ^2}\frac{1}{\mathcursive{r}(m)^2}
\dfrac{1+ [1- \frac{\mathcursive{r}(m)}{|z|{\mathfrak{m}_{m}}}]^2}
{
[2- \frac{\mathcursive{r}(m)}{|z|{\mathfrak{m}_{m}}}]^2
}
\qquad \forall z\in \mathfrak{Z}_-\cup \mathfrak{Z}_+ \ .
$$
We have
$$
\dfrac{1+ [1- \chi]^2}
{
[2- \chi]^2
}\leq 1
\qquad \forall \chi \leq 1
$$
and, because of (\ref{boundR}),
$$
\frac{\mathcursive{r}(m)}{|z|{\mathfrak{m}_{m}}}\leq \frac{(1+\mathcursive{r}(m))^2}{2\mathcursive{r}(m)|z|{\mathfrak{m}_{m}}}
\leq 1
\qquad \forall z: \frac{(1+\mathcursive{r}(m))^2}{2\mathcursive{r}(m){\mathfrak{m}_{m}}}  \leq  |z| \ .
$$
We conclude that (\ref{80}) gives finally
\begin{equation}
\label{boundFno}
|{F}_{n_o}(x)| \leq 
\frac{512}{\mu ^2}\frac{1}{n_{zm}}
\qquad \forall z\in \mathfrak{Z}_-\cup \mathfrak{Z}_+ \ .
\end{equation}
Hence, by putting together
\eqref{boundFno0} and
(\ref{boundFno}) we  obtain
\begin{equation}
\label{boundFno2}
|F_{n_o}(x)|
\ \leq  \ 
\max\left\{ \dfrac {4n_{z0}}{\varpi^2}  m^2 
\,  ,\,  
\frac{512}{\mu ^2}\frac{1}{n_{zm}}
\right\}
\qquad
\forall \, x \in (m,m+1)\ .
\end{equation}
Finally, by using
(\ref{23}),
\eqref{eq:zeta_pk_simplified3},
\eqref{eq:boundGno},
and
\eqref{boundFno2}, we obtain,
for each $x$ in $(m,m+1)$ with $m$ in $\{0,\dots,n_o-1\}$,
\begin{eqnarray}
\notag
\mathcursive{T}(x)
 & \leq &  
\left(\frac{7}{2} n_{z0}
+
\dfrac {4n_{z0}}{\varpi^2}  m^2 
+ 
\frac{512}{\mu ^2}\frac{1}{n_{zm}}
\right)  \\ \notag
 & \leq &   
\left(\frac{7}{2} n_{z0}
+
\left[\dfrac {4n_{z0}}{\varpi^2}  
+ 
\frac{512}{\mu ^2}\frac{1}{n_{z1}}\right]m^2
\right) \  ,
\\
\label{boundLemma1}
 & \leq &   
\left(\frac{7}{2} n_{z0}
+
\left[\dfrac {4n_{z0}}{\varpi^2}  
+ 
\frac{512}{\mu ^2}\frac{1}{n_{z1}}\right]x^2
\right) \  .
\end{eqnarray}

\paragraph*{Case b) $n_o< x < n_o+1$.}~
This case is a variant of the preceding one, with a (fictitious) singularity on left side only.
In particular, the expression 
\eqref{eq:zeta_pk_simplified}
simplifies as 
\begin{eqnarray*}
\mathcursive{T}(x) 
& =  &
G_{n_o}(x) + F_{n_o}(x) 
\end{eqnarray*}
where
$$
G_{n_o}(x) := \dfrac{\displaystyle
\sum_{\ell=0}^{n_o-1} 
\frac{\ell^2 + x^2}{\left(x^2-\ell^2\right)^2}n_{z\ell}}
{1 + \displaystyle\mu^2 x^2 \left(\sum_{\ell=0}^{n_o} \frac{1}
{x^2 -\ell^2}n_{z\ell}\right)^2} \ , 
\qquad
F_{n_o}(x): =  \dfrac{\dfrac{n_o^2 + x^2}{(x^2-n_o^2)^2} n_{zn_0}}
 {1 + \displaystyle\mu^2 x^2 \left(\sum_{\ell=0}^{n_o} \frac{1}
{x^2  -\ell ^2}n_{z\ell}\right)^2}\ .
$$
Concerning $G_{n_o}$,
by recalling that
$ n_{z0} >  n_{z\ell}$,for any $\ell\geq0$, and that
the function $\chi \mapsto \frac{\chi +1}{(\chi -1)^2}$ is increasing on $[0,1[$ and  decreasing on 
$]1,+\infty [$,  we compute
\begin{eqnarray*}
|G_{n_o}(x)| &\leq &
\sum_{\ell=0}^{n_o-1} 
\frac{\ell^2 + x^2}{\left(\ell^2-x^2\right)^2}n_{z\ell}
\\ &\leq &
\left[\sum_{\ell=0}^{n_o-2} 
\frac{\ell^2 + n_o^2}{(\ell ^2 - n_o^2 )^2}
+
\frac{(n_o-1)^2 + n_o^2}{((n_o-1)^2-n_o^2)^2}
\right]
n_{z0}
\\  &\leq &
\left[\int_0^{n_o-1}
\frac{\ell^2+n_o^2}{(\ell^2-n_o^2)^2} d\ell
+ 
\frac{(n_o-1)^2+n_o^2}{(2n_o-1)^2}
\right]   n_{z0}
\\ &\leq & 
\left[
\frac{n_o-1}{2n_o-1} 
+ 
\frac{(n_o-1)^2+n_o^2}{(2n_o-1)^2}
\right]   n_{z0}
\; \leq \;  n_{z0} \ .
\end{eqnarray*}
Then, concerning
$F_{n_o}$, we note that,
since $x>n_o$, all the terms in the sum of the denominator
are positive. As a consequence, we can write
$$
F_{n_o}(x)
\; \leq \; 
\dfrac{\dfrac{n_o^2 + x^2}{(x^2-n_o^2)^2} n_{zn_0}
 }{1 + \displaystyle\mu^2   \frac{x^2}
{(x^2  -n_o ^2)^2}n_{zn_0}^2}
\; \leq \; \dfrac{\dfrac{ 2x^2}{(x^2-n_o^2)^2} n_{zn_0}}
 {\displaystyle\mu^2   \frac{x^2}
{(x^2  -n_o ^2)^2}n_{zn_0}^2}
\; \leq \; 
 \dfrac{2}{\mu^2n_{zn_0}}
 \qquad \forall x\in ]n_o,n_o+1[.
$$
By combining all the bounds and (\ref{23}), we
finally  obtain, for each $x$ in $]n_o,n_o+1[$,
\begin{eqnarray}
\nonumber
\mathcursive{T}(x) 
&\leq &
\left(n_{z0}
+ 
\dfrac{2}{\mu^2 n_{zn_0}}
\right)
\  ,
\\\notag
&\leq & 
\left(n_{z0}
+ 
\dfrac{2n_0^2}{\mu^2 n_{z1}}
\right)
\\
\label{boundLemma2}
&\leq &
\left(n_{z0}
+ 
\dfrac{2}{\mu^2 n_{z1}} x^2
\right) \ .
\end{eqnarray}

\paragraph*{Case c) $n_o+1 \leq x$.}~ 
Since the function $\chi \mapsto \frac{\chi +1}{(\chi -1)^2}$ is increasing on $[0,1[$ and  decreasing on 
$]1,+\infty [$,
from 
\eqref{eq:zeta_pk_simplified},  we obtain,
 for each $x$ larger than $n_o+1$,
\begin{eqnarray}
\notag
\mathcursive{T}(x) 
&=&
\dfrac{\displaystyle\sum_{\ell=0}^{n_o-1} 
\frac{\ell^2 + x^2}
{\left(x^2-\ell^2\right)^2 }n_{z\ell}
+
\frac{n_o^2 + x^2}
{\left(x^2-n_o^2\right)^2 }n_{z\ell}
}
{1 + \displaystyle\mu^2 x^2 \left(\sum_{\ell=0}^{n_o} \frac{1}
{x^2-\ell^2  }n_{z\ell}\right)^2}
\; \leq \;
\left[
\int_0^{n_o}
\dfrac{x^2+\ell^2}{(x^2-\ell^2)^2}
d\ell 
+
\dfrac{(n_o+1)^2+n_o^2}{((n_o+1)^2-n_o^2)^2}
\right]
n_{z0}
\\
\notag
&\leq  &n_{z0}
\left[\left.
\dfrac{\ell}{x^2- \ell^2}
\right|_0^{n_o} 
+
1
\right] n_{z0}
\; \leq \;
\left[\dfrac{n_o}{x^2- n_o^2}+1\right] n_{z0}
\; \leq \;
\left[\dfrac{n_o}{(n_o+1)^2- n_o^2}+1\right] n_{z0}
\\
& \leq &  n_{z0}
\dfrac{3n_o+1}{2n_o+1}
\ \leq \  \tfrac32  n_{z0}.
\label{boundLemma3}
\end{eqnarray}

 \paragraph*{End of the proof.}~
Finally, we can now conclude the proof of the lemma.
For this, we combine 
bounds 
(\ref{eq:bound_zetaNzeta_exact}),
\eqref{boundLemma1}, 
\eqref{boundLemma2} and \eqref{boundLemma3}.
\begin{itemize}
\item[$\bullet$] For each $x$ integer in 
$\{0,\ldots,n_o\}$, inequality (\ref{eq:bound_zetaNzeta_exact})  gives
$$
\mathcursive{T}(x)
\; \leq  \; \frac{2}{\mu ^2 n_{z1}}\,  x^2
\  .
$$
\item[$\bullet$]  For each $x$ in $(m,m+1)$ with $m$ in $\{0,\dots,n_o-1\}$, 
inequality  \eqref{boundLemma1} gives
$$
\mathcursive{T}(x)
\; \leq \; \left(\frac{7}{2} n_{z0}
+
\left[\dfrac {4n_{z0}}{\varpi^2}  
+ 
\frac{512}{\mu ^2}\frac{1}{n_{z1}}\right]x^2
\right)
\  ,
$$
with $\varpi$ defined in (\ref{31}) does not depend on $x$ nor $n_o$.
\item[$\bullet$]  
For each $x$ in $]n_o,n_o+1[$,  inequality \eqref{boundLemma2} 
$$
\mathcursive{T}(x)
\; \leq \;
n_{z0}
+ 
\dfrac{2}{\mu^2 n_{z1}} x^2
\  .
$$
\item[$\bullet$]  For each $x$ larger than $n_o+1$, 
inequality  \eqref{boundLemma3} gives
$$
\mathcursive{T}(x)
\; \leq \; \tfrac32  n_{z0} 
\  .
$$
\end{itemize}
So we have established \eqref{32} with
$$
\kappa _0\;=\; \frac{7}{2} n_{z0}
\quad ,\qquad 
\kappa _1\;=\; 
\dfrac {4n_{z0}}{\varpi^2}  
+ 
\frac{512}{\mu ^2}\frac{1}{n_{z1}}
$$
which do not depend on $n_0$, thus concluding the proof.\qed
\end{proof}

\end{appendix}

\bibliographystyle{spmpsci}      
\bibliography{biblio}   

%


%
%

\end{document}